\documentclass[12pt]{extarticle}
\usepackage{stylefile}
\usepackage[margin=1in]{geometry}
\usepackage{natbib}

\usepackage{amssymb}

\begin{document}
\begin{titlepage}
\title{Discrete Choice under Risk with Limited Consideration\thanks{We are grateful to Liran Einav, three anonymous referees, Abi Adams, Jose Apesteguia, Miguel Ballester, Arthur Lewbel, Chuck Manski, and Jack Porter for useful comments and constructive criticism. For comments and suggestions we thank the participants to the 2017 Barcelona GSE Summer Forum on Stochastic Choice, the 2018 Cornell Conference ``Identification and Inference in Limited Attention Models'', the 2018 Penn State-Cornell Conference on Econometrics and IO, the 2018 GNYMA Conference, the 2019 ASSA meetings, the IFS 2019 ``Consumer Behaviour: New Models, New Methods'' Conference, and to seminars at Stanford, Berkeley, UCL, Wisconsin, Bocconi, and Duke. Part of this research was carried out while Barseghyan and Molinari were on sabbatical leave at the Department of Economics at Duke University, whose hospitality is gratefully acknowledged.  We gratefully acknowledge support from National Science Foundation grant SES-1824448 and from the Institute for Social Sciences at Cornell University.}}
\author{\normalsize{Levon Barseghyan}\thanks{Department of Economics, Cornell University, lb247@cornell.edu.} \and
            \normalsize{Francesca Molinari}\thanks{Department of Economics, Cornell University, fm72@cornell.edu.} \and
            \normalsize{Matthew Thirkettle}\thanks{Department of Economics, Rice University, mt77@rice.edu.}}
\date{\normalsize{\today}}
\maketitle
\begin{abstract}
\noindent
This paper is concerned with learning decision makers'
preferences using data on observed choices from a finite set of risky alternatives. We propose a discrete choice model with unobserved heterogeneity in consideration sets and
in standard risk aversion. We obtain sufficient conditions for the model's semi-nonparametric point identification,
including in cases where consideration depends on preferences and on some of the exogenous variables. Our method yields an estimator that is easy to compute and is applicable in markets with large choice sets. We illustrate its properties using a dataset on property insurance purchases.\\


\end{abstract}
\setcounter{page}{0}
\thispagestyle{empty}
\end{titlepage}
\pagebreak \newpage

\setcounter{page}{1}
\setcounter{section}{0}
\section{Introduction}
This paper is concerned with learning decision makers' (DMs) preferences using data on observed choices from a finite set of risky alternatives with monetary outcomes. The prevailing empirical approach to study this problem merges expected utility theory (EUT) models with econometric methods for discrete choice analysis.  Standard EUT assumes that the DM evaluates all available alternatives and chooses the one yielding the highest expected utility. The DM's risk aversion is determined by the concavity of her Bernoulli utility function. The set of all alternatives  -- the \emph{choice set} -- is assumed to be observable by the researcher.

We depart from this standard approach by proposing a discrete choice model with unobserved heterogeneity in preferences and unobserved heterogeneity in consideration sets. Specifically, preferences satisfy the classic Single Crossing Property (SCP) of \citet{mirrlees1971exploration} and \citet{spence1974market}, central to important studies of decision making under risk.\footnote{E.g., \cite{apesteguia2017single,chiappori2019aggregate}. While our focus is on decision making under risk, the SCP property is satisfied in many contexts, ranging from single agent models with goods that can be unambiguously ordered based on quality, to multiple agents models (e.g., \cite{athey2001single}).} That is, the preference order of any two alternatives switches only at one value of the preference parameter.\footnote{The EUT framework satisfies the SCP, which requires that if a DM with a certain degree of risk aversion prefers a safer lottery to a riskier one, then all DMs with higher risk aversion also prefer the safer lottery.} Given her \emph{unobserved} preference parameter, each DM evaluates only the alternatives in her \emph{unobserved} consideration set, which is a subset of the choice set.

Our first contribution is to provide a general framework for point identification of these models. Our analysis relies on two types of observed data variation. In the first case, we assume that the data include a single (common) excluded regressor affecting the utility of each alternative. In the second case, we assume that each alternative has its own excluded regressor. In both cases, the excluded regressor(s) is independent of unobserved preference heterogeneity. When the excluded regressor(s) also has large support it becomes a ``special regressor'' \citep{Lewbel00, lewbel2012overview}. For reasons we explain, the case of the single common excluded regressor is the most demanding from an identification standpoint. Nonetheless, under classic conditions for identification of full-consideration discrete choice models \citep[see, e.g., ][]{Lewbel00,matzkin07} and the SCP, we obtain  semi-nonparametric identification of the preference distribution given basically any consideration set formation mechanism (henceforth, consideration mechanism).\footnote{The identification results are semi-nonparametric because we specify the utility function up to a DM-specific preference parameter. We establish nonparametric identification of the distribution of the latter.}   We also prove identification of the consideration mechanism for the widely used Alternative-specific Random Consideration (ARC) model of \cite{manski1977structure} and \cite{man:mar14}. The identification argument is constructive and applicable beyond the ARC model. We establish identification results for preferences that do not require large support of the excluded regressor(s).  We also show that identification of both preferences and the consideration mechanism is attainable when consideration depends on preferences. In particular, we introduce (i) binary consideration types, and (ii) proportionally shifting consideration, both of which can capture the notion that the DM's attention probabilistically shifts from riskier to safer alternatives as her risk aversion increases. In these cases, identification requires that the distribution of the preference parameter admits a continuous density function.

We can significantly expand our results with alternative-specific excluded regressors. First, we can allow for essentially unrestricted dependence of consideration on preferences without assuming that the excluded regressors have large support. Second, we show that consideration can depend both on preferences and on some excluded regressors. We show this for two cases. In the first case, there is one alternative (the default) that is always considered.  The probability of considering other alternatives can depend on the default-specific excluded regressor. This is a generalization of the models in \citet{heiss2016inattention,ho2017impact, Abaluck2019}, where the consideration mechanism only allows for the possibility that either the default or the entire choice set is considered.  We, however, allow for each subset of the choice set containing the default to have its own probability of being drawn and this probability can vary with the DM's preferences. In the second case, we allow the consideration of each alternative to depend on its own excluded regressor, but not on the regressors of other alternatives \citep{Goeree2008,Abaluck2019,kawaguchi2019designing}. In addition, consideration may depend on preferences -- a feature unique to our paper. 

Our second contribution is to provide a simple method to compute our likelihood-based estimator. Its computational complexity grows polynomially in the number of parameters governing the consideration mechanism. Because the SCP generates a natural ordering of alternatives akin to vertical product differentiation, our method does not require enumerating all possible subsets of the choice set. If it did, the computational complexity would grow exponentially with the size of the choice set. Moreover, we compute the utility of each alternative only once for a given value of the preference parameter, gaining enormous computational advantage similar to that of importance-sampling methods.

Our third contribution is to elucidate the applicability and the advantages of our framework over the standard application of full consideration random utility models (RUMs) with additively separable unobserved heterogeneity (e.g., Mixed Logit). First, our model can generate zero shares for non-dominated alternatives. Second, the model has no difficulty explaining relatively large shares of dominated alternatives. Third, in markets with many choice domains, our model can match not only the marginal but also the joint distribution of choices across domains. Forth, our framework is immune to an important criticism by \cite{ape:bal18} against using standard RUMs to study decision making under risk. As these authors note, combining standard EUT with additive noise results in non-monotonicity of choice probabilities in the risk preferences, a clearly undesirable feature.

Random preference models like the ones we consider are random utility models as envisioned by \cite{Mcfadden1974}  \citep[for a textbook treatment see] []{Manski2007book}. We show that our random preference models can be written as RUMs with unobserved heterogeneity in risk aversion and with an additive error that has a discrete distribution with support $\{-\infty,0\}$.
Then, it is natural to draw parallels with the Mixed (random coefficient) Logit model \citep[e.g.,][]{McFaddenTrain2000}. In our setting, the Mixed Logit boils down to assuming that, given the DM's risk aversion, her evaluation of an alternative equals its expected utility summed with an unobserved heterogeneity term capturing the DM's idiosyncratic taste for unobserved characteristics of that alternative. However, in some markets it is hard to envision such characteristics.\footnote{Many insurance contracts are identical in \emph{all} aspects \emph{except} for the coverage level and price, e.g., employer provided health insurance, auto, or home insurance offered by a single company. In other contexts, unobservable characteristics may affect choice mostly via consideration -- as we model -- rather than via ``additive noise''. E.g., a DM may only consider those supplemental prescription drug plans that cover specific medications.} We show that limited consideration models and the Mixed Logit generate several contrasting implications. First, the Mixed Logit generally implies that each alternative has a positive probability of being chosen, while a limited consideration model can generate zero shares by setting the consideration probability of a given alternative to zero. Second, the Mixed Logit satisfies a \emph{Generalized Dominance Property} that we derive: if for any degree of risk aversion alternative $j$ has lower expected utility than either alternative $k$ or $l$, then the probability of choosing $j$ must be no larger than the probability of choosing $k$ or $l$. Limited consideration models do not necessarily abide Generalized Dominance. Third, in limited consideration models choice probabilities depend on the ordinal expected utility rankings of the alternatives, while in the Mixed Logit it depends on the cardinal ranking. This difference implies that choice probabilities may be monotone in risk preferences in the limited consideration models we propose, while in the Mixed Logit they are not \citep{ape:bal18}.

We illustrate our method in a study of households' deductible choices across three lines of insurance: auto collision, auto comprehensive, and home (all perils). We aim to estimate the distribution of risk preferences and the consideration parameters and to assess the resulting fit of the models. We find that the $\modA$ model does a remarkable job at matching the distribution of observed choices, and because of its aforementioned properties, outperforms the Mixed Logit. Under the ARC model, we find that although households are on average strongly risk averse, they consider lower coverages more often than higher coverages. We also find support for proportionally shifting consideration. In particular, risk-neutral DMs consider each of the safer alternatives $15\%$ ($11\%$) less often than do extremely risk averse DMs (DMs with median risk aversion). 

The rest of the paper is organized as follows. We describe the model of DMs' preferences in Section \ref{sec:preferences}, and study identification in Section \ref{sec:idchallenge}. In Section \ref{S:LikeTract} we describe the computational advantages of our approach. Section \ref{sec:modprop} compares our model to the Mixed Logit. Section \ref{sec:appplication} presents our empirical application. Section \ref{sec:discussion} contextualizes our contribution relative to the extant literature and offers concluding remarks.

\section{Preferences}\label{sec:preferences}
\subsection{Decision Making under Risk in a Market Setting: An Example}\label{sec:mainexample}
Consider as an example the following insurance market, which mimics the setting of our empirical application. There is an underlying risk of a loss that occurs with probability $\mu$ that may vary across DMs. A finite number of alternatives are available to insure against this loss. Conditional on risk type, i.e., given $\mu$, each alternative $j \in \Dc \equiv \{1,\dots,D\}$ is fully characterized by the pair $(d_j,p_j)$. The first element is the insurance deductible, which is the DM's out of pocket expense in the case a loss occurs. Deductibles are decreasing with index $j$, and all deductibles are less than the lowest realization of the loss. The second element is the price (insurance premium), which also varies across DMs. For each DM there is a baseline price $\bar{p}$ that determines prices for all alternatives faced by the DM according to the multiplication rule $p_{j}=g_j\cdot \bar{p}+\delta$. Lower deductibles provide more coverage and cost more, so $g_j$ is increasing with $j$. Both $g_j$ and $\delta$ are invariant across DMs. The lotteries that DMs face are $L_{j}(x)\equiv \left( -p_{j},1-\mu ;-p_{j}-d_j,\mu \right)$, where  $x\equiv\bar{p}$. DMs are expected utility maximizers. Given initial wealth $w$, the expected utility of deductible lottery $L_{j}(x)$ is
\begin{equation*}
U_{\pparam}(L_{j}(x))=\left( 1-\mu \right) u_{\pparam}\left( w-p_{j}\right) +\mu u_{\pparam}\left(w-p_{j}-d_j\right) ,
\end{equation*}
where $u_{\pparam}(\cdot)$ is a Bernoulli utility function defined over final wealth states. We assume that $u_{\pparam}(\cdot)$ belongs to a family of utility functions that are fully characterized by a scalar $\pparam$ (e.g. Constant Absolute Risk Aversion (CARA), Constant Relative Risk Aversion (CRRA), or Negligible Third Derivative (NTD)), which varies across DMs.\footnote{Under CRRA, it is implied that DMs' initial wealth is known to the researcher. NTD utility is defined in \cite{Cohen2007} and in \cite{Barseghyan2013}.}

Given the risk type, the relationship between risk aversion and prices is standard. At sufficiently high $\bar{p}$, less coverage is always preferred to more coverage for all $\pparam$ on the support: $U_{\pparam}(L_{1}(\covx))>U_{\pparam}(L_{2}(\covx)) >\dots> U_{\pparam}(L_{D}(\covx))$. At sufficiently low $\bar{p}$, we have the opposite ordering for all $\pparam$ on the support: $U_{\pparam}(L_{D}(\covx))>U_{\pparam}(L_{D-1}(\covx)) >\dots> U_{\pparam}(L_{1}(\covx))$. At moderate prices, for each pair of deductible lotteries $j<k$ there is a cutoff value $c_{j,k}(x)$ in the interior of $\pparam$'s support, found by solving $U_{\pparam}(L_{j}(\covx))=U_{\pparam}(L_{k}(\covx))$ for $\pparam$. On the left of this cutoff the higher deductible is preferred and on the right  the lower deductible is preferred. In other words, $c_{j,k}(x)$ is the unique coefficient of risk aversion that makes the DM indifferent between $L_j(\covx)$ and $L_k(\covx)$, known to the researcher at any given $\covx$. Those with lower $\pparam$ choose the riskier alternative $L_j(\covx)$, while those with higher $\pparam$ choose the safer alternative $L_k(\covx)$. Provided $U_{\pparam}(\cdot)$ is smooth in $\pparam$, $c_{j,k}(x)$ is smooth in $\covx$. In fact, under CARA, CRRA, or NTD, $c_{j,k}(x)$ is a continuously differentiable monotone function. The prices are such that, under CARA, CRRA, or NTD, whenever $U_{\pparam}(L_{1}(\covx))>U_{\pparam}(L_{j}(\covx))$ it is also the case that $U_{\pparam}(L_{1}(\covx))>U_{\pparam}(L_{j+1}(\covx))$.\footnote{We analytically verify this claim for our application in Appendix \ref{app:AppAssumption}, but it can also be checked numerically for any given dataset.} As we show below, this can be stated as $c_{1,j}(x)<c_{1,j+1}(x)$. That is, if the DM's risk aversion is so low that she prefers the riskiest lottery to a safer one, then she also prefers it to an even safer one. Finally, there are no three-way ties. That is, for a given $\covx$ there are no alternatives $\{j,k,l\}$ such that  $U_{\pparam}(L_{j}(\covx))=U_{\pparam}(L_{k}(\covx))=U_{\pparam}(L_{l}(\covx))$.\footnote{It is straightforward to very this condition, and we do so in our application.}

\subsection{Preferences with Single Crossing Property}\label{S:full}

There is a continuum of DMs.  Each of them faces a choice among a finite number of alternatives, i.e., a choice set, which is denoted ${\cal{D}}=\{1,\dots,D\}$. The number of alternatives is invariant across DMs. Alternatives vary by their utility-relevant characteristics and are distinguished by (at least) one characteristic, $\dor_j \in \R,~ j \in {\cal{D}}$, which is DM invariant. This characteristic reflects the quality of alternative $j$ (e.g., insurance deductible). When it is unambiguous, we may write $\dor_j$ instead of ``alternative $j$''. Other characteristics may vary across DMs or across alternatives. Our analysis rests on the excluded regressor(s) $\covx$. To keep the notation as lean as possible, we state our assumptions and results implicitly conditioning on all remaining characteristics. Hence, alternative $j$ is fully characterized by $(\dor_j,\covx_j)$. We consider two cases. In one case, all $\covx_j$'s are perfectly correlated with a single (common) excluded regressor, $\covx$ (e.g., $\bar{p}$ in our insurance example). In the other case, each $\covx_j$ has its own variation conditional on all other $\covx_k,~ k\neq j$ (e.g., each alternative on the market exhibits locally independent price variation).
\phantomsection
\begin{assumptionSP}{T0}\label{TA0}
	The random variable (or vector) $\covx$ has a strictly positive density on a set $\mathcal{S} \subset \R$ $\left(\mathcal{S} \subset \R^D, ~ \dim \mathcal{S}=D\right)$.
\end{assumptionSP}
Each DM's valuation of the alternatives is defined by a utility function $U_{\pparam}(d_j,x)$, which depends on a DM-specific index $\pparam$ distributed according to $F(\cdot)$ over a bounded support.\footnote{We assume that while $\pparam$ has bounded support, the utility function is well defined for any real valued $\pparam$.}
\phantomsection
\begin{assumptionSP}{T1}\label{TA1}
	The density of $F(\cdot)$, denoted $f(\cdot)$, is continuous and strictly positive on $[0,\bar \pparam]$ and zero everywhere else.
\end{assumptionSP}
The DMs' draws of $\pparam$ are not observed by the researcher. We require that DMs' preferences satisfy the Single Crossing Property (SCP).
\phantomsection
\begin{assumptionSP}{T2} [Single Crossing Property]\label{TA2} For any two alternatives, $d_j$ and $d_k$, there exists a continuously differentiable function $\cmap_{L,R}: \mathcal{S} \to \R_{[-\infty,\infty]}$ such that
	\begin{align*}
	&\util_\pparam(d_L,\covx) > \util_\pparam(\dor_R,\covx)  \quad \forall \pparam \in (-\infty,\cmap_{L,R}(\covx)) \\
	&\util_\pparam(d_L,\covx) = \util_\pparam(\dor_R,\covx)  \quad \pparam = \cmap_{L,R}(\covx) \\
	&\util_\pparam(d_L,\covx) < \util_\pparam(\dor_R,\covx)  \quad \forall \pparam \in (\cmap_{L,R}(\covx),\infty).
	\end{align*}
	where $(L,R)=(j,k)$ or $(L,R)=(k,j)$. We refer to $\cmap_{L,R}(\cdot)$ as the cutoff between $d_L$ and $d_R$.
\end{assumptionSP}
The SCP implies that the DM's ranking of alternatives is monotone in $\pparam$.  In the context of risk preferences, if a DM with a certain level of risk aversion prefers a safer asset to a riskier one, then all DMs with higher risk aversion also prefer the safer asset. Since the cutoffs may be infinite, the SCP does not exclude dominated alternatives.
\begin{definition} [Dominated Alternatives] Given $\covx$, alternative $d_j$ is dominated if there exists an alternative $d_k$ such that $\forall \pparam \in \R$, $\util_\pparam(d_k,\covx) > \util_\pparam(\dor_j,\covx)$.
\end{definition}
We now establish some useful facts that follow from Assumption \ref{TA2}. First, the index $L$ in $c_{L,R}(\cdot)$ indicates the alternative that is preferred on the left of the cutoff.  It is without loss of generality to assume $L=\min(j,k)$ and $R=\max(j,k)$ because of the following fact:
\begin{fact} [Natural Ordering of Alternatives]\label{A:NOA} Suppose Assumption \ref{TA2} holds. Then alternatives can be enumerated such that as $\pparam\rightarrow -\infty$, $\util_\pparam(d_1,\covx) > \util_\pparam(\dor_2,\covx) >\cdots>\util_\pparam(d_D,\covx)$
for all $\covx$ at which no alternative is dominated.
\end{fact}

We assume that alternatives are enumerated according to the Natural Ordering of Alternatives.\footnote{Under this enumeration, $d_{j}$ will be ordered in either ascending or descending order. In our example from the previous section, since $d_j$ refers to the deductible and $\pparam$ is the risk aversion coefficient, the natural ordering implies $d_1>d_2>\dots>d_D.$} As the next fact shows, for high values of $\pparam$ the preference over the Natural Ordering of Alternatives is reversed.
\begin{fact}[Rank Switch]\label{A:RSw}
	Suppose Assumption \ref{TA2} holds. Consider any $\covx$ such that no alternative is dominated. As $\pparam\rightarrow \infty$, $\util_\pparam(d_1,\covx) < \util_\pparam(\dor_2,\covx) <\cdots<\util_\pparam(d_D,\covx).$
\end{fact}
The SCP also has implications for the relative position of the cutoffs. For readability, we state them for alternatives $\{d_1, d_2,d_3\}$, but they hold for any $\{d_j,d_k,d_l\}$, $j<k<l$.

\begin{fact}[Simple Relative Order of Cutoffs]\label{Fact:simpleorder}
		Suppose Assumption \ref{TA2} holds.  Given $\covx$, if $\cmap_{1,2}(x)<\cmap_{1,3}(x)$, then $\cmap_{1,3}(x)<\cmap_{2,3}(x)$ or both $d_1$ and $d_2$ dominate $d_3$ ($\cmap_{1,3}(x)=\cmap_{2,3}(x)=\infty$).
\end{fact}

The next fact concerns the relative order of cutoffs for non-dominated alternatives. Before stating it, it is convenient to define Never-the-First-Best Alternatives.

\begin{definition} [Never-the-First-Best]\label{A:NFB} Given $\covx$, alternative $d_j$ is Never-the-First-Best in $\cal{D}$ if for every $\pparam$ there exists another alternative $d_k(\pparam)$ in $\cal{D}$ such that $\util_\pparam(d_k(\pparam),\covx)>\util_\pparam(d_j,\covx)$.
\end{definition}

\begin{fact}[Cutoff Relative Order]\label{A:CO}
	Suppose that Assumption \ref{TA2} holds. If, given $\covx$, alternatives $d_1$, $d_2$, and $d_3$ are not dominated, then one and only one of the following cases holds:
	\begin{enumerate}
		\item [(i)] $\cmap_{1,2}(x)<\cmap_{1,3}(x)<\cmap_{2,3}(x)$ and $d_2$ is the first best in $\{d_1,d_2,d_3\}$, $\forall \pparam \in (\cmap_{1,2}(x),\cmap_{1,3}(x))$;
		\item [(ii)] $\cmap_{1,2}(x)>\cmap_{1,3}(x)>\cmap_{2,3}(x)$ and $d_2$ is Never-the-First-Best in $\{d_1,d_2,d_3\}$;
		\item [(iii)]$\cmap_{1,2}(x)=\cmap_{1,3}(x)=\cmap_{2,3}(x)$ and $d_2$ is strictly worse than either $d_1$ or $d_3$ for all $\pparam$ except for $\pparam=\cmap_{1,2}(x)$ where there is a three-way tie among these alternatives.
	\end{enumerate}	
\end{fact}
Fact \ref{A:CO} is a convenient way to distill and exploit the SCP. In particular, for any $\covx$, the complete preference order of the alternatives is known for all DMs as well as the identity (of the preference parameter) of the DM indifferent between any two alternatives $d_j$ and $d_k$. 

\section{Identification}\label{sec:idchallenge}
The classic identification argument for discrete choice under full consideration rests on the following four canonical assumptions.
\phantomsection
\begin{assumptionSP}{I0}\label{IA0}
The random variable (or vector) $x$ is \textbf{independent} of preferences.
\end{assumptionSP}
\vspace{-1em}
\phantomsection
\begin{assumptionSP}{I1}\label{IA1}
$\exists \covxs \subset \mathcal{S}$ s.t. $\cmap_{1,2}(x)$ \textbf{covers the support} of $\pparam$: $[0,\bar{\pparam}]\subset \{\cmap_{1,2}(x), x\in \covxs\}$.
\end{assumptionSP}
\vspace{-1em}
\phantomsection
\begin{assumptionSP}{I2}\label{IA2}
Consideration is \textbf{independent} of preferences.
\end{assumptionSP}
\vspace{-1em}
\phantomsection
\begin{assumptionSP}{I3}\label{IA3}
Consideration is \textbf{independent} of $\covx$.
\end{assumptionSP}
The last two conditions are vacuous in the standard full consideration model, while the first two are typically stated as data requirements.

We first discuss how to obtain identification and the role of Assumptions \ref{IA0}-\ref{IA3} in the simplest case of two alternatives (Section \ref{S:two alt}). We then consider the general model with $D$ alternatives. Table \ref{OrgTable} organizes our results by assumptions imposed, the consideration mechanism assumed, data availability, and the theorems' conclusions. Theorems \ref{T:Benchmark}-\ref{T:Basic-G} in Section \ref{S:3Alt} demonstrate that the preference distribution and some features of the consideration mechanism are identified with a single excluded regressor. Next, we show that alternative-specific variation allows for identification of both the preference distribution and the consideration mechanism when consideration depends on preferences and one of the excluded regressors (Theorem \ref{Broad-alternative specific} and Corollary \ref{Broad-alternative specific_corollary} in Section \ref{subsec:altspec}). We discuss testing for limited consideration in Section \ref{s:testing}. We then turn to the ARC model in Section \ref{S:ARC}. We show that the full model is identified with a single excluded regressor (Theorem \ref{T:ARC_independent}). Moreover, identification attains for a particular case where consideration depends on preferences (Theorem \ref{T:ARC_dependent_1}). Finally, Theorem \ref{Broad-alternative specific_ARC} shows that with alternative-specific variation, identification attains when consideration of each alternative depends both on preferences and its own regressor, without requiring full support.

\begin{table}[!h]\centering

\scriptsize

	\caption{Identification Theorems}
		\label{OrgTable}
	\begin{tabular}{|l| c c c c|c c c| c c  | c c|}
		\hline \hline \\[-1em]
		              & \multicolumn{4}{c|}{Assumptions}   &\multicolumn{3}{c|}{Consideration Mechanism}  & \multicolumn{2}{c|}{Excluded Regressor} & \multicolumn{2}{c|}{Identification Result} \\[0.5em]	
		              &&&&&Generic& Loosely & ARC&Single& Alternative Specific & Preferences & Consideration \\[0em]
		              & I0      & I1      &  I2  & I3      & &Ordered &      &     &&  &   \\[0.5em]
		\hline \\[-1em]
 		Theorem 1     & \checkmark  & \checkmark  &\checkmark& \checkmark &  \checkmark &        &          & \checkmark&                               & \checkmark & $\checkmark^1$ \\[0.5em]
 		 Theorem 2     & \checkmark  & \checkmark  &      & \checkmark &  \checkmark          &          & &   \checkmark&                            & \checkmark & $\checkmark^1$  \\[0.5em]
 	    Theorem 3     & \checkmark  &         &\checkmark& \checkmark &         & \checkmark &          & \checkmark&                               & $\checkmark$ & \\[0.5em]
 	\hline \\[-1em]
     	Theorem 4     & \checkmark  &         &      & \checkmark &  \checkmark          &          &      &&\checkmark                         & \checkmark & \checkmark  \\[0.5em]
      	Corollary 1  & \checkmark  &         &      &        &  \checkmark          &          &       &&\checkmark                         & \checkmark & \checkmark  \\[0.5em]  	    	
     	\hline \\[-1em]
     	Theorem 5     & \checkmark  & \checkmark  &\checkmark& \checkmark &    &             &\checkmark   &   \checkmark&                               & \checkmark & \checkmark \\[0.5em]      	
     	Theorem 6     & \checkmark  & \checkmark  &      & \checkmark &                 &&\checkmark     & \checkmark&                               & \checkmark & \checkmark   \\[0.5em]
     	Theorem 7     & \checkmark  &         &      &        &               &  &\checkmark     &       &\checkmark                         & \checkmark & \checkmark   \\[0.5em]
     \hline \hline
	\end{tabular}
\caption*{\footnotesize  \begin{flushleft} Note 1: In Theorems 1 and 2 we identify features of the consideration mechanism.  \\
\end{flushleft}}
\end{table}
\subsection{The Role of the Canonical Assumptions}\label{S:two alt}

Let the choice set be binary and suppose that the DM considers both alternatives. In addition, let $\covx$ be a scalar so that there is a single excluded regressor. Under Assumptions \ref{TA0}-\ref{TA2} and \ref{IA0}-\ref{IA3}, any realization of $x$ is associated with a single conditional moment in the data:
\begin{equation*}
	\Pr(d=d_1|x) = \int_{0}^{\cmap_{1,2}(\covx)}dF = F(\cmap_{1,2}(x)),
\end{equation*}
because the DM chooses $d_1$ if and only if her preference parameter is less than $\cmap_{1,2}(x)$.
The distribution $F(\cdot)$ is non-parametrically identified, since for any $\pparam$ on the support there is an $x$ such that $\pparam=\cmap_{1,2}(x)$.

We emphasize two points. First, given a family of utility functions, for any $x$ the value of the cutoff can be solved for. Hence, the function $\cmap_{1,2}(x)$ (and its derivatives) can be treated as data. Second, Assumption \ref{IA1} requires that the cutoff reaches both ends of the support: there exist $x^0$ and $x^1$ such that $F(\cmap_{1,2}(x^0))=0$ and $F(\cmap_{1,2}(x^1))=1$.

Turning to limited consideration, suppose that $d_1$ is considered with probability $0<\aparam_{1}\leq 1$, and whenever it is considered so is $d_2$.\footnote{With two alternatives this implies that $d_2$ is always considered.} Then, $d_1$ is chosen when it is considered and it is preferred to $d_2$, yielding:
\begin{align}
	\Pr(d=d_1|x) = \aparam_1 F(\cmap_{1,2}(x)) \quad \text{and} \quad
	\frac{d\Pr(d=d_1|x)}{dx} = \aparam_1 f(\cmap_{1,2}(x))\frac{d\cmap_{1,2}(x)}{dx}.\label{eq:L0}
\end{align}

At first glance, it appears that the distribution of preferences is identified up to a constant. Yet, at the boundary of the support   $\Pr(d=d_1|x^1) = \aparam_1 F(\cmap_{1,2}(x^1)) = \aparam_1$, so that $\aparam_{1}$ is identified. Once $\aparam_1$ is known, the distribution $F(\cdot)$ is identified by varying $\cmap_{1,2}(\covx)$ over the support of $\pparam$, similar to the full consideration case. We now explore what happens to identification if Assumptions \ref{IA0}--\ref{IA3} are not satisfied.

\textbf{Assumption \ref{IA0} fails:} the variation in $x$ is not independent of preferences. Then $F(\cdot)$ is not non-parametrically identified under either full or limited consideration.

\textbf{Assumption \ref{IA1} fails:} the variation in $\covx$ is such that $\cmap_{1,2}(x)$ only covers an interval $[\pparam^{l},\pparam^{u}] \subset [0,\bar{\pparam}]$. Then the data provide no information about preferences outside of the interval $[\pparam^{l},\pparam^{u}]$.  Inside the interval, the conditional distribution $F(\pparam|\pparam \in [\pparam^{l},\pparam^{u}])=\frac{F(\pparam)-F(\pparam^{l})}{F(\pparam^{u})-F(\pparam^{l})}$ is identified under both limited and full consideration. The consideration probability (and hence the scale of $F(\cdot)$) is partially identified and satisfies the bounds $\Pr(d=d_1|\covx^{u}) \leq \aparam_1 \leq 1$, where $\covx^{u}$ is such that $\cmap_{1,2}(\covx^{u})=\pparam^{u}$. Point identification can be attained if an additional assumption is maintained to pin down the scale of $F(\cdot)$. For example, one can simply assume full consideration and set $\aparam_{1}=1$.

\textbf{Assumption \ref{IA2} fails:} $\aparam_{1}$ depends on preferences and this dependence is arbitrary. Then identification breaks down completely as there is one data moment to identify two unknown objects. However, since we assume -- as it is common in the econometrics literature -- that the density function of $\pparam$ is continuous and strictly positive, identification is possible for some types of dependence between consideration and preferences. Suppose there are two consideration types:
\begin{align*}
\aparam_1(\pparam)&=
\begin{cases}
 \underline{\aparam}_1, & \forall \pparam \in [0,\pparam^{*})\\
 \overline{\aparam}_1, & \forall \pparam \in [\pparam^{*},\bar{\pparam}]
\end{cases},
\end{align*}
where $\pparam^{*}$ is an unobserved breakpoint. We show that $\underline{\aparam}_1$, $\overline{\aparam}_1$, and $\pparam^{*}$ are identified. First, the product $\aparam_1(\pparam)f(\pparam)$ is identified under Assumptions \ref{IA0}, \ref{IA1}, and \ref{IA3}, since
\begin{equation}\label{eq:L1}
\frac{d\Pr(d=d_1|x)}{dx}=\frac{d}{dx}\left(\int_{0}^{\cmap_{1,2}(\covx)}\aparam_1(\pparam)dF\right)= \aparam_1(\pparam)
f(\pparam)\frac{d \cmap_{1,2}(x)}{dx}
\end{equation}
at $\pparam=\cmap_{1,2}(x)$. The product $\aparam_1(\pparam)f(\pparam)$ is discontinuous only at the point $\pparam^{*}$. Thus, the breakpoint is identified by continuously varying $c_{1,2}(\covx)$ across $[0,\bar{\pparam}].$ Next, the ratio $\frac{\underline{\aparam}_1}{\overline{\aparam}_1}$ is identified by the ratio of the right and left derivatives of $\Pr(d=d_1|x)$ at the breakpoint $\covx^{*}$ ($\pparam^{*}=\cmap_{1,2}(\covx^{*})$). The quantity $F(\pparam^{*})$ is identified by the ratio:
\begin{align*}
\frac{\Pr\left(d=d_1|\covx^{*}\right)}{\Pr(d=d_1|\covx^1)-\Pr\left(d=d_1|\covx^{*}\right)} & =
\frac{\underline{\aparam}_1}{\overline{\aparam}_1}\cdot
\frac{F(\pparam^{*})}{1-F(\pparam^{*})}.
\end{align*}
Hence, $\underline{\aparam}_1$ and $\overline{\aparam}_1$ are identified. Identification of $F(\cdot)$ on the entire support follows from Assumption \ref{IA1}. The same argument above applies if the probability of considering an alternative discretely jumps in $\covx$ (i.e., \textbf{Assumption \ref{IA3} fails}). Concretely, suppose there is a breakpoint in $\aparam_{1}(\covx)$ at $\covx^{*}$ and let $\pparam^{*}=\cmap_{1,2}(x^{*})$. The breakpoint $\covx^{*}$ is identified by the point of discontinuity in Equation \eqref{eq:L1}, and the rest follows.

To summarize the case of the binary choice set, the only seemingly real difference in identification is that without large support the scale of the preference distribution $F(\cdot)$ is partially identified under limited consideration, while it is assumed to be known under full consideration. The key to identification is a one-to-one mapping from a data moment, $\Pr\left(d=d_1|\covx\right)$, and the preference distribution $F(\cdot)$ at a single point on the support, $\cmap_{1,2}(x)$. As we will show next, even with just a single excluded regressor, Assumptions \ref{IA0}-\ref{IA3} allow for such a mapping to be constructed for a generic consideration mechanism and a choice set of arbitrary size.

\subsection{Single Common Excluded Regressor}\label{S:3Alt}

We start by introducing general notation for consideration probabilities.
\begin{definition}
	Let $\mathcal{Q}_{\pparam}^{\covx}(\cal{K})$  be the probability that, given $\covx$, the DM with preference parameter $\pparam$ draws consideration set $\cal{K} \subset \Dc$ conditional on $\covx$.
	
	Let $\mathcal{O}_{\pparam}^{\covx}(\mathcal{A};\mathcal{B})$ be the probability that, given $\covx$, every alternative in set $\mathcal{A}$ is in the consideration set and every alternative in set $\mathcal{B}$ is not for the DM with preference parameter $\pparam$:
	\[ \mathcal{O}_{\pparam}^{\covx}(\mathcal{A};\mathcal{B})\equiv\sum_{\Kc: ~\mathcal{A} \subset \Kc, ~ \mathcal{B} \cap \Kc =\emptyset}\mathcal{Q}_{\pparam}^{\covx}(\Kc).
	\]
\end{definition}
The subscript is suppressed when consideration does not depend on preferences, and the superscript is suppressed when it does not depend on the excluded regressor(s).

To ease exposition, we build our discussion around a choice set with three alternatives, $ \Dc = \{d_1,d_2,d_3\}$, such that $\cmap_{1,2}(x)<\cmap_{1,3}(x)<\cmap_{2,3}(x)$ for all $\covx$. That is, by Fact \ref{Fact:simpleorder}, if $U_{\pparam}(d_1,\covx)>U_{\pparam}(d_2,\covx)$ then $U_{\pparam}(d_1,\covx)>U_{\pparam}(d_3,\covx)$ for all $\covx$.  Suppose consideration is independent of preferences and of the excluded regressor. Then the choice frequencies of $d_1$ and $d_3$ are
\begin{align*}
	\Pr(d=d_1|x) = &\mathcal{O}(\{d_1,d_2\};\emptyset) F(\cmap_{1,2}(x)) +\mathcal{O}(\{d_1,d_3\};d_2)F(\cmap_{1,3}(\covx))+\mathcal{O}(d_1;\{d_2,d_3\});\\
	\\
		\Pr(d=d_3|x) = &\mathcal{O}(\{d_1,d_3\};d_2)(1-F(\cmap_{1,3}(\covx)))
		+\mathcal{O}(\{d_2,d_3\};\emptyset) (1-F(\cmap_{2,3}(x))) +\mathcal{O}(d_3;\{d_1,d_2\}).
\end{align*}
Consider the expression for $\Pr(d=d_1|x)$. Its RHS has three terms. The first term captures the case when $d_1$ is considered along with $d_2$, which happens with probability $\mathcal{O}(\{d_1,d_2\};\emptyset).$ Given the relative position of the cutoffs, whether $d_3$ is considered or not is irrelevant. The DM will choose $d_1$ over $d_2$ if and only if her preference parameter is below $\cmap_{1,2}(x)$. The second term captures the case when $d_1$ is considered along with $d_3$, but $d_2$ is not considered, which happens with probability $\mathcal{O}(\{d_1,d_3\};d_2)$. Then the relevant cutoff for choosing $d_1$ is $\cmap_{1,3}(x)$. 
Third, when $d_1$ is the only alternative considered, it is chosen regardless of the DM's risk aversion. This event occurs with probability $\mathcal{O}(d_1;\{d_1,d_2\})$.

Since there are two cutoffs, $\cmap_{1,2}(x)$ and $\cmap_{1,3}(x)$, that enter the moment $\Pr(d=d_1|x)$, there is not, without additional assumptions, a one-to-one mapping between the moment and the preference distribution at one point on the support, as it was the case in Section \ref{S:two alt}. That is, as $\covx$ changes, the observed choice frequency of $d_1$ may change because of two types of marginal DMs: those indifferent between $d_1$ and $d_2$, and those indifferent between $d_1$ and $d_3$.  This is apparent in the following derivative:
\begin{align}\label{eq:L2}
	\frac{d  \Pr(d=d_1|x)}{dx}= &\mathcal{O}(\{d_1,d_2\};\emptyset)f(\cmap_{1,2}(x))\frac{d \cmap_{1,2}(x)}{dx}+
	\mathcal{O}(\{d_1,d_3\};d_2)\frac{d \cmap_{1,3}(x)}{dx}.
\end{align}
The corresponding equation for $\Pr(d=d_3|x)$ does not immediately help, as it brings about $f(\cdot)$ evaluated at yet another cutoff, $\cmap_{2,3}(x)$:
\begin{align}\label{eq:L23}
	\frac{d  \Pr(d=d_3|x)}{dx}= &-\mathcal{O}(\{d_1,d_3\};d_2)\frac{d \cmap_{1,3}(x)}{dx}-\mathcal{O}(\{d_2,d_3\};\emptyset)\frac{d \cmap_{2,3}(x)}{dx}.
\end{align}
\subsubsection{Identification with Large Support}
When Assumption \ref{IA1} holds, we can construct a one-to-one mapping sequentially. The algorithm for doing so consists of four steps. First, we rewrite Equation \eqref{eq:L2} as
\begin{equation}\label{eq:L_2-1}
\frac{d\Pr(d=d_1|\covx)}{dx}= \hat{f}(\cmap_{1,2}(\covx))\frac{d\cmap_{1,2}(\covx)}{dx}
+\phi\hat{f}(\cmap_{1,3}(\covx))\frac{d\cmap_{1,3}(\covx)}{dx},
\end{equation}
	
where $\phi \equiv\frac{ \mathcal{O}(\{d_1,d_3\};d_2)}{\mathcal{O}(\{d_1,d_2\};\emptyset)}$ and $\hat f(\pparam) \equiv \mathcal{O}(\{d_1,d_2\};\emptyset) f(\pparam)$.
Second, for $\pparam$'s near the far end of the support, we can find $\covx$ and $\covx'$ such that $\pparam= \cmap_{1,2}(x)<\bar{\pparam}<\cmap_{1,3}(x)$ and $\pparam=\cmap_{1,3}(x')<\bar{\pparam}<\cmap_{2,3}(x')$. For any such pair, $f(\cmap_{1,3}(x))=f(\cmap_{2,3}(x'))=0$, and, hence, by Equations \eqref{eq:L2} and \eqref{eq:L23}:
	\begin{align*}
		\frac{d\Pr(d=d_1|\covx)}{dx}
		=\textcolor{white}{\phi}\hat{f}(\pparam)\frac{d\cmap_{1,2}(x)}{dx} \quad \text{and} \quad
		\frac{d\Pr(d=d_3|\covx')}{dx}
		= -\phi\hat{f}(\pparam) \frac{d\cmap_{1,3}(x')}{dx}.
	\end{align*}
The first equation identifies $\hat{f}(\pparam)$, while the ratio of the two equations identifies $\phi$. Third, whenever $\hat{f}(\cmap_{1,3}(\covx))$ is known, $\hat{f}(\cmap_{1,2}(\covx))$ is uniquely pinned down by Equation \eqref{eq:L_2-1}. Because $\cmap_{1,2}(\covx)<\cmap_{1,3}(\covx)$, $\forall \covx$, we can learn $\hat{f}(\cdot)$ sequentially:
\begin{enumerate}
	\item Take an $\covx^1$ such that $\hat{f}(\cmap_{1,3}(\covx^1))$ is already known, learn $\hat{f}(\cmap_{1,2}(x^1))$;
	\item Take $\covx^{2}$ such that $\cmap_{1,3}(\covx^{2})=\cmap_{1,2}(\covx^{1})$, learn $\hat{f}(\cmap_{1,2}(\covx^{2}))$;
	\item Let $\covx^{1}=\covx^{2}$. Repeat Step 2 until the entire support has been covered, i.e., $\cmap_{1,2}(\covx^{2})\leq 0$.
\end{enumerate}
For this approach to work, $ \cmap_{1,3}(\covx)$ cannot ``catch up'' to $\cmap_{1,2}(\covx)$ (i.e., as assumed, $\cmap_{1,2}(x) < \cmap_{1,3}(\covx)$ whenever $\cmap_{1,2}(x)$ is on the support). This requires that DMs with preference coefficients on the support are never indifferent between $d_1$ and two other alternatives -- i.e. there are no three way ties involving $d_1$. 
Fourth, integration of $\hat{f}(\pparam)$ over the entire support recovers the scale and the true density. Indeed,
\[
\int_{0}^{\bar{\pparam}} \hat{f}(\pparam)d\pparam=\mathcal{O}(\{d_1,d_2\};\emptyset)\int_{0}^{\bar{\pparam}} f(\pparam)d\pparam =\mathcal{O}(\{d_1,d_2\};\emptyset)
\]
pins down $\mathcal{O}(\{d_1,d_2\};\emptyset)$, and hence $f(\cdot)$ is identified. A generalization of this strategy yields our first formal result.

\begin{theorem}\label{T:Benchmark}
	Suppose Assumptions \ref{IA0}, \ref{IA2}, \ref{IA3}, \ref{TA0}-\ref{TA2} hold, and
	\begin{enumerate}
		\item The consideration mechanism is s.t. with positive probability $d_1$ and $d_2$ are considered together;
       	\item Assumption \ref{IA1} holds for $\mathcal{X} \subset \mathcal{S}$ s.t. $\forall \covx \in \covxs$
		\[
		U_{\pparam}(d_1,x)>U_{\pparam}(d_j,\covx) \Rightarrow U_{\pparam}(d_1,\covx)>U_{\pparam}(d_{j+1},\covx), \quad \forall j>1.
		\]
	\end{enumerate}
	Then $f(\cdot)$ is identified and so are $\mathcal{O}(d_1;\emptyset)$ and $\mathcal{O}(\{d_1,d_2\};\emptyset)$.
	For $j>2$, if $\Pr(d=d_j|\covx)>0$ for some $\covx$, then $\mathcal{O}(\{d_1,d_j\};\{d_2,\dots, d_{j-1}\})$ is identified.
\end{theorem}
The first assumption of the theorem ensures that a generalized version of Equation \eqref{eq:L_2-1} is informative. The second assumption implies that the cutoffs for alternative $d_1$ are ordered: $\cmap_{1,j}(x)<\cmap_{1,j+1}(x)$. While Theorem \ref{T:Benchmark} requires large support for the excluded regressor, it does not generally require it to exhibit variation that forces alternative $d_1$ to go from being the first best to the least preferred. Rather, the theorem requires that at one extreme of the support alternative $d_1$ dominates all others. However, at the other extreme we only require that $d_2$ is preferred to $d_1$ for all DMs. Identification is attained for any consideration mechanism that allows $d_1$ and $d_2$ to be considered together with positive probability. Moreover, if the probability of being considered together is zero for $d_1$ and $d_2$, but positive for $d_1$ and $d_3$, the theorem still holds as long as the assumptions of the theorem hold for $d_3$ instead of $d_2$. Theorem \ref{T:Benchmark} identifies some features of the consideration mechanism. These features may be sufficient for identifying the entire mechanism. In particular, as shown in Section \ref{S:ARC}, Theorem \ref{T:Benchmark} yields identification of the ARC model, including the consideration mechanism.

\textbf{Dependence between consideration and preferences.} We next generalize the example in Section \ref{S:two alt} by allowing for high/low consideration types.
\begin{assumptionSP}{I2.BCT}[Binary Consideration Types]\label{IA2.BCT}
	For some  unknown $\pparam^* \in (0, \bar{\pparam})$:
	\[
	\mathcal{Q_{\pparam}(K)} =
	\begin{cases}
		\mathcal{\underline{Q}(K)} & \text{if } \pparam < \pparam^*  \\
		\mathcal{\overline{Q}(K)}  & \text{if }      \pparam > \pparam^*  \\
	\end{cases}
	\]	
	where, $\forall \pparam$ and $\forall \mathcal{K} \subset \Dc$, $\sum_{\mathcal{K} \subset \Dc} \mathcal{Q_{\pparam}(K)}=1$ and $\mathcal{Q_{\pparam}(K)}\geq 0$.
\end{assumptionSP}
\begin{theorem}\label{T:Binary}
	Suppose Assumptions \ref{IA0}, \ref{IA2.BCT}, \ref{IA3}, \ref{TA0}-\ref{TA2}, and Condition 2 of Theorem \ref{T:Benchmark} hold. Suppose Condition 1 of Theorem \ref{T:Benchmark} holds for all $\pparam$. Then $f(\cdot)$ is identified and so is $\mathcal{O}_{\pparam}(\{d_1,d_2\};\emptyset)$.
	Suppose $\frac{d\Pr(d=d_1|\covx)}{dx}$ is discontinuous.  Then $\pparam^*$ is identified.  If, in addition,  $\cmap_{1,j}(\covx) < \pparam^*$ for some $\covx \in \mathcal{X}$ and $j>2$, then $\mathcal{O}_{\pparam}(\{d_1,d_j\};\{d_2,\dots, d_{j-1}\})$ is also identified.
\end{theorem}
A discontinuity in $\frac{d\Pr(d=d_1|\covx)}{dx}$ may occur when a cutoff $\cmap_{1,j}(\covx)$ crosses $\pparam^*$. In some cases it may not happen despite binary consideration. For example, the probability of considering $d_1$ and $d_2$ may jump but in a way that $\mathcal{O}_{\pparam}(\{d_1,d_2\};\emptyset)$ remains constant. In such a case, $f(\cdot)$ is identified but not necessarily the breakpoint $\pparam^*$.

The theorem holds if Assumption \ref{IA2.BCT} is replaced with
\[
\mathcal{Q}^{\covx}(\mathcal{K}) =
\begin{cases}
	\mathcal{\underline{Q}(K)} & \text{if } \covx < \covx^*  \\
	\mathcal{\overline{Q}(K)}  & \text{if }      \covx > \covx^*
\end{cases}
\]	
for some unknown $\covx^* \in \mathcal{S}$. In sum, preferences can be identified even when there are threshold effects affecting consideration. Assumption \ref{IA2.BCT} is one instance where Assumption \ref{IA2} does not hold but identification attains. Another instance, which we establish for the ARC model in Section \ref{sec:gendependence}, is proportionately shifting consideration.

\subsubsection{Identification without large support}
Returning to our example with three alternatives, it is immediate to see that if whenever $d_1$ is considered so is $d_2$, i.e. $\mathcal{O}(\{d_1,d_3\};d_2)  = 0$, the one-to-one mapping is restored. Indeed, the second term on the RHS of Equation \eqref{eq:L2} disappears and we are back to Equation \eqref{eq:L0}.
\begin{proposition}\label{P:Basic}
	Suppose Assumptions \ref{IA0}, \ref{IA2}, \ref{IA3}, \ref{TA0}-\ref{TA2} hold, and
	\begin{enumerate}
		\item The consideration mechanism is such that $d_1$ is considered with positive probability and whenever it is considered so is $d_2$;
		\item There exists $\covxs \subset \mathcal{S}$ such that $\cmap_{1,2}(\covx)$, $\covx \in \covxs$, covers $[\pparam^{l},\pparam^{u}]\subset [0,\bar \pparam]$ and $\forall \covx \in \covxs$
		\[
		U_{\pparam}(d_1,\covx)>U_{\pparam}(d_2,\covx) \Rightarrow U_{\pparam}(d_1,\covx)>U_{\pparam}(d_j,\covx), \quad \forall j>2.
		\]
	\end{enumerate}
	Then $F(\pparam|\pparam \in [\pparam^{l},\pparam^{u}])$ is identified.
\end{proposition}
The proposition above uses (the derivative of) $\Pr(d=d_1|x)$ to create the one-to-one mapping from data to the preference density function. Depending on the consideration mechanism, the same can be achieved using the derivative of $\Pr(d\in \{d_1,d_2,\dots,d_j\}|x)$.
\begin{definition}[Loosely Ordered Consideration]\label{L-Ordered}
	The consideration mechanism is loosely ordered around $j$, $j<D$, if whenever alternatives $d_k$ and $d_l$, $k\leq j<l$, are both considered, so are $d_{j}$ and $d_{j+1}$. In addition, $d_{j}$ and $d_{j+1}$ have a positive probability of being considered together.
\end{definition}
	
\begin{theorem}\label{T:Basic-G}
	Suppose Assumptions \ref{IA0}, \ref{IA2}, \ref{IA3}, \ref{TA0}-\ref{TA2} hold, and
	\begin{enumerate}
		\item The consideration mechanism is loosely ordered around $j$.
		\item There exists $\covxs \subset \mathcal{S}$ such that $\cmap_{j,j+1}(\covx)$, $\covx \in \covxs$, covers $[\pparam^{l},\pparam^{u}]\subset [0,\bar \pparam]$ and $\forall \covx \in \covxs$
		\begin{align*}
			U_{\pparam}(d_j,\covx)>U_{\pparam}(d_{j+1},\covx) &\Rightarrow U_{\pparam}(d_j,\covx)>U_{\pparam}(d_k,\covx), \quad~~ \forall k>j+1,\\
			U_{\pparam}(d_{j+1},\covx)>U_{\pparam}(d_j,\covx) &\Rightarrow U_{\pparam}(d_{j+1},\covx)>U_{\pparam}(d_k,\covx), \quad \forall k<j.
		\end{align*}
	\end{enumerate}
	Then $F(\pparam|\pparam \in [\pparam^{l},\pparam^{u}])$ is identified.	
\end{theorem}
Condition 1 in Theorem \ref{T:Basic-G} -- a loosely ordered consideration mechanism -- splits the choice set into ``low quality'' and ``high quality'' sets. Any subset of the low quality set can form the consideration set and so can any subset of the high quality set. However, if a consideration set contains both high and low quality alternatives, then it must also contain the ``bridging'' alternatives $\{d_{j},d_{j+1}\}$. The following mechanisms can generate loosely ordered consideration:
\begin{enumerate}
	\item [I.] \textbf{Bottom-Up consideration:} Alternative $d_k$ is considered only if $d_{k-1}$ is considered;
	\item [II.] \textbf{Top-Down consideration:} Alternative $d_k$ is considered only if $d_{k+1}$ is considered;
	\item [III.] \textbf{Center-to-edges consideration:} Alternative $d_{j}$, $1<j<D$,
	 is always considered. Alternative $d_k$, $k>j$, is considered only if $d_{k-1}$ is considered. Alternative $d_k$, $k<j$, is considered only if $d_{k+1}$ is considered;
	\item [IV.] \textbf{Trimmed-from-the-edges consideration:} Only consideration sets of the form $\mathcal{K} = \{d_k,d_{k+1},\dots,d_{k+l}\}$ can occur with positive probability.
\end{enumerate}
The identification result in Theorem \ref{T:Basic-G} extends to mixtures of these mechanisms. They cover a wide array of models including  versions of  threshold models \citep{kimya2018choice}, (partial) elimination-by-aspects \citep{Tversky72}, extremeness aversion \citep{simonson_tversky_1992}, and edge aversion \citep{teigen1983studies,christenfeld1995choices,rubinstein1997naive,attali2003guess}, 
as well as models that embed budget or liquidity constraints.

Condition 2 in Theorem \ref{T:Basic-G} requires that whenever a DM prefers $d_j$ to $d_{j+1}$, she also prefers $d_j$ to all high quality alternatives; and whenever a DM prefers $d_{j+1}$ to $d_{j}$, she also prefers $d_{j+1}$ to all low quality alternatives. This condition can be tested in any given dataset and is automatically satisfied if no alternative is never-the-first-best.

The fundamental difference between Theorems \ref{T:Benchmark} and \ref{T:Basic-G} is that the former imposes the large support requirement, while the latter does not. On the other hand, Theorem \ref{T:Benchmark} imposes less restrictions on the consideration mechanism than Theorem \ref{T:Basic-G}.

\subsection{Alternative-specific Excluded Regressors}\label{subsec:altspec}

With alternative-specific excluded regressors we can allow for consideration to depend on preferences. To illustrate, we continue to assume that the choice set is $\{d_1,d_2,d_3\}$. However, now each alternative has its own regressor $x_j$ that only affects the utility of alternative $j$: $\covx=(\covx_{1},\covx_{2},\covx_{3})$. In addition, these regressors vary independently of one another and each consideration set contains at least two alternatives.

Identification is built on the following insight. Consider the change in the choice frequency of alternative $d_1$ in response to an incremental change in $\covx_2$ (e.g., a price increase for alternative $d_2$). The DMs who may switch to $d_1$ are those indifferent between $d_1$ and $d_2$ and consider them both. If these DMs prefer $d_1$ and $d_2$ to $d_3$, whether $d_3$ is considered is irrelevant; otherwise, for the response to occur, $d_3$ should not be considered.  These two cases translate to the following statements: \emph{(i)}  $\cmap_{1,2}(x)<\cmap_{1,3}(x)<\cmap_{2,3}(x)$; and \emph{(ii)} $\cmap_{2,3}(x)<\cmap_{1,3}(x)<\cmap_{1,2}(x)$ and $d_3$ is not considered. No other ordering of cutoffs can occur by Fact \ref{A:CO}. With alternative specific variation we can construct two vectors of regressors, $\covx^i$ and $\covx^{ii}$, such that $\pparam =  c_{1,2}(\covx^i) <\cmap_{1,3}(\covx^i)<\cmap_{2,3}(\covx^i)$ and $c_{2,3}(\covx^{ii}) <\cmap_{1,3}(\covx^{ii})<\cmap_{1,2}(\covx^{ii})=\pparam$.\footnote{ First, we can construct an $\covx$ such that $U_{\pparam}(d_1,\covx) = U_{\pparam}(d_2,\covx) = U_{\pparam}(d_3,\covx)$. To do so, we fix the price of the first alternative, $\covx_1$, and find a price for the second alternative that makes the DM with preference parameter $\pparam$ indifferent between $d_1$ and $d_2$. Since $\covx_3$ does not affect the utility of $d_1$ nor $d_2$, we can find an $\covx_3$ so that the DM is indifferent between $d_3$ and $d_1$, and hence she is indifferent between all three alternatives. The two cases are then constructed by taking a small perturbation of $\covx_3$. Taken in a direction that reduces  $U_{\pparam}(d_3,\covx)$ generates Case \emph{(i)}; and in the opposite direction Case \emph{(ii)}.} The derivative
of the choice frequency of $d_1$ with respect to $\covx_2$ for these cases are, respectively:
\begin{align*}
(i) \quad \frac{\partial \Pr(d=d_1|x)}{\partial x_2}
	&=\left[\mathcal{Q}_{\pparam}(\{d_1,d_2\})+\mathcal{Q}_{\pparam}(\{d_1,d_2,d_3\})\right] f(\pparam ) \frac{\partial \cmap_{1,2}(x) }{\partial x_2};\\
	(ii) \quad \frac{\partial \Pr(d=d_1|x)}{\partial x_2} &=\left[\mathcal{Q}_{\pparam}(\{d_1,d_2\}) \color{white}+ \mathcal{Q}_{\pparam}(\{d_1,d_2,d_3\}) \color{black}\right] f(\pparam)   \frac{\partial \cmap_{1,2}(x) }{\partial x_2}.
\end{align*}
It follows that $\mathcal{Q}_{\pparam}(\{d_1,d_2\})f(\pparam)$ and $\mathcal{Q}_{\pparam}(\{d_1,d_2,d_3\})f(\pparam)$ are identified. In a similar fashion, $Q_{\pparam}(\{d_1,d_3\})f(\pparam)$ and $Q_{\pparam}(\{d_2,d_3\})f(\pparam)$ are identified. Hence, the consideration probability of each non-singleton  set is identified up to the same scale. The scale, however, is identified because consideration probabilities must sum to one: $f(\pparam)= \sum_{\mathcal{K}} \mathcal{Q}_{\pparam}(\mathcal{K})f(\pparam)$.  Hence, $\mathcal{Q}_{\pparam}(\mathcal{K})$ is also identified for each $\mathcal{K}$. The following theorem generalizes this idea.
\begin{definition}[Alternative-Specific Variation]\label{ASV-def}
	We say that there is alternative-specific variation if $U_{\pparam}(d_j,\covx)$ depends only on $\covx_j$: $\frac{\partial U_{\pparam}(d_j,\covx)}{\partial x_k}\neq 0 \Leftrightarrow k=j$.
\end{definition}
\begin{theorem}\label{Broad-alternative specific}
	Suppose Assumptions \ref{IA0}, \ref{IA3}, \ref{TA0}-\ref{TA2} hold, there is alternative-specific variation, and the choice set contains at least three alternatives. Suppose
	\begin{enumerate}
		\item  Each consideration set contains at least two alternatives and $\mathcal{Q}_{\pparam}(\cdot)$ is measurable;
		\item For a given value of $\pparam$, there exists an $\covx$ with an open neighborhood around it in $\mathcal{S}$ s.t.
		\begin{align*}
			& U_{\pparam}(d_1,\covx) = U_{\pparam}(d_2,\covx) =\cdots= U_{\pparam}(d_D,\covx).
		\end{align*}		
	\end{enumerate}
Then $f(\pparam)$ is identified and so are  $\mathcal{Q}_{\pparam}(\cal{K})$, $\forall \mathcal{K}\ \subset \Cs$.
\end{theorem}
The assumptions of the theorem above rule out singleton (and empty) consideration sets: identification is impossible with singleton consideration sets and arbitrary dependence on preferences, because any empirical choice frequency can be explained by such consideration sets.  An alternative approach is to have one alternative -- the ``default'' -- that is always considered as the following corollary demonstrates. The identification argument exploits the response of $\Pr(d=d_j|\covx)$ to changes in $x_k$, but not the response of $\Pr(d=d_k|\covx)$ to changes in $\covx_j$. Hence, $D-1$ excluded regressors are sufficient for identification, allowing for arbitrary dependence of consideration on one (the default's) excluded regressor.
\begin{corollary}\label{Broad-alternative specific_corollary}
Suppose Assumptions \ref{IA0}, \ref{TA0}-\ref{TA2} hold, there is alternative-specific variation, the choice set contains at least three alternatives, all consideration sets contain $d_1$, and
	\begin{enumerate}
		\item  Consideration is independent of $\covx_{-1}\equiv(\covx_2,\dots,\covx_D)$: $\mathcal{Q}_{\pparam}^{\covx}(\mathcal{K})=\mathcal{Q}_{\pparam}^{\covx_1}(\mathcal{K})$, and $\mathcal{Q}_{\pparam}^{\covx_1}(\cdot)$ are measurable functions, continuous in $\covx_{1}$;
\item  The consideration of $\mathcal{K}=\{d_1\}$ is independent of $\pparam$:
$\mathcal{Q}_{\pparam}^{\covx_1}(d_1) =\mathcal{Q}^{\covx_1}(d_1)<1$, $\forall \pparam$;
\item  For a given value of $\covx_1$ and each value of $\pparam \in [0,\bar{\pparam}]$, there exists an $\covx_{-1}$ and an open neighborhood around $x=(x_1,x_{-1})$ in $\mathcal{S}$ s.t.
	\begin{align*}
		& U_{\pparam}(d_1,\covx) = U_{\pparam}(d_2,\covx) =\cdots= U_{\pparam}(d_D,\covx).
	\end{align*}
		
	\end{enumerate}
	Then $f(\pparam)$ is identified and so are  $\mathcal{Q}_{\pparam}^{\covx_{1}}(\cal{K})$, $\forall \mathcal{K} \subset \Cs$, for all $\pparam$ on the support.
\end{corollary}

Corollary \ref{Broad-alternative specific_corollary} generalizes the model of   \citet{heiss2016inattention,ho2017impact} in two dimensions. First, here each subset of the choice set containing the default has its own probability of being drawn. Second, this probability can vary with the DM's preferences as well as with the excluded regressor of the default alternative.

\subsection{Testing for limited consideration}\label{s:testing}
Since full consideration is a special case of limited consideration, it follows from the identification results above that under the SCP one can test for full consideration. The theorem below states one way of doing so \emph{without}: (1) relying on large support; (2) specifying a consideration mechanism; or (3) invoking the independence assumptions \ref{IA2} and \ref{IA3}.
\begin{proposition}\label{T:Test}
	Suppose Assumptions \ref{IA0}, \ref{TA0}-\ref{TA2} hold. Suppose there exist $\covx,\covx' \in \mathcal{S}$, and sets $\Lc, \Lc' \subset \Dc$ s.t. for some $\pparam^* \in [0,\bar{\pparam}]$
	\begin{enumerate}
		\item $\argmax_{j \in \Dc} \util_\pparam(d_j,x) \in \Lc$, $\forall \pparam \in [0,\pparam^*)$,   and   $\argmax_{j \in \Dc} \util_\pparam(d_j,x) \in \Dc \setminus \Lc$, $\forall \pparam \in (\pparam^*,\bar \pparam]$
		\item $\argmax_{j \in \Dc} \util_\pparam(d_j,x') \in \Lc'$, $\forall \pparam \in [0,\pparam^*)$,   and   $\argmax_{j \in \Dc} \util_\pparam(d_j,x') \in \Dc \setminus \Lc'$, $\forall \pparam \in (\pparam^*,\bar \pparam]$
	\end{enumerate}	
	If $\Pr(d\in \Lc|x)\neq \Pr(d\in \Lc'|x')$, then there is limited consideration.
\end{proposition}

Condition \emph{1} of the theorem requires that, given $\covx$, the first-best alternative belongs to $\Lc$ for all DMs with $\pparam <\pparam^*$ and to $\Dc\setminus\Lc$ for all DMs with $\pparam > \pparam^*$. Condition \emph{2} is the identical requirement, but given $\covx'$ and stated for $\Lc'$. Under these conditions and full consideration, the probability of choosing an alternative in $\Lc$ or, respectively, $\Lc'$ should be $F(\pparam^*)$ in both cases. Thus, if $\Pr(d\in \Lc|x)\neq \Pr(d\in \Lc'|x')$, then there is a limited consideration mechanism pushing DMs' choices away from $\Lc$ and $\Lc'$ at different rates.

\subsection{The $\modA$ Model}\label{S:ARC}
We now introduce a specific consideration mechanism, while maintaining the preference structure, including the SCP, from Section \ref{S:full}.  We refer to this model as the Alternative-specific Random Consideration ($\modA$) model \citep{manski1977structure,man:mar14}.  Each alternative $\dor_j$ appears in the consideration set with probability $\aparam_j$ independently of other alternatives. For now, we assume that these probabilities do not depend on DMs' preferences or the excluded regressor. Once the consideration set is drawn, the DM chooses the best alternative according to her preferences. To avoid empty consideration sets, following \cite{manski1977structure}, we assume that at least one alternative whose identity is unknown to the researcher is always considered.\footnote{In the previous version of this paper \citep{BaMoTh19} this completion rule is called Preferred Option(s). There we also provide identification results for other completion rules,  including Coin Toss (if the empty consideration set is drawn, the DM randomly uniformly picks one alternative from the choice set, i.e. each alternative has probability $1/D$ of being chosen), Default Option (there is a preset alternative that is chosen if the empty set is drawn), and Outside Option (the DM exits the market if the empty set is drawn).}
\phantomsection
\begin{assumptionSP}{ARC}[The Basic ARC Model]\label{TA3} The probability that the consideration set takes realization $\Ks$ is
	\begin{align*}
		\mathcal{Q}(\Ks) & \equiv \prod_{k \in \Kc }\aparam_k \prod_{k \in \Dc \setminus \Kc} (1-\aparam_k), \quad \forall \Ks \subset \Cs,
	\end{align*}
	where $\aparam_j>0, \, \forall j$, and $\exists d^*$ s.t. $\aparam_{d^*}=1$.
\end{assumptionSP}
By assuming $\aparam_{j}>0$, we omit never-considered alternatives from the choice problem. Since a never-considered alternative is never compared to any other alternative, whether it is in the choice set or not does not affect the DM's problem. Hence, never-considered alternatives have no impact on what we can learn about preferences.

Under the assumptions of Theorem \ref{T:Benchmark}, identification attains. Notably, each consideration parameter $\aparam_{j}$ is identified (as long as $d_j$ is chosen with positive probability at some $\covx$).
\begin{theorem}\label{T:ARC_independent}
	Suppose Assumptions \ref{IA0}, \ref{IA2}-\ref{IA3}, \ref{TA0}-\ref{TA2}, \ref{TA3} hold, and Assumption \ref{IA1} holds for $\covxs \subset \mathcal{S}$ s.t. $\forall \covx\in \covxs$
	\[
	U_{\pparam}(d_1,x)>U_{\pparam}(d_j,\covx) \Rightarrow U_{\pparam}(d_1,\covx)>U_{\pparam}(d_{j+1},\covx), \quad \forall j>1.
	\]
	Then $f(\cdot)$ is identified and so are $\aparam_{1}$ and $\aparam_{2}$. In addition, if $\Pr(d = d_j|\covx) \neq 0$ for some $x$, then $\aparam_j$ is identified.
\end{theorem}

\subsubsection{Preference-Dependent Consideration}\label{sec:gendependence}

Returning to our example with three alternatives, recall that we have an additional moment $\Pr(d=d_3|\covx)$. The information it provides allows us to identify
some forms of dependence between consideration and preferences, i.e., to relax Assumption \ref{IA2}. To see how, suppose $d_2$ is always considered. Then, with preference dependence, the choice frequencies become:
\begin{align*}
	\Pr(d=d_1|x) = \int_{0}                       ^{\cmap_{1,2}(x)}\aparam_1(\pparam) dF \quad \text{and} \quad
	\Pr(d=d_3|x) = \int_{\cmap_{2,3}(x)}^{\bar{\pparam}}\aparam_3(\pparam) dF.
\end{align*}
The ratio of the derivatives of these two moments yields $\frac{\aparam_1(\pparam)}{\aparam_3(\pparam)}$.
More assumptions are required to obtain point identification of the $\aparam_{j}(\pparam)$'s.  In Section \ref{S:3Alt} we provided identification results for Binary Consideration Types. Here, leveraging the additional structure provided by the ARC model, we can allow for more flexible dependence between consideration and preferences. We do so through a \emph{proportionally shifting} consideration mechanism, formally defined below. This mechanism may arise when there is a cost to evaluate each alternative. In such a case, the DMs may consider alternatives that they ex-ante deem more aligned with their preferences (e.g., the DM's consideration shifts away from riskier to safer alternatives as her risk aversion increases).

\begin{assumptionSP}{ARC.P}[ARC with Proportional Consideration]\label{TA3b}
	The consideration mechanism follows the $\modA$ model with $\{D\geq 4 ~ \& ~ 1\leq d^*\leq D\}$ or $\{D=3 ~ \& ~ d^*=2\}$, and\[
	\aparam_j(\pparam) =
	\begin{cases}
	\aparam_j(1 - \alpha(\pparam)) & \text{if } j < d^*  \\
	1							   & \text{if } j = d^*  \\
	\aparam_j(1 + \alpha(\pparam)) & \text{if } j > d^*
	\end{cases}
	\]
   s.t. $\alpha(\cdot)$ is differentiable a.e., $\alpha'(\cdot) \neq 0$ a.e., $\alpha(\bar{\pparam}) = 0$, $0<\aparam_j(\pparam)< 1, ~ \forall j \neq d^*$, $\forall\pparam \in [0,\bar \pparam]$.	
\end{assumptionSP}

In the case with three alternatives, $\frac{\aparam_1(\pparam)}{\aparam_3(\pparam)}=\frac{\aparam_1(1 - \alpha(\pparam))}{\aparam_3(1 + \alpha(\pparam))} $.   From this, $\frac{\aparam_1}{\aparam_3}$ is identified when $\covx$ and $\covx'$ are chosen such that $\cmap_{1,2}(\covx) = \cmap_{2,3}(\covx') =\bar{\pparam}$. Once $\frac{\aparam_1}{\aparam_3}$ is identified, $\frac{1-\alpha(\pparam)}{1+\alpha(\pparam)}$ is known for all $\pparam$; hence, $\alpha(\pparam)$ can be solved for. Identification of $f(\pparam)$ follows from substituting $\alpha(\pparam)$ into the expression for $\frac{\Pr(d=d_1|x)}{dx}$. The theorem below generalizes this argument.
\begin{definition}{}(No Three Way Ties)\label{NoTWT} For a given $\covx$, there are no-three way ties if $\nexists \pparam \in [0,\bar{\pparam}]$ and $\{j,k,l\}$ s.t. $U(d_j,x)=U(d_k,x)=U(d_l,x)$.
\end{definition}
\phantomsection
\begin{theorem}\label{T:ARC_dependent_1}
	Suppose Assumptions \ref{IA0}, \ref{IA3}, \ref{TA0}-\ref{TA2}, \ref{TA3b} hold, and Assumption \ref{IA1} holds for $\covxs$ s.t. $\forall \covx\in \covxs$ there are no three-way ties and
	\begin{align*}
	U_{\pparam}(d_1,x)>U_{\pparam}(d_j,x) &\Rightarrow U_{\pparam}(d_1,x)>U_{\pparam}(d_{j+1},x), \quad \forall j>1,\\
	U_{\pparam}(d_D,x)>U_{\pparam}(d_j,x) &\Rightarrow U_{\pparam}(d_D,x)>U_{\pparam}(d_{j-1},x), \quad \forall j<D,
	\end{align*}
	and $\exists x \in \covxs$ s.t. $\cmap_{j,k}(x)\leq0, ~ \forall j,k, ~j<k$.

Then $f(\cdot)$ and $\{\aparam_{j}(\cdot)\}_{j=1}^D$ are identified.
\end{theorem}
The conditions of the theorem are stronger than in Theorem \ref{T:ARC_independent}, as they impose relative order of the cutoffs not only for alternative $d_1$ but also for $d_D$. In many cases, the relative order of $\cmap_{1,j}(x)$'s alone is sufficient, for example when $d_1$ is always considered.

\subsubsection{Identification with Alternative-specific Excluded Regressors}\label{Extreme Variation}
By leveraging features of the ARC model, the identification results in Section \ref{subsec:altspec} can be extended to the case where the consideration of $d_j$ is a function both of $\covx_j$ and preferences.  This differs from Corollary \ref{Broad-alternative specific_corollary}, which restricted the consideration of alternative $d_j$ to depend  only on the default alternative's excluded regressor.  We continue to assume that the choice set is $\{d_1,d_2,d_3\}$ and that $d_2$ is always considered. Let the consideration of $d_j$ be a measurable function of $\covx_j$ and $\nu$,  continuous in its first argument:  $\aparam_{j} =\aparam_j(\covx_j,\pparam)$.  Similar to the example in Section \ref{subsec:altspec}, we construct two vectors, $\covx^i$ and $\covx^{ii}$, such that: \emph{(i)}  $\pparam = \cmap_{1,2}(\covx^i)<\cmap_{1,3}(\covx^i)<\cmap_{2,3}(\covx^i)$; and \emph{(ii)} $\cmap_{2,3}(\covx^{ii})<\cmap_{1,3}(\covx^{ii})<\cmap_{1,2}(\covx^{ii}) = \pparam$.
 The derivative of the choice frequency of $d_1$ with respect to $\covx_2$ for these cases are,  respectively:
\begin{align}\label{eq:richvar}
	\textcolor{white}{i} (i) \quad \frac{\partial \Pr(d=d_1|x)}{\partial x_2}
	&=\textcolor{white}{(1-\aparam_3(\covx_3,\pparam))}\aparam_1(x_{1},\pparam) f( \pparam) \frac{\partial \cmap_{1,2}(x) }{\partial x_2};\\
	(ii) \quad \frac{\partial \Pr(d=d_1|x)}{\partial x_2} &=(1-\aparam_3(\covx_3,\pparam))\aparam_1(\covx_1,\pparam)f(\pparam ) \frac{\partial \cmap_{1,2}(x) }{\partial x_2}. \nonumber
\end{align}
The ratio of the expressions in Equation \eqref{eq:richvar} identifies $\aparam_3(\covx_3,\pparam)$. Using a similar logic, we can identify $\aparam_1(\covx_1,\pparam)$. Plugging these consideration probabilities into Equation \eqref{eq:richvar} identifies $f(\pparam)$. In sum, alternative-specific variation yields identification without large support and without the independence Assumptions \ref{IA2} and \ref{IA3}. It is also possible to allow consideration of $d_1$ (and $d_3$) to depend on $x_1$, $\pparam$, as well as $x_3$. The key exclusion restriction in this case is that the consideration of $d_2$ is independent of all components of $\covx$. Our last identification result generalizes this example.

\begin{assumptionSP}{ARC.AS}\label{TA3.AS}
	The consideration mechanism follows the $\modA$ model. The consideration probability of each alternative $d_j$ is a measurable function of $\covx_j$ and preferences: $\aparam_{j}= \aparam_{j}(\covx_j,\pparam)$, continuous in the first argument. Default alternative $d^*$ is s.t. $\aparam_{d^*}(\covx_{d^*},\pparam)=1$ for all $\covx_{d^*} \in\mathcal{S}$ and for all $\pparam \in [0,\bar \pparam]$.
\end{assumptionSP}

\begin{theorem} \label{Broad-alternative specific_ARC}
	Suppose Assumptions \ref{IA0}, \ref{TA0}-\ref{TA2}, \ref{TA3.AS} hold. Suppose there is alternative-specific variation and the choice sets contain at least three alternatives. Suppose for a given value of $\pparam$ there exists an $\covx=(\covx_1,\covx_2,\dots,\covx_D)$, and an open neighborhood around it in $\mathcal{S}$, s.t.
\[
		U_{\pparam}(d_1,\covx) = U_{\pparam}(d_2,\covx) =\cdots= U_{\pparam}(d_D,\covx).
\]
	Then $f(\pparam)$ and $\{\aparam_{j}(\covx_j,\pparam)\}_{j=1}^{D}$ are identified.
\end{theorem}

Existing identification results that rely on alternate-specific variation \citep{Goeree2008,Abaluck2019,kawaguchi2019designing} allow for consideration dependence on its own regressor, but not preferences. Theorem \ref{Broad-alternative specific_ARC} states identification for a general version of the ARC model where the alternative-specific consideration probability can depend on both its own regressor and DMs' preferences.

\section{Likelihood and Tractability}\label{S:LikeTract}
We now turn to the computational aspects of limited consideration models under the SCP and, in particular, of their likelihood function. Consider a generic consideration mechanism.

A computationally appealing way to write the likelihood function is to determine the probability that a DM with preference parameter $\pparam$ chooses alternative $\dor_j$ conditional on $\covx$. Alternative $\dor_{ j}$ is chosen if and only if $d_j$ is in the consideration set and every alternative that is preferred to $d_j$ is not. Denote the set of alternatives that are preferred to $\dor_{j}$ by
\begin{align*}
	\Bc_{\pparam}(\dor_j,x)\equiv\{k :\,\util_\pparam(\dor_k,\covx) > \util_\pparam(\dor_j,\covx)\}.
\end{align*}
Then,
\begin{align}\label{basic_p}
	\Pr(\dor_j | \covx) &= \int \Pr(\dor_j | \covx, \pparam)dF=  \int \mathcal{O}_{\pparam}^{\covx}(d_j;\Bc_{\pparam}(\dor_j,x))dF,
\end{align}

The object on the RHS does not require evaluating the utility of each alternative within each possible consideration set. In fact,  $U_{\pparam}(d_j,x)$ needs to be computed only once for each $\pparam$, $d_j$, and $\covx$ to create $\Bc_{\pparam}(\dor_j,x)$, which does not vary with the consideration set. Hence the computational complexity lies in the mapping from $\mathcal{O}_{\pparam}^{\covx}(\cdot)$ to the parameters governing the consideration mechanism. This, however, may not even require enumerating all possible consideration sets. To demonstrate this with a concrete example, we proceed with the basic ARC model. In this case, the RHS of Equation \eqref{basic_p} is:
\begin{align}\label{ARClike}
	I(\dor_j|\covx)&\equiv \aparam_j\int \prod_{k \in \Bc_{\pparam}(\dor_j,x)}(1-\aparam_k)dF.
\end{align}
Given $\{\aparam_j\}_{j=1}^D$, the integrand $\prod_{k\in \Bc_{\pparam}(\dor_j,x)}(1-\aparam_k)$ is piecewise constant in $\pparam$ with at most $D-1$ breakpoints, corresponding to indifference points between alternatives $j$ and $k$, i.e., $\cmap_{j,k}(x)$, that are computed only once for each observed $\covx$. There are at least two methods to compute this integral. First, for every $d_j$ and $\covx$, we can directly compute the breakpoints and hence write $I(\dor_j|\covx)$ as a weighted sum:
\begin{align*}
	I(\dor_j|\covx)&=\aparam_j\sum_{h=0}^{D-1} \left( (F(\pparam_{h+1})-F(\pparam_{h}))\prod_{k \in \Bc_{\pparam_h}(\dor_j,x)}(1-\aparam_k) \right),
\end{align*}
where $\pparam_{h}$'s are the sequentially ordered breakpoints augmented by the integration endpoints: $\pparam_{0}=0$ and $\pparam_{D}=\bar{\pparam}$.   This expression is trivial to evaluate given $F(\cdot)$ and breakpoints $\{\pparam_{h}\}_{h=0}^{\Cn}$. More importantly, since the breakpoints are invariant with respect to the consideration probabilities, they are computed only once for each $\covx$. This simplifies the likelihood maximization routine by orders of magnitude, as each evaluation of the objective function involves a summation over products with at most $\Cn$ terms. A second approach is to compute $I(\dor_j|\covx)$ using Riemann approximation:
\begin{align*}
	I(\dor_j|\covx) &\approx \aparam_j \frac{\bar{\pparam}}{M} \sum_{m=1}^{M}\left(f({\pparam_m})\prod_{k \in \Bc_{\pparam_m}(d_j,x)}(1-\aparam_k)\right),
\end{align*}
where $M$ is the number of intervals in the approximating sum, $\frac{\bar{\pparam}}{M}$ is the intervals' length, $\pparam_{m}$'s are the intervals' midpoints, and $f(\cdot)$ is the density of $F(\cdot)$. Again, one does not need to evaluate the utility from different alternatives in the likelihood maximization. Instead, one \emph{a priori} computes the utility rankings for each $\pparam_m$, $m=1,\dots,M$. These rankings determine $\Bc_{\pparam_m}(d_j,x)$.  The likelihood maximization is now a standard search routine over $\{\aparam_j\}_{j=1}^{D}$ and $f(\cdot)$. Our theory restricts $f(\cdot)$ to the class of continuous and strictly positive functions.  In practice, the search is over a class of non-parametric estimators for $f(\cdot)$.\footnote{One could use a mixture of Beta distributions \citep{ghosal2001}, as we do in Section \ref{sec:appplication}.} If the density is parameterized, i.e., $f(\pparam_m)\equiv f(\pparam_m;\theta^f)$, then the maximization is over $\{\aparam_j\}_{j=1}^{D}$ and $\theta^f$. Finally, the interval midpoints are the same across all DMs as they do not depend on $\covx$, further reducing computational burden.\footnote{Depending on the class of $f(\cdot)$, it may be more accurate to compute $I(\dor_j|\covx)$ by substituting $\frac{\bar{\pparam}}{M}f({\pparam_m})$ with $F(\overline{\pparam}_m)-F(\underline{\pparam}_m)$, where $\overline{\pparam}_m$ and $\underline{\pparam}_m$ are the endpoints of the corresponding interval.}

Allowing consideration to depend on preferences (or on $\covx$) introduces only minimal adjustments to the likelihood function. For example, let each consideration function be parameterized by $\theta_j$: $\aparam_{j}(\pparam)\equiv\aparam_{j}(\pparam;\theta_{j})$. Then, at each $\pparam$, we can substitute $\aparam_{j}$ with the corresponding $\aparam_{j}(\pparam;\theta_{j})$, and the likelihood maximization is now over $\{\theta_j\}_{j=1}^{D}$ and $\theta^f$. Given the desired level of parameterization -- i.e., the dimensionality of the parameter vectors $\theta_{j}$ and $\theta^{f}$ -- the computational complexity of the problem grows polynomially in $D$.

As a final remark, if alternative $d_j$ is never chosen, then one can conduct estimation as if $d_j$ were not in the choice set. Indeed, per Equation \eqref{ARClike}, $\aparam_j$ contributes positively to the likelihood if and only if alternative $d_j$ is chosen. When it is never chosen, it may only enter via the term $(1-\aparam_j)$; hence, the likelihood will be maximized by setting $\aparam_j=0$. Therefore, setting $\aparam_{j}=0$ for all zero-share alternatives, regardless of why they were not chosen, has no impact on estimation. This too may speed up estimation.
\section{Limited Consideration and RUM: A Comparison}\label{sec:modprop}

We focus on a \emph{standard} application of the RUM with full consideration in the context of our example in Section \ref{sec:mainexample}.  The final evaluation of the utility that the DM derives from alternative $j$ now includes a separately additive error term:
\begin{align}
	V_{\pparam}(L_j(x)) = U_{\pparam}(L_j(x))+\varepsilon_j,\label{eq:RUM}
\end{align}
where, as before, $\pparam$ captures unobserved heterogeneity in preferences, and $\varepsilon_j$ is assumed independent of the random coefficients (in this application, $\pparam$).

Typical implementations of this model further specify that $\varepsilon_j$ is i.i.d. across alternatives (and DMs) with a Type 1 Extreme Value distribution, following the seminal work of \cite{Mcfadden1974}. This yields a Mixed Logit that is distinct from the commonly used one in \cite{McFaddenTrain2000}.  In their model, random coefficient(s) enter the utility function linearly, while in the context of expected utility they enter nonlinearly.  We now discuss two properties of the Mixed Logit that hinder its applicability in our context.

\subsection{Monotonicity}
Coupling utility functions in the hyperbolic absolute risk aversion (HARA) family, for example CARA or CRRA, with a Type 1 Extreme Value distributed additive error yields:
\begin{proposition}\citep[Non-monotonicity in RUM,][]{ape:bal18,Wilcox08}
	In Model \eqref{eq:RUM} with HARA preferences and $\varepsilon_j$ i.i.d. Type 1 Extreme Value, as the DM's risk aversion increases, the probability that she chooses a riskier alternative declines at first but eventually starts to increase.
\end{proposition}
To see why, consider two non-dominated alternatives $\dor_j$ and $\dor_k$ such that $\dor_j$ is riskier than $\dor_k$. A risk neutral DM prefers $\dor_j$ to $\dor_k$, and hence will choose the former with higher probability. As risk aversion increases, the DM eventually becomes indifferent between $\dor_j$ and $\dor_k$ and chooses either of these alternatives with equal probability. As risk aversion increases further, she prefers $\dor_k$ to $\dor_j$ and chooses the latter with lower probability. However, as risk aversion gets even larger, the expected utility under HARA of any lottery with finite stakes converges to zero. Consequently, the choice probabilities of all alternatives, regardless of their riskiness, converge to a common value.\footnote{Recall that in the Mixed Logit the magnitude of the utility differences is tied to differences in (log) choice probabilities, $
U_{\pparam}(L_k(x))-U_{\pparam}(L_j(x))=\log(\Pr(d=d_k|x,\pparam))-\log(\Pr(d=d_j|x,\pparam))$, so that as $\nu\to\infty$ the choice probabilities are predicted to be all equal.} Hence, at some point the probability of choosing $\dor_j$ is increasing in risk aversion.

To the contrary, our model with a limited consideration mechanism that is independent of preferences yields choice probabilities that are monotone in the preference parameter.
\begin{property}[Generalized Preference Monotonicity]\label{prop:gendom}
	A model satisfies generalized preference monotonicity if for any $\pparam_1 < \pparam_2$ and $J  \in \{1,2,\dots,\Cn\}$:
	\begin{align*}
		\Pr\left(\bigcup_{j=1}^J \dor_j \bigg|\covx,\pparam_1\right) \geq  \Pr\left(\bigcup_{j=1}^J \dor_j  \bigg|\covx,\pparam_2\right).
	\end{align*}
\end{property}
In the context of risk preferences, Property \ref{prop:gendom} states that the probability of choosing one of the $J$ riskiest alternatives declines as $\pparam$ increases. Since Property \ref{prop:gendom} is satisfied for any choice set under the SCP and full consideration, it is also satisfied under limited consideration:
\begin{proposition}\label{prop:monotonicity}
	A model that satisfies the SCP (i.e., Assumption \ref{TA2}) and Assumption \ref{IA2} satisfies Generalized Preference Monotonicity.
\end{proposition}
\subsection{Generalized Dominance}
Next, we establish the relation between utility differences across two alternatives and their respective choice probabilities. Because our random expected utility model features unobserved preference heterogeneity, we work with an analog of the rank order property in \cite{Manski75} that is conditional on $\pparam$:
\begin{definition} (Conditional Rank Order of Choice Probabilities)
	The model yields conditional rank order of the choice probabilities if for given $\nu$ and alternatives $j,k\in \cal{D}$,	
	\begin{align*}
		U_{\pparam}(L_j(x))>U_{\pparam}(L_k(x)) &\Rightarrow \Pr(d=d_j|x,\pparam)>\Pr(d=d_k|x,\pparam).
	\end{align*}
\end{definition}
We show that the conditional rank order property implies the following upper bound on the probability that suboptimal alternatives are chosen.
\begin{property} (Generalized Dominance)\label{property:GD} A model satisfies Generalized Dominance if for any $\covx$, $d_j$, and set $\mathcal{K} \subset {\cal{D}}\setminus \{d_j\}$ s.t. alternative $d_j$ is never-the-first-best in $\mathcal{K} \cup \{d_j\}$
	\[\Pr(d=d_j|x)<\sum_{k \in \mathcal{K}}\Pr(d=d_{k}|x).
	\]
\end{property}
Generalized Dominance holds in the Mixed Logit model and, more broadly, in models that satisfy the conditional rank order property. However, it may not hold in some limited consideration models. For example, Generalized Dominance is violated if $d_j$ is never-the-first best among $\{d_j,d_k,d_l\}$, is almost always considered,  and alternatives $d_k$ and $d_l$ are rarely considered.

\subsection{Limited Consideration as Ordinal RUM}
In the Mixed Logit, the cardinality of the differences in the (random) expected utility of alternatives plays a crucial role in the determination of choice probabilities, as it interacts with the realization of the additive error. In contrast, in models that satisfy the SCP, the DMs' choices are determined by the ordinal expected utility ranking of the alternatives. Hence, limited consideration models can be recast as Ordinal Random Utility models (ORUM), where the key departure from standard RUMs is the distribution of the additive error term.

\begin{proposition}(Limited Consideration as ORUM)\label{fact:modGeneric}
	A Limited Consideration Model is equivalent to an additive error random utility model with unobserved preference heterogeneity where all alternatives are considered, the DM's utility of each alternative $d_j \in \mathcal{D}$ is given by\[
	V_{\pparam}(d_j,x)=U_{\pparam}(d_j,x)+\varepsilon_{j}(\pparam,x),
	\]
	and $\vec{\varepsilon}(\pparam,x)=\{\varepsilon_{j}(\pparam,x)\}_{j=1}^D$ is distributed on $\{-\infty,0\}^D$ according to $\Pr\left(\vec{\varepsilon}(\pparam,x)\right)= \mathcal{Q}_{\pparam}^{\covx}(\cal{K})$ for $\cal{K}$ s.t. $d_j \in \mathcal{K}$  if  $\varepsilon_j(\pparam,x) = 0$ and  $d_k \in \mathcal{D} \setminus \mathcal{K}$ if $\varepsilon_j(\pparam,x) = -\infty$.
\end{proposition}
Casting our limited consideration model as an ORUM clearly demonstrates its flexibility.  In particular, our results show how to obtain identification when the errors are correlated with the excluded regressors, the preference parameter, and across alternatives.

\begin{table}\centering
\scriptsize
\begin{threeparttable}
\caption{Model Comparisons}
\label{T:modelprops}
	\begin{tabular}{l c c c c c}
		
		\hline \hline \\[-0.5em]
		Error Structure                           &  Mixed Logit           & Basic ARC     &  Binary Types    & Prop. Shifting      &  Generic Consideration  \\[1 ex]
		\hline \\[-0.5em]		
		Support		                                 &$\mathbbm{R}$& $\{-\infty,0\}$& $\{-\infty,0\}$ &$\{-\infty,0\}$&$\{-\infty,0\}$\\[2ex]
		Independent of $x$	                         &Yes		 & Yes	 & Yes$^1$	 &    Yes        & No	  \\[1.5ex]
		Independent of $\pparam$	                 &Yes		 & Yes   & No$^1$	 &    No        & No	  \\[1.5ex]
		Independent across alternatives		         &Yes        & Yes   & Yes   &    Yes       & No     \\[1.5ex]
		Identical   across alternatives		         &Yes        & No    & No    &    No        & No     \\[0.5em]
		\hline \hline
	\end{tabular}
\caption*{\footnotesize  \begin{flushleft} Note 1: Binary Types can also be dependent on $\covx$ but would require independence with $\pparam$ in that case.  \\
\end{flushleft}}
\end{threeparttable}
\end{table}

We conclude this section with Table \ref{T:modelprops}, listing the differences across the Mixed Logit and limited consideration models. The first two columns summarize the differences between the basic $\modA$ model and the Mixed Logit. The third column and fourth column remind the reader our two models with consideration depending on preferences.  Finally, the last column highlights the fact that with alternative-specific variation we may also have dependence of the error term on the excluded regressor(s) as well as on the preference parameter.

\section{Application}\label{sec:appplication}
We offer an empirical analysis of households' decisions under risk. This analysis aims to illustrate how our method works and its ability to fit the data.
\subsection{\label{Data}Data}
We study households' deductible choices across three lines of property insurance: auto collision, auto comprehensive, and home all perils. The data come from a U.S. insurance company. Our analysis uses a sample of 7,736 households who purchased their auto and home policies for the first time between 2003 and 2007 and within six months of each other.\footnote{The dataset is an updated version of the one used in \cite{Barseghyan2013}. It contains information for an additional year of data and puts stricter restrictions on the timing of purchases across different lines. These restrictions are meant to minimize potential biases stemming from non-active choices, such as policy renewals, and temporal changes in socioeconomic conditions.} Table \ref{T:descripive} provides descriptive statistics for households' observable characteristics, which we use later to estimate households' preferences.\footnote{These are the same variables that are used in \cite{Barseghyan2013} to control for households' characteristics. See discussion there for additional details.} We observe the exact menu of alternatives available at the time of the purchase for each household and each line of coverage. The deductible alternatives vary across lines of coverage but not across households. Table \ref{T:deductible} presents the frequency of chosen deductibles in our data.

\begin{table}\centering
\scriptsize
\begin{threeparttable}
\caption{Premium Quantiles for the $\$ 500$ Deductible}
\label{T:premium500}
	\begin{tabular}{l c c c c c c c }
		\hline \hline \\[-1.5ex]
		Quantiles      & 0.01   & 0.05  & 0.25 & 0.50 & 0.75 & 0.95 & 0.99\\\\[-1.5ex]
		\hline \\[-1.5ex]
		Collision      &  53   & 74   & 117 & 162 & 227 & 383 & 565 \\\\[-1.5ex]
		Comprehensive  &  29   & 41   & 69  & 99  & 141 & 242 & 427\\\\[-1.5ex]
		Home 		   &  211  &305   & 420 & 540 & 743 & 1,449& 2,524\\\\[-1.5ex] 		
		\hline
	\end{tabular}
\end{threeparttable}
\end{table}

Premiums are set coverage-by-coverage as in the example from Section \ref{sec:mainexample}. Table \ref{T:premium} reports the average premium by context and deductible, and Table \ref{T:premium500} summarizes the premium distributions for the $\$500$ deductible. Premiums vary dramatically.  The 99th percentile of the $\$500$ deductible is more than ten times the corresponding 1st percentile in each line of coverage.

Claim probabilities stem from \citep{barseghyan2018different}, who derived them using coverage-by-coverage Poisson-Gamma Bayesian credibility models applied to a large auxiliary panel. Predicted claim probabilities (summarized in Table \ref{T:claim}) exhibit extreme variation: The 99th percentile claim probability in collision (comprehensive and home) is 4.3 (12 and 7.6) times higher than the corresponding 1st percentile. Finally, the correlation between claim probabilities and premiums for the $\$500$ deductible is 0.38 for collision, 0.15 for comprehensive, and 0.11 for home all perils. Hence, there is independent variation in both.
\begin{normalsize}
\begin{table}\centering
\scriptsize
\begin{threeparttable}
\caption{Claim Probabilities Across Contexts}
\label{T:claim}
	\begin{tabular}{lc c c c c c c}
        \hline \hline \\[-1.5ex]
		Quantiles      & 0.01   & 0.05 & 0.25   & 0.50   & 0.75 & 0.95   & 0.99  \\\\[-1.5ex]
		\hline \\[-1.5ex]
		Collision      & 0.036 & 0.045 & 0.062 & 0.077 &  0.096 & 0.128 & 0.156 \\\\[-1.5ex]
		Comprehensive  & 0.005 & 0.008 & 0.014 & 0.021 &  0.030 & 0.045 & 0.062 \\\\[-1.5ex]
		Home 		   & 0.024 & 0.032 & 0.048 &0.064 &  0.084 &0.130 &0.183 \\\\[-1.5ex] 		
		\hline
	\end{tabular}
\end{threeparttable}
\end{table}
\end{normalsize} 
\subsection{Estimation Results}
\subsubsection{The basic $\modA$ Model: Collision}\label{subsec:collision}
We start by presenting estimation results in a simple setting where the only choice is the collision deductible and observable demographics do not affect preferences. To execute our estimation procedure we set $\bar{\pparam}=0.02$, which is conservative \citep[see][]{Barseghyan2016}. We \emph{ex post} verify that this does not affect our estimation by checking that the density of the estimated distribution is close to zero at the upper bound. We approximate $F(\cdot)$ non-parametrically through a mixture of Beta distributions. In practice, however, both AIC/BIC criteria indicate that a single component is sufficient for our analysis, resulting in a total of seven parameters to be estimated.  We let the data speak to the identity of the always-considered alternative.\footnote{In fact, the estimation is run under the Coin Toss completion rule that nests the possibility that any alternative can be always considered. The data chooses $\aparam_{1000} = 1$.}

The estimated distribution and consideration parameters are reported in Table \ref{T:mleARCcoll}. As the first panel in Figure \ref{F:ARCagg} shows, the model closely matches the aggregate moments observed in the data. The second panel in Figure \ref{F:ARCagg} illustrates side-by-side the frequency of predicted choices, consideration probabilities, and the distribution of households' first-best alternatives (i.e., the distribution of optimal choices under full consideration). Predicted choices are determined jointly by the preference induced ranking of deductibles and by the consideration probabilities: Limited consideration forces households' decision towards less desirable outcomes by stochastically eliminating better alternatives. The two highest deductibles ($\$1000$ and $\$500$) are considered at much higher frequency (1.00 and 0.92, respectively) than the other alternatives, suggesting that households have a tendency to regularly pay attention to the cheaper items in the choice set. Yet, the most frequent model-implied optimal choice under full consideration is the $\$250$ deductible, which is considered with low probability.
\begin{figure}\centering
	\caption{The $\modA$ Model}
	\label{F:ARCagg}
	\includegraphics[trim={5cm 0 0 0.5cm},scale=0.3]{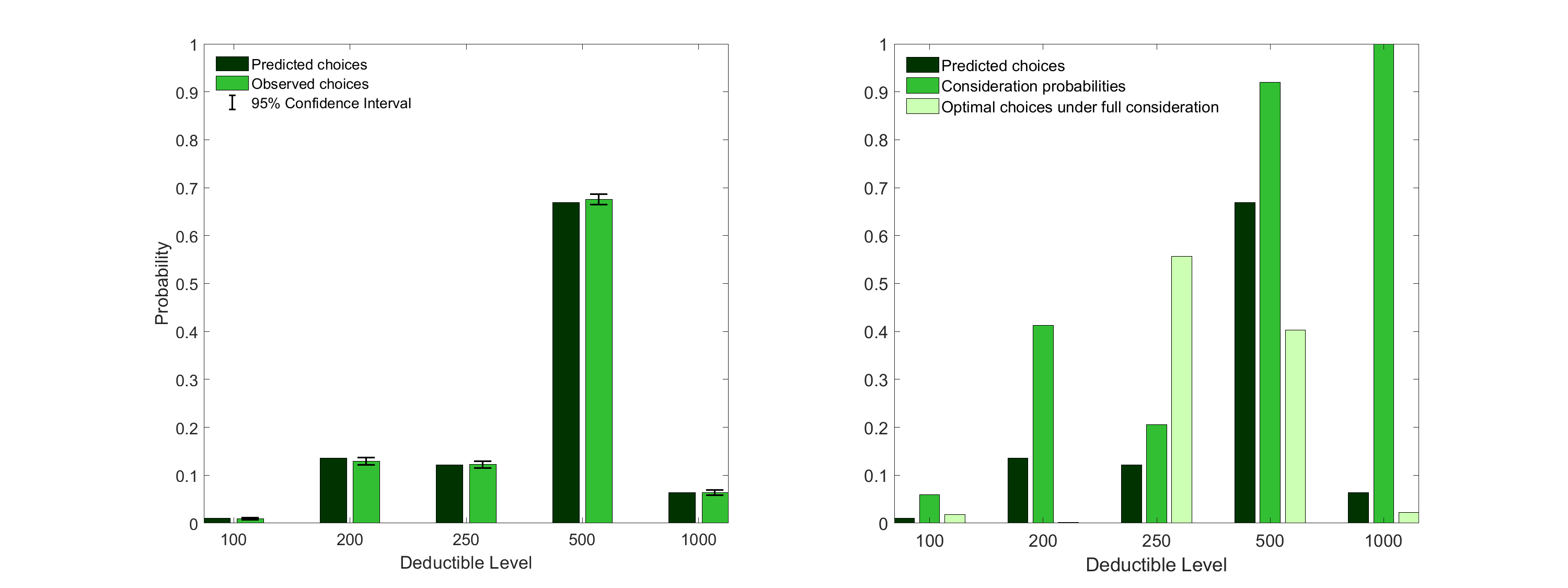}
	\captionsetup{font={footnotesize,sf},justification=raggedright}
	\caption*{The first panel reports the distribution of predicted and observed choices.  The second panel displays consideration probabilities and the distribution of optimal choices under full consideration.}
\end{figure}
In this application, assuming full consideration leads to a significant downward bias in the estimation of the underlying risk preferences. To see why, consider increasing the consideration probabilities for the lower deductibles to the same levels as the $\$500$ deductible. Holding risk preferences fixed, the likelihood that the lower deductibles are chosen increases and therefore the higher deductibles are chosen with lower probability. Average risk aversion must decline to compensate for this shift. This is exactly the pattern we find when we estimate a near-full consideration model. In particular, we find that average risk aversion decreases by about 32\% from 0.0037 to 0.0025 when all consideration parameters equal 0.9999.\footnote{We cannot assume that all consideration probabilities are equal to one, since the $\$200$ deductible is never the first best under full consideration and is chosen with positive probability.} To put these numbers into context, a DM with risk aversion equal to 0.0037 is willing to pay $\$431$ to avoid a $\$1000$ loss with probability 0.1, while a DM with risk aversion equal to 0.0025 is only willing to pay $\$300$ to avoid the loss.

The basic $\modA$ model's ability to match the data extends also to conditional moments. The first two panels of Figure \ref{F:ARCbygroup} show observed and predicted choices for the fraction of households facing low and high premiums, respectively, and the next two panels are for households facing low and high claim probabilities.\footnote{Low and high groups here are defined as households whose claim rate (or baseline price) are in the first and third terciles, respectively.} Finally, the last two panels display households who face both low claim probabilities and high prices and \emph{vice versa}. It is transparent from Figure \ref{F:ARCbygroup} that the model matches closely the observed frequency of choices across different subgroups of households facing a variety of prices and claim probabilities, even though some of these frequencies are quite different from the aggregate ones.
\begin{figure}\centering
    \caption{The $\modA$ Model: Conditional Distributions}
    \label{F:ARCbygroup}
    \includegraphics[trim={5cm 0 0 0cm},scale=0.4]{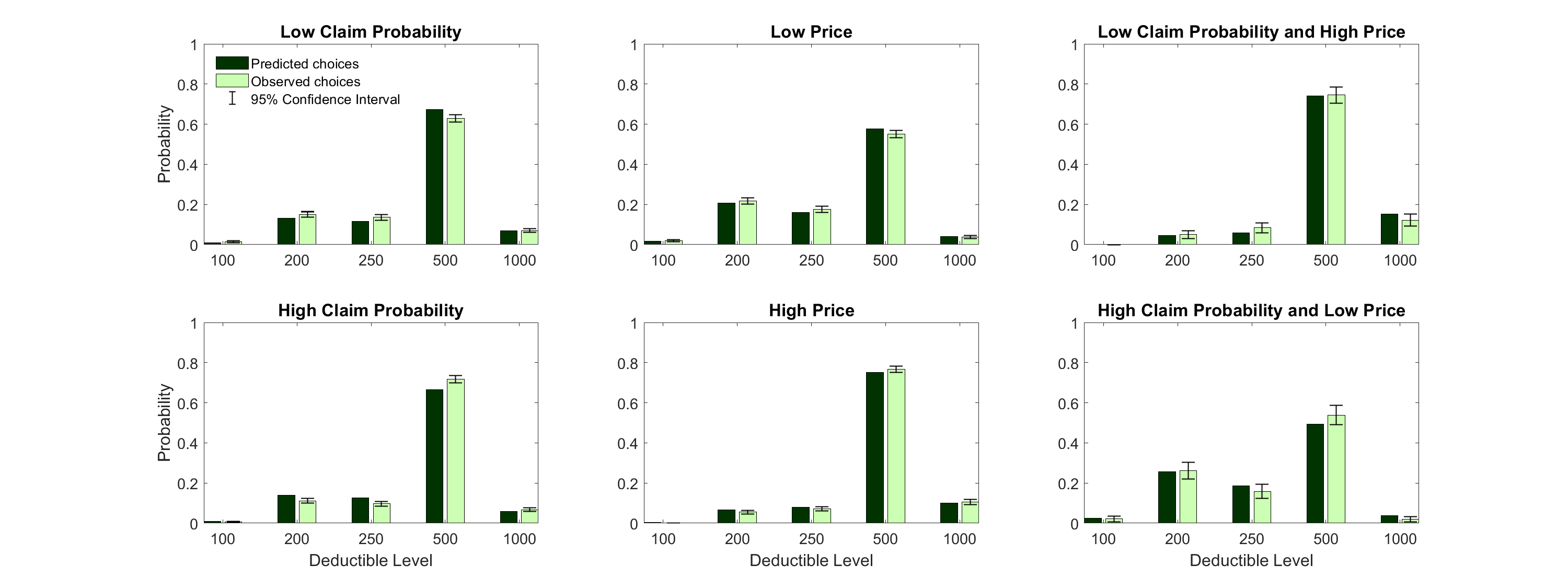}
\end{figure}

The $\modA$ model's ability to violate Generalized Dominance is key in matching the data. In our dataset, because of the pricing schedule in collision, the $\$200$ is never-the-first best among $\{\$100,\$200,\$250\}$ for $99.84\%$ of all households and $100\%$ of households who have chosen the $\$200$ deductible. It costs the same to get an additional $\$50$ of coverage by lowering the deductible from $\$250$ to $\$200$ as it does to get an additional $\$100$ of coverage by lowering the deductible from $\$200$ to $\$100$. If a household's risk aversion is sufficiently small, then it prefers the $\$250$ deductible to the $\$200$ deductible. If, on the other hand, the household's level of risk aversion is such that it would prefer the $\$200$ deductible to the $\$250$ deductible, then it would also prefer getting twice the coverage for the same increase in the premium. That is, for any level of risk aversion, the $\$200$ deductible is dominated either by the $\$100$ deductible or by the $\$250$ deductible.\footnote{This pattern is at odds not only with EUT but also many non-EU models \citep{Barseghyan2016}.} Yet, overall the $\$200$ deductible is chosen roughly as often as the $\$100$ and  $\$250$ deductibles combined. More so, for certain  sub-groups the $\$200$ deductible is chosen \emph{much more often} than the $\$100$ and $\$250$ deductible combined. It follows that a model satisfying Generalized Dominance cannot rationalize these choices.

Next we relax the assumption that demographic variables, $\mathbf{Z}$, do not influence risk preferences. In particular, conditional on demographics, preferences are distributed Beta($\beta_{1}(\mathbf{Z}),\beta_{2}$), where $\log\frac{\beta_{1}(\mathbf{Z})}{\beta_{2}}= \mathbf{Z}\gamma$, yielding a conditional mean preference value $ E(\pparam|\mathbf{Z})=\frac{e^{\mathbf{Z}\gamma}}{1+e^{\mathbf{Z} \gamma}}\bar{\pparam}
$. The details of this step and the results are reported in Appendix \ref{app:tables}. Both consideration and preference estimates remain close to those reported above.

\subsubsection{Proportionally Shifting Consideration}
We estimate the model with proportionally shifting consideration where the preference distribution and function $\alpha(\pparam)$ may depend on demographic variables (Table \ref{T:mleARCcollprop}). Motivated by the findings of the previous section, we assume that the cheapest/riskiest alternative is always considered. The consideration probability of the remaining alternatives is equal to $\aparam_j(\covx,\pparam)=  \aparam_j(1-\alpha(\pparam|\mathbf{Z}))$, where $\alpha(\pparam|\mathbf{Z}) = \xi_1(\mathbf{Z})\left(1-\frac{\pparam}{\bar{\pparam}}\right)^{\xi_2}$,  $\xi_1(\mathbf{Z})=\frac{e^{\mathbf{Z}\rho}}{1+e^{\mathbf{Z}\rho}}$, and $\xi_2$ is positive. We continue to assume that preferences are distributed Beta($\beta_{1}(\mathbf{Z}),\beta_{2}$).

The estimated average value of $\xi_1(\mathbf{Z})$ is $0.15$, with the $95\%$ CI of $[0.02, 0.21]$.  When $\xi_1(\mathbf{Z}) = 0.15$, a risk-neutral DM considers each of the safer alternatives $\{\$100, \$200, \$250, \$500\}$ $15\%$ less often than does an extremely risk averse DM. 
The estimated value of $\xi_2$ is $7.14$. This implies that as the risk aversion parameter increases from 0 to its estimated average median value of $0.0034$, consideration probability of the safer alternatives increases by $11\%$, and it is essentially flat after that rising by an additional $4\%$ as risk aversion reaches its upper bound.

\subsubsection{The Mixed Logit Random Utility Model}
As in the case of the $\modA$ model, we assume that $\pparam$ is Beta distributed on $[0,\bar{\pparam}]$, where $\bar{\pparam}=0.02$. The Mixed Logit satisfies the Generalized Dominance and smoothly spreads households' choices around their respective first bests. Consequently, it cannot match the observed distribution and, in particular, is unable to explain the relatively high observed share of the $\$200$ deductible.
Table \ref{T:mleRCLcoll} reports the estimation results and Figure \ref{F:RCLagg} compares the observed distribution of choices to the predicted choices. The predicted distribution is a much poorer fit relative to the $\modA$ model. In fact, the \cite{vuong1989likelihood} test soundly rejects (at 1$\%$ level) the Mixed Logit in favor of the $\modA$ model.

\subsubsection{The $\modA$ Model: All Coverages}\label{S:estmodAall}
We now proceed with estimation of the full model. 
We assume that households' consideration sets are formed over the entire deductible portfolio. There are 120 possible alternative triplets $(d^{coll},d^{comp},d^{home})$, each having its own probability of being considered. This model is flexible as it nests many rule of thumb assumptions such as only considering contracts with the same deductible level across the three contexts or only considering contracts with a larger collision deductible than comprehensive deductible. Figure \ref{F:ARCwide} and Table \ref{T:mleARCwide} present estimation results. The first panel of the figure shows the predicted distribution of choices across triplets, ranked in descending order by observed frequencies. The second panel plots the differences between predicted and observed choice distributions. Clearly, the predicted distribution is close to the observed distribution.

\begin{figure}[!t]\centering
	\caption{The $\modA$ Model, Three Coverages}
	\label{F:ARCwide}
	\includegraphics[trim={0cm 0 0 0cm},scale=0.3]{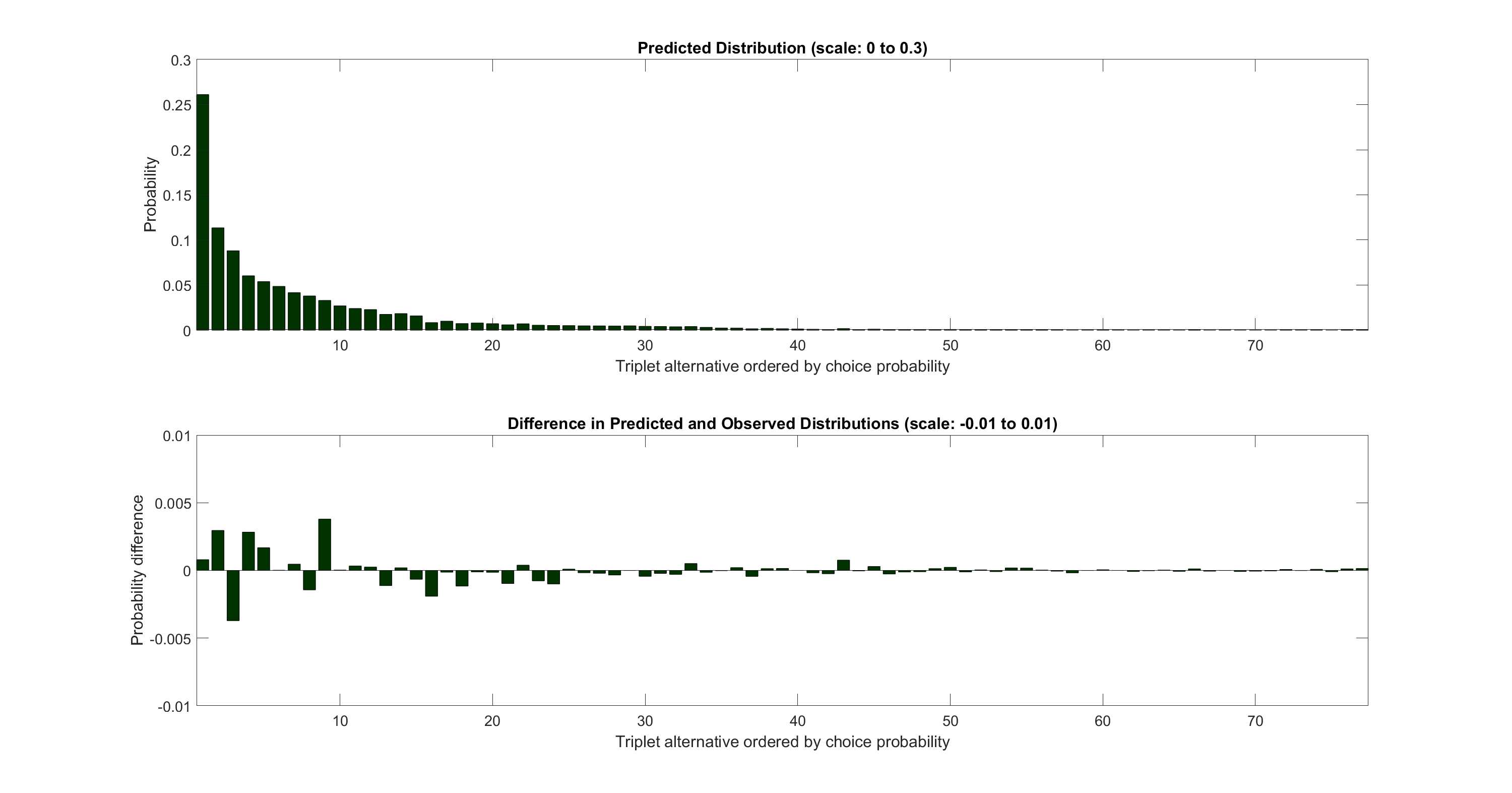}
	\captionsetup{font={footnotesize,sf},justification=raggedright}
	\caption*{\footnotesize Triplets are sorted by observed frequency at which they are chosen.  The first panel reports the predicted choice frequency and the second panel reports the difference in predicted and observed choice frequencies. The scale of the second panel is magnified by a factor of 30 to make the probability difference visible.}
\end{figure}
The largest difference between the predicted and observed shares equals $0.38$ percentage points, which is for the ($\$250,\$250,\$500$) triplet that is chosen by $2.9\%$ of the households. The integrated absolute error across all triplets is $3.46\%$. In our data, 43 out of 120 triplets are never chosen (these are omitted from  Figure \ref{F:ARCwide}). As discussed in Section \ref{S:LikeTract}, the likelihood maximization implies that the consideration probabilities for these triplets must be zero, so their predicted shares are zero.  Hence, the likelihood maximization routine is faster and more reliable as we do not need to search for $\aparam_j$ for these alternatives.

Another virtue of the $\modA$ model is that it effortlessly reconciles two sides of the debate on stability of risk preferences \citep{Barseghyan2011, Einav2012, Barseghyan2016}. On the one hand, households' risk aversion relative to their peers is correlated across lines of coverage, implying that households preferences have a stable component. On the other hand, analyses based on revealed preference reject the standard models: under full consideration, for the vast majority of households one cannot find a level of (household-specific) risk aversion that justifies their choices simultaneously across all contexts. Limited consideration allows the model to match the observed joint distribution of choices, and hence their rank correlations.

The estimated risk preferences are similar to those estimated with collision only data, although the variance is slightly smaller. The triplet considered most frequently is the cheapest one: ($\$1000,\$1000,\$1000$). Its consideration probability is 0.76, while the next two most considered triplets are (\$500, \$500, \$1000)  and (\$500, \$500, \$500). These are considered with probability 0.47 and 0.43, respectively. Overall, there is a strong positive correlation $(0.53)$ between the consideration probability and the sum of the deductibles in a given alternative.

We summarize once more the computational advantages of our procedure. First, estimation of our model remains feasible for a large choice set.\footnote{In our setting, it is feasible to estimate an additive error RUM assuming the DMs consider each deductible triplet as a separate alternative (Figure \ref{F:RUMwide} and Table \ref{T:mleRUMbroad}). As the figure shows, the failure to match the data is evident. The Vuong test formally rejects it in favor of the ARC model.} Second, the model's parameters grow linearly with the size of the choice set -- one parameter per an additional alternative. Third, enlarging the choice set does not call for new independent sources of data variation. For example, in our model whether there are five deductible alternatives or one hundred twenty does not make any difference either from an identification or an estimation stand point: with sufficient variation in $\bar{p}$ and/or $\mu$, the model is identified and can be estimated. As a final remark, once the model is estimated, one can compute the average monetary cost of limited consideration. In our data it is $\$50$ (see Appendix \ref{sec:monetarycost}).
\section{Discussion}\label{sec:discussion}
The literature concerned with the formulation, identification, and estimation of discrete choice models with limited consideration is vast. However, to our knowledge, there is no previous work applying such models to the study of decision making under risk, except for the contemporaneous work of \cite{Barseghyan2019}. In particular, this paper is the first to exploit the SCP for identification purposes. As a result, several fundamental differences emerge between our work and existing papers. First, we achieve identification in the most challenging case where there is a single excluded regressor that affects the utility of all alternatives.\footnote{This setting is common in insurance markets, see, e.g.,\cite{Cohen2007, Einav2012,  Sydnor2010, Barseghyan2011, Barseghyan2013, handel2013adverse, Bhargava2017}.} Second, we allow for consideration to depend on preferences. Third, with alternative-specific excluded regressors, this dependence can be essentially unrestricted and can be combined with dependence of consideration on (some of) the excluded regressors. Fourth, we scrutinize the large support assumption, show why it may be necessary, and when and how it is possible to make progress when it is not satisfied. Fifth, our approach comes with an easy to implement and computationally fast estimation strategy. Finally, we make a contribution specific to the study of decision making under risk by proposing a model that is immune from \cite{ape:bal18} criticism and features two sources of unobserved heterogeneity -- risk aversion and limited consideration -- whose distributions are identified. More generally, the paper establishes that, as long as the DMs' preferences satisfy the SCP, allowing for limited consideration does not hinder the model's identifiability or applicability. Hence, we view our framework as a stepping stone for studies of consumer behavior in markets where limited consideration may be present \citep[one example is][who builds on our framework to study consumer choice in Medicare Part D markets]{coughlin2019}.

Papers that allow for limited consideration or more broadly for choice set heterogeneity can be classified in four groups. The first relies on auxiliary information about the composition or distribution of DMs' choice sets, such as brand awareness \citep[e.g.,][]{Draganska2011, honka2017simultaneous} or search activity \citep[e.g.,][]{honka2017simultaneous, DelosSantos2012, Kim2010, HonkaRAND2017}.\footnote{For canonical cites see, e.g., \citet{Roberts1991} and \citet{Ben-Akiva1995}.} We do not require such information.

The second group attains identification via two-way exclusion restrictions, i.e., by assuming that some variables impact consideration but not utility and \emph{vice versa}. A well-known example of this approach is \citet{Goeree2008}, who posits that advertising intensity affects the likelihood of considering a computer, but does not impact consumer preferences, while computer attributes such as CPU speed affect preferences but not consideration (see also \citet{vanNierop2010} and \citet{Gaynor2016}). \citet{Hortacsu2017} create an exclusion restriction by exploiting the dynamic aspect of consumer choice.{\footnote{Time variation is used also in \citet{Crawford2019}, who show that with panel data and preferences in the logit family, point identification of preferences is possible, without any exclusion restrictions, under the assumption that choice sets and preferences are independent conditional on observables and with restrictions on how choice sets evolve over time. These restrictions enable the construction of proper subsets of DMs' true choice sets (`sufficient sets') that can be utilized to estimate the preference model.} The consumer's decision to consider alternatives to her current service provider is a function of (her experiences with) the last period provider but not her next period provider (see also \citet{heiss2016inattention}). In contrast, we achieve identification with as little as one common excluded regressor and a single cross section.

The third group relies on restricting the consideration mechanism to a specific class of models. \citet{Abaluck2019} consider two such models (and their hybrid): a variant of the ARC and a ``default specific'' model  \citep[as in, e.g.,][]{ho2017impact,heiss2016inattention} in which each DM's consideration set comprises either a single default alternative or the entire feasible set. They assume that consideration and preferences are independent, and that each alternative has a characteristic with large support that is additively separable in utility and may only affect its own consideration but not the consideration of other alternatives.\footnote{The exception is the ``default'' alternative, whose characteristic may trigger the consideration of the entire choice set.} They exploit violations of symmetry in the Slutsky matrix (i.e., in  cross-alternative demand responses to prices) to detect limited consideration. \citet{kawaguchi2019designing} study beverage purchases from vending machines, allowing advertisement to be a driver of consideration, but also to affect utility. Their approach is close to that of \citet{Goeree2008}, though they provide a formal argument for identification with large support and exclusion restrictions even when there is no choice set variation.  A key assumption is that all beverages are considered with probability equal to one as the advertising intensity of each beverage becomes very large.

The methods we propose relate to the papers in the third group in two aspects. First, we too \emph{sometimes} require large support as a ``fail safe'' assumption, but only in the most challenging case of a single common excluded regressor.  Second, we too rely on exclusion restrictions. The reliance on these assumptions is inescapable given the econometrics literature on point identification of discrete choice models. Our approach elucidates the identifying power of a single excluded regressor in models that satisfy the SCP and, in particular, the relative ranking of alternatives encapsulated in Facts \ref{Fact:simpleorder} and \ref{A:CO} \citep[see][for related results for average treatment effects in ordered discrete choice models]{lewbel2016identifying}. We further exploit this structure to establish identification in models with substantially richer levels of unobserved heterogeneity, by allowing for dependence between consideration and preferences.

The fourth group of papers has a different goal than what we pursue here, as it provides partial rather than point identification results. \citet{Cattaneo2019} propose a random attention model with homogeneous preferences, and they require that the probability of each consideration set is monotone in the number of alternatives in the choice problem. Their analysis yields testable implications and partial identification for preference orderings. \cite{Barseghyan2019} study discrete choice models, where consideration may arbitrarily depend on preferences as well as on all observed characteristics. They show that such unrestricted forms of heterogeneity generally yield partial, but not point, identification of the preference distribution and obtain bounds on the distribution of consideration sets' size. Finally, \citet{Dardanoni2018} consider a stochastic choice model with homogeneous preferences and heterogeneous cognitive types. They show how one can learn the moments of the distribution of cognitive types from a single cross section of aggregate choice shares.

\clearpage
\newpage

\linespread{1.22}
\bibliographystyle{apacite}
\vspace{-2em}
\begin{flushleft}
\bibliography{BMT_BCMT}
\end{flushleft}

\clearpage
\newpage

\linespread{1.3}
\numberwithin{equation}{section}
\begin{appendices}

\section{Proofs}\label{App:Identification}

\begin{proof}[Proof of Fact \ref{A:CO}]  If $\cmap_{1,2}(x)$ is less than $\cmap_{2,3}(x)$, then, for any DM with preference $\pparam$ s.t.  $\cmap_{1,2}(x) < \pparam < \cmap_{2,3}(x)$, $d_2$ is preferred to both $d_1$ and $d_3$, i.e. we are in Case (i). If $\cmap_{1,2}(x)>\cmap_{2,3}(x)$, then $d_2$ is either dominated by either $d_1$ or $d_3$. The relative location of $\cmap_{1,3}(x)$ is established as follows. First, suppose $\cmap_{1,3}(x)<\cmap_{1,2}(x)<\cmap_{2,3}(x)$. For any $\pparam \in (\cmap_{1,3}(x),\cmap_{1,2}(x))$ we have $ \util_\pparam(d_3,\covx) > \util_\pparam(d_1,\covx) > \util_\pparam(d_2,\covx) > \util_\pparam(d_3,\covx)$, which is an obvious contradiction. Second, suppose $\cmap_{2,3}(x)<\cmap_{1,2}(x)<\cmap_{1,3}(x)$. Then, for any $\pparam \in (\cmap_{1,2}(x),\cmap_{1,3}(x))$  we have $ \util_\pparam(d_1,\covx) > \util_\pparam(d_3,\covx) > \util_\pparam(d_2,\covx) > \util_\pparam(d_1,\covx)$, which is an obvious contradiction. The remaining two possibilities are excluded following the same logic.
\end{proof}

We maintain that $\covx$ has strictly positive density on $\cal{S}$ (Assumption \ref{TA0}), its density is continuous (Assumption \ref{TA1}), and that preferences are continuous and strictly monotone in $\covx$.  Therefore, if $\covx$ is a scalar and $c_{1,2}(\covx)$ covers  $[\pparam^l,\pparam^u]$, it is sufficient to consider  an interval $[x^l,x^u] \subset \cal{S}$ such that $[\pparam^l,\pparam^u]=\{c_{1,2}(\covx) : \covx \in [x^l,x^u]\}$.

The following lemma is useful for establishing Theorem \ref{T:Benchmark}.

\begin{lemmaA}\label{lem0}
	Suppose Assumptions \ref{TA0}-\ref{TA2} and \ref{IA1} hold. Suppose $\cmap_{1,2}(x)<\cmap_{1,j}(x)$, $\forall \covx \in \mathcal{X}$. Let $\{\covx^t\}_{t=1}^{\infty}$ be s.t. $\cmap_{1,2}(\covx^t)=\cmap_{1,j}(x^{t+1})$, $\covx^t \in \mathcal{X}$. Then $\exists T<\infty$ s.t. $c_{1,2}(x^{T})<0$.
\end{lemmaA}
\begin{proof}
	The cutoff $\{c_{1,2}(x^{t})\}_{t=1}^{\infty}$ is a strictly declining sequence. Suppose all its  elements are non-negative. Then it converges to some $\pparam^{\infty}\geq 0$ such that $\pparam^{\infty}=\cmap_{1,2}(x^{\infty})=\cmap_{1,j}(\covx^{\infty})$ for some $x^{\infty} \in \mathcal{X}$, a contradiction.
\end{proof}

\begin{proof}[Proof of Theorem \ref{T:Benchmark}]
The second condition in the theorem implies that the cutoffs are ordered: $c_{1,j}(\covx) < c_{1,j+1}(\covx)$ for all $\covx \in \mathcal{X}$. Hence
\begin{align*}
\Pr(d=d_1|\covx)
	& = \sum_{j=2}^D \sum_{\substack{\Kc \subset \Dc: \\1,j \in \Kc,\\2,\dots,j-1 \not \in \Kc}}Q(\Kc)  F(\cmap_{1,j}(x))
   + Q(\{d_1\}) \\
   & =  \sum_{j=2}^D\mathcal{O}(\{d_1,d_j\};\{d_2,\dots,d_{j-1}\})  F(\cmap_{1,j}(x))     + \mathcal{O}(d_1;\emptyset)\\
    & \equiv \sum_{j=2}^D \Lambda_j F(\cmap_{1,j}(x)) + \Lambda_{1},
\end{align*}
so that
\begin{align*}
\frac{d\Pr(d=d_1|\covx)}{dx}
& =\sum_{j=2}^D \Lambda_j f(\cmap_{1,j}(x))\frac{d\cmap_{1,j}(x)}{dx}.
\end{align*}

By Assumption \ref{IA1}, we can set $\covx^u=\covx^{\bar{\pparam}}$ s.t. $\cmap_{1,2}(\covx^{\bar{\pparam}}) =\bar{\pparam} $ and similarly $\covx^l=\covx^0$ s.t. $\cmap_{1,2}(\covx^0)= 0$. It may be the case that $\cmap_{1,\hat{D}}(x^0)<\bar{\pparam}$ and $\cmap_{1,\hat{D}+1}(x^0)>\bar{\pparam}$ for some $\hat D \geq 2$. Then,  $\forall j>\hat{D}$,  $\Lambda_j$ does not enter the expression for the derivative of $\Pr(d=d_1|\covx)$, $\forall \covx \in [x^0,x^{\bar{\pparam}}]$, because $f(\cmap_{1,j}(x))=0$. Henceforth, we  only consider the relevant alternatives for the derivative of $\Pr(d=d_1|\covx)$, namely $j\leq \hat{D}$.

Next, consider the derivative of $\Pr(d=d_j|\covx)$. By Fact \ref{Fact:simpleorder}, the term $\Lambda_j$ is the leading coefficient on $f(\cdot)$ for this derivative. There exists $\covx^j \in \mathcal{X}$ such that  $\cmap_{1,j}(x^j)  = \bar\pparam$. Thus,

\begin{align*}
\lim\limits_{\covx\nearrow\covx^{j}}\frac{d  \Pr(d=d_j|\covx)}{d\covx}
	= -\Lambda_j f(\bar \pparam)\frac{d\cmap_{1,j}(x^{j})}{dx}, \quad \forall j : 2 \leq j \leq \hat D.
\end{align*}
The ratio of $\lim\limits_{\covx\nearrow\covx^{j}}d\Pr(d=d_j|\covx)/d\covx$ and $\lim\limits_{\covx\nearrow\covx^{2}}d\Pr(d=d_2|\covx)/d\covx$ identifies $\Omega_j \equiv \frac{\Lambda_j}{\Lambda_2}$, where $\Lambda_2 \neq 0$ by the first assumption in the theorem. Rewrite the derivative of $\Pr(d=d_1|x)$ as follows:
\begin{align*}
\frac{d\Pr(d=d_1|\covx)}{dx}
	& = \sum_{j=2}^{\hat{D}} \Lambda_j f( \cmap_{1,j}(\covx))\frac{d\cmap_{1,j}(\covx)}{dx} \\
    & = \sum_{j=2}^{\hat{D}} \frac{\Lambda_j}{\Lambda_2}\Lambda_2 f(\cmap_{1,j}(\covx))\frac{d\cmap_{1,j}(\covx)}{dx} \\
    & =  \sum_{j=2}^{\hat{D}} \Omega_j [\Lambda_2 f( \cmap_{1,j}(\covx))]\frac{d\cmap_{1,j}(\covx)}{dx}  \\
    & = \sum_{j=2}^{\hat{D}} \Omega_j\hat f( \cmap_{1,j}(\covx))\frac{d\cmap_{1,j}(\covx)}{dx},
\end{align*}
where $\hat f(\pparam)\equiv \Lambda_2 f( \cmap_{1,j}(\covx))$. Equipped with $\Omega_j$, we can recover $\hat f(\pparam)$ sequentially. Note that $\forall \covx$ s.t. $c_{1,2}(\covx)\leq \bar{\pparam}$ and $c_{1,3}(\covx)> \bar{\pparam}$, the up-to-scale density $\hat f(\cmap_{1,2}(x))$ is identified. Indeed, it is the only unknown in the expression above. We proceed as follows.

First, let $\covx^1$ be such that $\cmap_{1,3}(\covx^1)=\bar{\pparam}$. Then, $\hat f(\cdot)$ is identified on $[\cmap_{1,2}(x^1),\bar{\pparam}]$.

Second, let $\covx^2$ be such that $\cmap_{1,2}(\covx^2)=\cmap_{1,3}(\covx^1)$. Now $\hat f(\pparam)$ is identified on $[\cmap_{1,2}(x^2),\bar{\pparam}]$ because in the expression for the derivative of $\Pr(d=d_1|x)$ all cutoffs $\cmap_{1,j}(\covx)$, $j>2$, lie on the part of the support where the up-to-scale density is known.

Identification of $\hat f(\pparam)$ on $[0,\bar \pparam]$ attains by repeating the above step. Indeed, by Lemma \ref{lem0} $\cmap_{1,2}(x^t)$ reaches the lower end of the support  in a finite number of steps. Finally, the scale is recovered by integrating $\hat f( \pparam)$ over its support:
\[
\Lambda_2 = \Lambda_2 \int_{0}^{\bar{\pparam}} f(\pparam)d\pparam  = \int_{0}^{\bar{\pparam}} \hat{f}(\pparam)d\pparam.
\]
Therefore $f(\cdot)$ is identified, as required. The term $\mathcal{O}(d_1;\emptyset)= \Lambda_1 = \Pr(d = d_1|\covx^{\bar{\pparam}})$ is also identified, and so are $\mathcal{O}(d_1,d_2;\emptyset) = \Lambda_2$ and $\mathcal{O}(\{d_1,d_j\};\{d_2,\dots,d_{j-1}\})  = \Lambda_j$.

\end{proof}

\begin{proof}[Proof of Theorem \ref{T:Binary}]
The second condition of  Theorem \ref{T:Benchmark} implies that the cutoffs are ordered: $c_{1,j}(\covx) < c_{1,j+1}(\covx)$ for all $\covx \in \mathcal{X}$. Hence,
\begin{align*}
	\frac{d\Pr(d=d_1|\covx)}{dx}
	= \sum_{j=2}^D \Lambda_j(\cmap_{1,j}(\covx))f(\cmap_{1,j}(\covx))\frac{d\cmap_{1,j}(\covx)}{d\covx},
\end{align*}
where
	\[
	\Lambda_j(\pparam) =
	\begin{cases}
	\underline{\Lambda}_j\equiv\sum_{\substack{\Kc \subset \Dc: \\1,j \in \Kc,\\2,\dots,j-1 \not \in \Kc}}\mathcal{\underline{Q}(K)} & \text{if } \pparam < \pparam^*  \\
	\overline{\Lambda}_j\equiv \sum_{\substack{\Kc \subset \Dc: \\1,j \in \Kc,\\2,\dots,j-1 \not \in \Kc}}\mathcal{\overline{Q}(K)}   & \text{if }      \pparam \geq \pparam^*  \\
	\end{cases}.		
	\]
Similar to the proof of Theorem \ref{T:Benchmark}, we only consider the relevant alternatives for the derivative of $\Pr(d=d_1|\covx)$, namely $j\leq \hat{D}$.	
	
We start at $\covx^{\bar{\pparam}}$ and hence $\cmap_{1,2}(\covx^{\bar{\pparam}})=\bar{\pparam}$. As we lower $\covx$ we check whether $\frac{d\Pr(d=d_1|\covx)}{dx}$ jumps. If it does not, identification of $f(\cdot)$ attains by the proof of Theorem \ref{T:Benchmark}.
	
Suppose there is a point of discontinuity. It arises when a cutoff $c_{1,j}(\covx)$ crosses the breakpoint $\pparam^*$. The identity of the cutoff and hence $\pparam^*=c_{1,j}(\covx)$ is identified by the fact there is a unique $\frac{d\Pr(d=d_j|\covx)}{dx}$ that also jumps. Equipped with the identity of $\pparam^*$ the proof proceeds similarly to that of Theorem \ref{T:Benchmark}. Indeed, all $\overline{\Omega}_j\equiv \overline{\Lambda}_j/\overline{\Lambda}_2$ are identified and so is $\hat{f}(\pparam)\equiv\overline{\Lambda}_2 f(\pparam)$ for all $\pparam>\pparam^*$.
	
The additional step is how to identify  $\underline{\Lambda}_j$ and $\underline{\Omega}_j\equiv \underline{\Lambda}_j/\underline{\Lambda}_2$. Start with $\underline{\Lambda}_2$. Consider an $\covx^*$ s.t. $\cmap_{1,2}(\covx^*)=\pparam^*$. The derivatives $\frac{d\Pr(d=d_1|\covx)}{dx}$ from the left and from the right of $\covx^*$ identify $\underline{\Lambda}_2 f(\pparam^*)$ and $\overline{\Lambda}_2 f(\pparam^*)$. Hence, the ratio $\underline{\Lambda}_2/\overline{\Lambda}_2$ is identified. Exactly the same logic applies to all other $\underline{\Lambda}_j$'s whenever $\cmap_{1,j}(\covx)$ crosses $\pparam^*$. We can then rewrite
	\begin{align*}
	\frac{d\Pr(d=d_1|\covx)}{dx}
	= \sum_{j=2}^{\hat{D}} \left(\overline{\Omega}_j\right)^{\mathbbm{1}( c_{1,j}(x) \geq \nu^*)}\left(\underline{\Omega}_j\frac{\underline{\Lambda}_2}{\overline{\Lambda}_2}\right)^{\mathbbm{1}( c_{1,j}(x) < \nu^*)} \hat{f}(\cmap_{1,j}(\covx))\frac{d\cmap_{1,j}(\covx)}{d\covx},
	\end{align*}
Now all coefficients of $\hat{f}(\cmap_{1,j}(\covx))$ are identified, and identification of $\hat{f}(\cdot)$ proceeds to the left of $\pparam^*$. Once it is identified, we integrate it over the support to recover $\overline{\Lambda}_2$. Hence $\underline{\Lambda}_2$ and $f(\cdot)$ are identified.

When the breakpoint occurs at $x^*$ rather than $\nu^*$, the same proof strategy can be applied.

\end{proof}

\begin{proof}[Proof of Proposition \ref{P:Basic}]
The second condition in the theorem implies that $\cmap_{1,2}(\covx)<c_{1,j}(\covx)$ for  any $j>2$. The cutoffs $c_{1,j}(\covx)$'s, $j>2$, are irrelevant for evaluating $\Pr(d=d_1|x)$  by the first condition of the theorem.
Therefore,
\begin{align*}
	\frac{d  \Pr(d=d_1|\covx)}{dx}
& = \left(\sum_{\Kc \subset \Dc: 1,2  \in \Kc} Q(\Kc)\right)f(\cmap_{1,2}(\covx))\frac{d \cmap_{1,2}(\covx)}{dx} \equiv \alpha f(\cmap_{1,2}(\covx))\frac{d \cmap_{1,2}(x)}{dx}.
\end{align*}
Consequently, the product $\alpha f(\cmap_{1,2}(\covx))$ can be written in terms of data:
\[
\alpha f(\cmap_{1,2}(\covx)) = \frac{\frac{d  \Pr(d=d_1|\covx)}{dx}}{\frac{d \cmap_{1,2}(x)}{dx}},
\]
and hence variation in $\covx$ guarantees that $f(\pparam)$ is identified up-to-scale on $[\cmap_{1,2}(x^l),\cmap_{1,2}(x^u)] = [\pparam^l,\pparam^u]$.  It follows that $F(\pparam|\pparam \in[\pparam^l,\pparam^u])$ is identified.
\end{proof}

\begin{proof}[Proof of Theorem \ref{T:Basic-G}]
Under Condition \emph{1} of the Theorem there can only be three types of consideration sets. The first type are all possible subsets of $\{d_1,\dots,d_j\}$; the second type are all possible subsets of $\{d_{j+1},\dots,d_D\}$. The third type necessarily contains both $d_j$ and $d_{j+1}$. The probability of choosing an alternative in $\{d_1,\dots,d_j\}$ is one for the first type and zero for the second type. Hence,
\begin{align*}
	\frac{d  \Pr(d \in \{d_1,\dots,d_j\}|\covx)}{dx}
& = \sum_{\Kc \subset \Dc: j,j+1  \in \Kc} Q(\Kc)f(\cmap_{j,j+1}(\covx))\frac{d \cmap_{j,j+1}(\covx)}{dx}.
\end{align*}
The rest of the proof follows the same steps as in the proof of Proposition \ref{P:Basic}, except we now track $\cmap_{j,j+1}(\covx)$.
\end{proof}

\begin{proof}[Proof of Theorem \ref{Broad-alternative specific}]
Let $\pparam$, $\tilde \covx$,  $\mathcal{N}_{ \epsilon}(\tilde \covx) \equiv \{ \covx : \| \covx - \tilde \covx\| <  \epsilon\}$ satisfy Condition \emph{2} in the theorem. Then, $\pparam = c_{j,k}(\tilde \covx)$ for all $j,k$. Consider any pair of alternatives $(d_j,d_k)$.    Since utility is strictly monotone in $x_j$ and continuous, for each $\mathcal{L} \subseteq \mathcal{D}\setminus\{j,k\}$ we can find $x \in \mathcal{N}_{ \epsilon}(\tilde \covx)$ such that
		\begin{align}
        & U_{ \pparam}(d_j, \covx)= U_{ \pparam}(d_k,   \covx); \label{thm4proof:eq1}\\
		& U_{ \pparam}(d_l, \covx)> U_{ \pparam}(d_j,  \covx) \quad  \forall l \in \mathcal{L}; \label{thm4proof:eq2} \\
		&U_{ \pparam}(d_j, \covx) > U_{ \pparam}(d_l,  \covx) \quad  \forall l \in \mathcal{D}\setminus\{\mathcal{L}\cup\{j,k\}\}; \label{thm4proof:eq3}
        \end{align}
The remainder of the proof proceeds in two steps.

\emph{Step 1: Identification of $f(\pparam)\mathcal{Q}_{\pparam}(\mathcal{K})$}: The singleton sets occur with zero probability by Condition \emph{1} in the theorem, so it remains to show identification  for consideration sets larger than one.  Consider any two alternatives $(d_j,d_k)$. We claim that the following statement holds for $n=0,\dots,D$:

\begin{quote}
  $P(n)$: For all $\mathcal{K} \subset \Dc \setminus\{j,k\}$  satisfying $|\mathcal{K}| \leq n$, the quantity
   $f(\pparam)\mathcal{Q}_{\pparam}(\{j,k\} \cup \mathcal{K})$
   \newline \hspace*{2.6em} is identified.
\end{quote}

To show this for $P(0)$, set $\mathcal{L}  = \mathcal{D}\setminus\{j,k\}$. In this case $\mathcal{K} = \emptyset$.  Let $ \covx$ satisfy Equations \eqref{thm4proof:eq1}-\eqref{thm4proof:eq3}. Then, all alternatives $d_{l},~ l\neq j,k$, are preferred to $d_j$ and $d_{k}$ at $\pparam$ and $c_{j,k}(\covx)=\pparam$. Hence:
\begin{align*}
\frac{\frac{\partial \Pr(d=d_{j}|x)}{\partial x_{k}}}{\frac{\partial c_{j,k}(\covx)}{\partial x_{k}}}
& =f(\pparam) \mathcal{Q}_{\pparam}( \{ j,k\}).
\end{align*}
It follows that $f(\pparam)\mathcal{Q}_{\pparam}( \{ j,k\})$ is identified.

Next, suppose $P(n-1)$ is true. Consider any $\mathcal{K}\subset \Dc \setminus\{j,k\}$  such that $|\mathcal{K}|=n$.  Let $\mathcal{L}  = \mathcal{D}\setminus(\mathcal{K} \cup \{j,k\})$.  Let $\covx$ satisfy Equations \eqref{thm4proof:eq1}-\eqref{thm4proof:eq3}. Then,
\begin{align*}
\frac{\frac{\partial \Pr(d=d_{j}|x)}{\partial x_{k}}}{\frac{\partial c_{j,k}(x)}{\partial x_{k}}}
& = f(\pparam) \sum_{\mathcal{C} \subset \mathcal{K}} \mathcal{Q}_{\pparam} (\{j,k\} \cup \mathcal{C}) \\
& = f(\pparam)\mathcal{Q}_{\pparam}(\{j,k\} \cup \mathcal{K})
+
 \sum_{\mathcal{C} \subset  \mathcal{K}:|\mathcal{C}| < n} f(\pparam) \mathcal{Q}_{\pparam} (\{ j,k \} \cup \mathcal{C}).
\end{align*}
The LHS of this expression is known, and the second term on the RHS is identified by the induction step. Therefore $P(n)$ holds.

Since $d_j$ and $d_k$ were chosen arbitrarily, it follows that $f(\pparam)\mathcal{Q}_{\pparam}(\mathcal{K})$ is identified for all $\mathcal{K} \subset \Dc$.

\emph{Step 2: Identification of $f(\pparam)$ and $\mathcal{Q}_{\pparam}(\mathcal{K})$}. Since
\[
\sum_{\mathcal{K} \subset \Dc} f(\pparam)\mathcal{Q}_{\pparam}(\mathcal{K})
=
f(\pparam) \sum_{\mathcal{K} \subset \Dc} \mathcal{Q}_{\pparam}(\mathcal{K})
=
f(\pparam),
\]
$f(\pparam)$ is identified. Identification of $\mathcal{Q}_{\pparam}(\mathcal{K})$ follows from Step 1.
\end{proof}

\begin{proof}[Proof of Corollary \ref{Broad-alternative specific_corollary}]
 The proof follows the same steps as the proof of Theorem \ref{Broad-alternative specific}, but with the following two modifications:

First modification: In Step 1 in the  proof of Theorem \ref{Broad-alternative specific},  we start with $d_j = d_1$ and loop over $d_k \in \{d_2,\dots,d_D\}$.  This ensures that we only take derivatives with respect to $x_k$, $k>1$. Hence,  $f(\pparam)\mathcal{Q}^{\covx_1}_{\pparam}(\mathcal{K})$ is identified for all sets $\mathcal{K} \subset \Dc : |\mathcal{K}|>1$.

Second modification: In Step 2 we obtain
\[
 f(\pparam)\mathcal{Q}^{\covx_1}(d_1)
 +
 \sum_{\mathcal{K} \subset \Dc : |\mathcal{K}|>1} f(\pparam)\mathcal{Q}^{\covx_1}_{\pparam}(\mathcal{K})
=
f(\pparam) \sum_{\mathcal{K} \subset \Dc} \mathcal{Q}^{\covx_1}_{\pparam}(\mathcal{K})
=
f(\pparam).
\]
Since the second term on the LHS is known, $f(\pparam)(1 - \mathcal{Q}^{\covx_1}(d_1))$ is identified for all $\pparam \in [0,\bar \pparam]$.  The scale is identified, because
\[
(1 - \mathcal{Q}^{\covx_1}(d_1)) = \int_{0}^{\bar \pparam} f(\pparam)(1 - \mathcal{Q}^{\covx_1}(d_1)) d \pparam.
\]
Once the scale is identified, $f(\pparam)$ is identified and so are $\mathcal{Q}^{x_1}_{\pparam}(\mathcal{K}), ~\forall \Kc \subset \Dc$.
\end{proof}

\begin{proof}[Proof of Proposition \ref{T:Test}]
For the purpose of obtaining a contradiction, suppose that there is full consideration.  Then
\begin{align*}
\Pr(d\in \Lc|x) & = \Pr\left(\argmax_{j \in \Dc} \util_\pparam(d_j,\covx) \in  \Lc\bigg|\covx\right) \\
& = \Pr(\pparam  \in [0, \pparam^*)) \\
& = F(\pparam^*) \\
& =  \Pr\left(\argmax_{j \in \Dc} \util_\pparam(d_j,\covx') \in  \Lc'\bigg|\covx'\right)\\
& = \Pr(d\in \Lc'|x').
\end{align*}
This is a contradiction.  Therefore there is limited consideration.

\end{proof}

The following two Lemmas are used in the proof of Theorem \ref{T:ARC_independent}. The proofs of these Lemmas rest on the following claims.
\begin{enumerate}
	\item The probability of alternative $d_j$ being chosen can only increase in its consideration probability.
	\item The probability of alternative $d_j$ being chosen can only decline in consideration probability of any other alternative $d_k$.
	\item The probability of alternative $d_j$ or $d_k$ being chosen can only increase in the consideration probability of $d_j$ as the positive effect of this change on $Pr(d=d_j|x)$ dominates the negative effect on $Pr(d=d_k|x)$.
	\item The probability of an alternative in $\Kc$ being chosen can only decline in consideration probability of any alternative that does not belong to $\Kc$.
\end{enumerate}

\begin{lemmaA}\label{lem1}
	Consider the Basic \ref{TA3} model. For any  $\Kc \subset \Dc$, $\sum_{j \in\Kc} \Pr(d=d_j|x)$ is increasing in $\aparam_{k}$, $\forall k \in\Kc$, and decreasing in $\aparam_{k}$, $\forall k \notin\Kc$.
\end{lemmaA}
\begin{proof}
Fix $\Kc$ and consider any $j \in \Kc$.   For each $\nu$ and $l \in \Kc$, $l \neq j$, either $j \in \Bc_{\pparam}(\dor_l,x)$ or not.  If  $j \not \in  \Bc_{\pparam}(\dor_l,x)$, then $\Pr(d=\dor_l| \covx,\pparam)$ does not depend on $\aparam_j$.  Hence,
\begin{align*}
\sum_{l \in\Kc} \Pr(d= d_l|\covx, \pparam)
& = A + \Pr(d=d_j|\covx, \pparam)  + \sum_{l \in\Kc: j  \in  \Bc_{\pparam}(\dor_l,x) } \Pr(d=d_l|\covx, \pparam) ,
\end{align*}
where $A$ is a constant that collects terms that do not depend on $\aparam_j$.  Continuing,
\begin{align*}
\sum_{l \in\Kc} \Pr(d=d_l|\covx, \pparam)  & =
A+ \aparam_j \prod_{k \in \Bc_{\pparam}(\dor_j,x)} (1-\aparam_k)
+  \sum_{l \in\Kc: j  \in  \Bc_{\pparam}(\dor_l,x) } \aparam_l \prod_{k \in \Bc_{\pparam}(\dor_l,x)} (1-\aparam_k) \\
& =  A + \aparam_j \prod_{k \in \Bc_{\pparam}(\dor_j,x)} (1-\aparam_k)
+   \sum_{l \in\Kc: j  \in  \Bc_{\pparam}(\dor_l,x) } \aparam_l(1-\aparam_j) \prod_{k \in \Bc_{\pparam}(\dor_l,x) \setminus \{j\}} (1-\aparam_k) \\
& = A +  \sum_{l \in\Kc: j  \in  \Bc_{\pparam}(\dor_l,x)}  \left(\aparam_l \prod_{k \in \Bc_{\pparam}(\dor_l,x) \setminus \{j\}} (1-\aparam_k)\right)
\\
& \quad \quad +  \left( \prod_{k \in \Bc_{\pparam}(\dor_j,x)} (1-\aparam_k) - \sum_{l \in\Kc: j  \in  \Bc_{\pparam}(\dor_l,x)} \aparam_l \prod_{k \in \Bc_{\pparam}(\dor_l,x) \setminus \{j\}} (1-\aparam_k) \right) \aparam_j \\
& \equiv \tilde A + B \aparam_j.
\end{align*}
Since $\Bc_{\pparam}(j,x) \subset \Bc_{\pparam}(\dor_l,x)$ whenever $j \in  \Bc_{\pparam}(\dor_l,x)$,
\[
B =\left(\prod_{k \in \Bc_{\pparam}(\dor_j,x)} (1-\aparam_k)  \right)\left(1 - \sum_{l \in\Kc: j  \in  \Bc_{\pparam}(\dor_l,x)} \aparam_l \prod_{k \in   \Bc_{\pparam}(\dor_l,x) \setminus \{\Bc_{\pparam}(\dor_j,x),j\}} (1-\aparam_k) \right)  \geq 0.
\]
Therefore, $\sum_{j \in\Kc} \Pr(d=d_j|x) = \int \sum_{j \in\Kc} \Pr(d=d_j|\covx, \pparam)dF$ is increasing in $\aparam_j$.

Finally, for any $k\not \in \Kc$, $\aparam_k$ may only appear on the RHS as $(1-\aparam_k)$. Hence $\sum_{j \in\Kc} \Pr(d=d_j|x)$ is decreasing in $\aparam_k$.
\end{proof}

\begin{lemmaA}\label{lem2}
	Consider the Basic \ref{TA3} model. For any  $\Kc \subset \Dc$, $\sum_{j \in\Kc} \Pr(d=d_j|x)$ is strictly increasing in $\aparam_{j}$, $j \in \Kc$,  whenever there is an open interval of $\pparam$'s at which alternative $d_j$ is preferred to all of the always-considered alternatives.  It is strictly decreasing in $\aparam_{k}$, $j \not \in \Kc$, whenever there is an open interval of $\pparam$'s and $l \in \Kc$ such that at these $\pparam$'s alternative $d_k$ is preferred to $d_l$ and $d_l$ is preferred to all of the always-considered alternatives.
\end{lemmaA}

\begin{proof}
To show the first claim, notice that $B =0$ in the proof of Lemma \ref{lem1} if and only if $\aparam_k = 1$ for some $k \in \Bc_{\pparam}(\dor_j,x)$.

To show the second claim, consider any $j \in \Kc$.  Then,
\[
\Pr(d=d_j|\covx, \pparam)  =
\aparam_j \prod_{k \in \Bc_{\pparam}(\dor_j,x)} (1-\aparam_k).
\]
For this to be strictly decreasing in $\aparam_k$, it must be the case that $k \in \Bc_{\pparam}(\dor_j,x)$ and $\aparam_l < 1$ for all $l \in  \Bc_{\pparam}(\dor_j,x) \setminus \{k\}$.
\end{proof}
	
\begin{proof}[Proof of Theorem \ref{T:ARC_independent}]
By the proof of Theorem \ref{T:Benchmark}, $f(\cdot)$ and $\Lambda_2 = \aparam_1 \aparam_2$ are identified.  The consideration parameter $\aparam_1$ is identified by $\Pr(d=d_1|\covx^{\bar{\pparam}}) = \aparam_1$, where $\covx^{\bar{\pparam}}$ is s.t. $c_{1,2}(\covx^{\bar{\pparam}})=\bar \pparam$.  Since $\Lambda_2$ is known, $\aparam_2$ is also identified.  The rest of the proof is about identification of the remaining consideration parameters.

To identify $\aparam_{j}$ take an $\covx$ such that $\Pr(d=d_j|x)\neq 0$. Denote $\mathcal{E} = \{k : \Pr(d=d_k|x)\neq 0\}$. We claim that $\Pr(d=d_k|x), \forall k \in \mathcal{E}$, does not depend on $\aparam_{l}$ for $l \notin \mathcal{E}$. Suppose otherwise. That is, suppose there exists $d_l$ such that $\Pr(d=d_l|x) = 0$ and $\Pr(d=d_k|x)$ depends on $\aparam_l$ for some $k \in \mathcal{E}$. Then, for each $\pparam$ there is an always-considered alternative that is preferred to $d_l$. Since $\Pr(d=d_k|x)$ depends on $\aparam_{l}$, there exists $\pparam \in [0,\bar \pparam]$ such that $d_l$ is preferred to $d_k$. However, the always-considered alternative that is preferred to $d_l$ at $\pparam$ is also preferred to $d_k$  by transitivity. This leads to a contradiction, because a DM with such preferences will never choose $d_k$ in the first place. Therefore, $\Pr(d=d_k|x)$ does not on $\aparam_{l}$ for any $l \not \in \mathcal{E}$.

Since $F(\cdot)$ is already identified, $\{\Pr(d=d_k| \covx)\}_{k \in \mathcal{E}} $ defines a system of $|\mathcal{E}|$ non-linear equations, where the only unknowns are $\aparam_k$, $k \in \mathcal{E}$.  This system has a unique solution. Suppose to the contrary that two sets of consideration parameters $\{\aparam_k\}_{k \in \mathcal{E}}$ and $\{\aparam_k'\}_{k \in \mathcal{E}}$ solve this system and they are distinct. Denote $\mathcal{E_{+}} = \{k :\aparam_k>\aparam_k'\}$. By Lemma \ref{lem2}, $\sum_{k \in \mathcal{E_{+}}} \Pr (d=d_k| \covx)$ is strictly larger at $\{\aparam_k\}_{k \in \mathcal{E}}$ than at $\{\aparam_k'\}_{k \in \mathcal{E}}$.  Hence, only one of these sets could satisfy data. Therefore there is a unique set of $\{\aparam_k\}_{k \in \mathcal{E}}$ that solves this system of equations, and $\aparam_j$ is identified as claimed.
\end{proof}

\begin{proof}[Proof of Theorem \ref{T:ARC_dependent_1}]
Step 1. Suppose $D=3$. Then $d^*=2$ and
\begin{align*}
	\Pr(d=d_1|x) = \int_{0}^{\cmap_{1,2}(x)}\aparam_1(\pparam) dF \quad \text{and} \quad
	\Pr(d=d_3|x) = \int_{\cmap_{2,3}(x)}^{\bar{\pparam}}\aparam_3(\pparam) dF.
\end{align*}
The ratio of the derivatives of these two moments yields $\frac{\aparam_1(\pparam)}{\aparam_3(\pparam)}$, where
\begin{align*}
	\aparam_1(\pparam)&=\aparam_1(1-\alpha(\pparam))\\
	\aparam_3(\pparam)&=\aparam_3(1+\alpha(\pparam)),
\end{align*}
First, $\frac{\aparam_1}{\aparam_3}$ is identified when $\covx$ and $\covx'$ are chosen such that $\cmap_{1,2}(\covx) = \cmap_{2,3}(\covx') =\bar{\pparam}$. Once $\frac{\aparam_1}{\aparam_3}$ is identified, $\frac{1-\alpha(\pparam)}{1+\alpha(\pparam)}$ is known for all $\pparam$; hence, $\alpha(\pparam)$ can be solved for. Identification of $f(\pparam)$ follows from substituting $\alpha(\pparam)$ into the expression for 	$\frac{\Pr(d=d_1|x)}{dx}$.

Step 2. Let $D>3$. We identify $d^*$ at $\covx^{\bar{\pparam}}$. The smallest $j$ such that $\Pr(d=d_j|\covx^{\bar{\pparam}})=0$ yields $d^*=j-1$ (or $d^*=D$ if no such $j$ exists).  We return to the case that $d^* = 1$, $d^* = 2$, or $d^* = D$ at the end of this proof.

Step 3. Using large support we establish that $\aparam_{D}$ is a decreasing function of $\aparam_{1}$. We have
\begin{equation*}
	\Pr(d=d_1|\covx^{\bar{\pparam}}) = \aparam_{1}\int_{0}^{\bar{\pparam}} (1-\alpha(\pparam))dF(\pparam) = \aparam_{1}(1-E \alpha(\pparam)).
\end{equation*}
Similarly,
\begin{equation*}
	\Pr(d=d_D|x^0) = \aparam_{D}\int_{0}^{\bar{\pparam}} (1+\alpha(\pparam))dF(\pparam) = \aparam_{D}(1+E \alpha(\pparam)).
\end{equation*}
Hence
\[
\aparam_1 = \frac{\Pr(d=d_1|x^1)}{2 - \frac{\Pr(d=d_D|x^0)}{ \aparam_D}}.
\]

Step 4. This is an intermediate step, which we use later in the proof. By Fact \ref{Fact:simpleorder},
\[
\cmap_{1,j}(x)<c_j^*(\covx) \equiv \min_{k}\{\{\cmap_{k,j}(\covx)\}_{1<k<j}, \{\cmap_{j,k}(\covx)\}_{j<k\leq D}\}, ~\forall j .
\]
Moreover any sequence $\{x^s\}_{s=1}^{\infty}$ such that $c_j^*(\covx^{s}) = \cmap_{1,j}(\covx^{s+1})$  will reach the lower bound of the support in finite number of steps. Otherwise, by the argument in the proof of Lemma \ref{lem0}, $c_j^*(\covx^{s})$ and $\cmap_{1,j}(\covx^{s})$ converge to the same point in the interior of the support, which contradicts the assumptions of the theorem.

Step 5. Identification of $\{\aparam_j\}_{1<j<d^*}$. For each $\pparam$ in a sufficiently small neighborhood near $\bar \pparam$, say $(\bar \pparam - \eps,\bar \pparam)$, and for each $j$, there is an $\covx^{j}$ such that $\cmap_{1,j}(\covx^{j})=\pparam$ and $\cmap_{k,j}(\covx^{j}) > \bar \pparam$ for all $k \neq 1,j$. It follows by Step 4 that the following equations hold:
\begin{align*}
\frac{d\Pr(d=d_2|\covx^2)}{d\covx}/\frac{d\cmap_{1,2}(\covx^2)}{dx}=& \aparam_{1}(\pparam)\aparam_{2}(\pparam)f(\pparam)  \nonumber \\ \nonumber
\frac{d\Pr(d=d_3|\covx^3)}{dx}/\frac{d\cmap_{1,3}(\covx^3)}{dx}=& \aparam_{1}(\pparam)(1-\aparam_{2}(\pparam))\aparam_{3}(\pparam)f(\pparam)\\
\frac{d\Pr(d=d_4|\covx^5)}{dx}/\frac{d\cmap_{1,4}(\covx^4)}{dx}=& \aparam_{1}(\pparam)(1-\aparam_{2}(\pparam))(1-\aparam_{3}(\pparam))\aparam_{4}(\pparam)f(\pparam)\\ \nonumber
\vdots \\ \nonumber
\frac{d\Pr(d=d_{d^*}|\covx^{d^*})}{dx}/\frac{d\cmap_{1,d^*}(\covx^{d^*})}{dx}=& \aparam_{1}(\pparam)(1-\aparam_{2}(\pparam))(1-\aparam_{3}(\pparam))\dots (1-\aparam_{d^*-1}(\pparam))f(\pparam) \nonumber
\end{align*}
The summation of these expressions recovers the quantity $\aparam_{1}(\pparam)f(\pparam)$. Next, substitute $\aparam_{1}(\pparam)f(\pparam)$ into the expressions above to sequentially recover $\{\aparam_{j}(\pparam)\}_{2 \leq j \leq d^*}$. Since $ \pparam$ can be made arbitrarily close to $\bar \pparam$ by selecting a smaller value of $\varepsilon$ and since $\alpha(\cdot)$ is continuous, $\lim_{\pparam \to \bar \pparam} \aparam_j(\pparam) = \aparam_{j}(\bar{\pparam}) = \aparam_j(1 - \alpha(\bar \pparam)) = \aparam_j$ is also identified.  Hence, $\{\aparam_j\}_{j=2}^{d^*}$ are identified.

Step 6: Identification of $\aparam_1$ and $\{\aparam_j\}_{d^*< j \leq D}$. The cutoffs are monotone in $x^t$ and  all cutoffs are on the right of $\bar{\pparam}$ at $\covx^{\bar{\pparam}}$. Consequently, $\Pr(d=d_j|x^{\bar{\pparam}})=0$ for all $j>d^*$. Continuously decrease $t$ until $\Pr(d=d_{j_1}|x^t)>0$ for some $j_1 \in \mathcal{J} \equiv \{ d^* + 1 ,\dots,D\}$ and $\Pr(d=d_{k}|x^t)=0$ for all $k \in \mathcal{J}\setminus\{j_1\}$. This will happen when $\cmap_{\dcon,j_1}(\covx^t)$ crosses $\bar{\pparam}$, yielding
\begin{equation*}
	\frac{d\Pr(d=d_{j_1}|\covx^t)}{dx}/\frac{d\cmap_{d*,j_1}(\covx^t)}{dx}= -\aparam_{j_1}(\pparam)f(\pparam) M_1(\pparam),
\end{equation*}
where $\pparam$ is in a small neighborhood near $\bar \pparam$ s.t. $\cmap_{k,j_1}(\covx^t)>\bar \pparam$ for all $k>d^*, k \neq j_1$ and
\[
\quad \quad
M_1(\pparam) \equiv \prod_{k \in \{2,\cdots,d^*-1\}  : c_{k,j_1}(x^t)>\bar{\pparam}} (1-\aparam_k(\pparam)).
\]
Near the end of this proof-step we will take the limit as $\pparam \to \bar \pparam$ after dividing out $f(\pparam)$.  Importantly, $M_1 \equiv \lim_{\pparam \to \bar \pparam}M_1(\pparam)$ is known, since all relevant $\aparam_{k}$'s are known, and $M_1$ does not depend on $\aparam_1$, since $c_{1,j_1}(x^t)< c_{d^*,j_1}(x^t) < \bar{\pparam}$.

Next, continuously decrease $t$ further until  $\Pr(d=d_{j_2}|x^t)>0$ for some $j_2 \in \mathcal{J}\setminus \{j_1\}$ and $\Pr(d=d_k|x^t)=0$ for all $k \in \mathcal{J}\setminus\{j_1,j_2\}$.  Again, this will happen when $\cmap_{\dcon,j_2}(\covx)$ crosses $\bar{\pparam}$. Hence,
\begin{align*}
&\frac{d\Pr(d=d_{j_2}|\covx^t)}{dx}/\frac{d\cmap_{d*,j_2}(\covx^t)}{dx}= -\aparam_{j_2}(\pparam)f(\pparam) M_2(\pparam),	\\
&M_2(\pparam) \equiv \prod_{k \in \{2,\dots,d^*-1,j_1\}: c_{k,j_2}(x^t)>\bar{\pparam}} (1-\aparam_k(\pparam)).
\end{align*}
	
The term $M_2 \equiv \lim_{\pparam \to \bar \pparam}M_2(\pparam)$ is known, except possibly for the term $(1-\aparam_{j_1})$, since all other relevant $\aparam_{k}$'s are known. The expression above defines $\aparam_{j_2}(\pparam)$ as a strictly increasing function of $\aparam_{j_1}(\pparam)$ regardless of whether $M_2(\pparam)$ depends on $\aparam_{j_1}(\pparam)$ or not. Indeed, for the case where $j_1<j_2$ we have
\begin{align*}
  \aparam_{j_2}(\pparam)&\propto
     \begin{cases}
    \aparam_{j_1}(\pparam) & \text{if }  c_{j_1,j_2}(x^t) < \bar \pparam \\
    \frac{\aparam_{j_1}(\pparam)}{1-\aparam_{j_1}(\pparam)}   & \text{if } c_{j_1,j_2}(x^t)  \geq  \bar \pparam
    \end{cases},
\end{align*}
where the coefficients of proportionality are known in the limit. A similar expression holds when $j_1>j_2$. This argument immediately extends to all $j \in \mathcal{J}$. In particular, for the case where $j_2<j_3$
\begin{align*}
\aparam_{j_3}&\propto
\begin{cases}
\aparam_{j_2(\pparam)}  & \text{if }  c_{j_2,j_3}(x^t) < \bar \pparam \\
\frac{\aparam_{j_2}(\pparam)}{1-\aparam_{j_2}(\pparam)}  & \text{if } c_{j_2,j_3}(x^t)  \geq  \bar \pparam
\end{cases}.
\end{align*}
Since $Pr(d=d_D|x^0) \neq 0$, the above sequential argument yields that $\aparam_{D}(\pparam)$ is an increasing function of $\aparam_{j_1}(\pparam)$. In turn, recall that $\aparam_1(\pparam) f(\pparam)$ is known for $\pparam$ arbitrary close to $\bar \pparam$. The limit of the ratio between $\aparam_1 (\pparam)f(\pparam)$ and  $\aparam_{j_1}(\pparam)f(\pparam)$, which is also known, yields $\aparam_D$ as an increasing functions of $\aparam_{1}$. Hence, taken with the result in Step 3, the quantity $\aparam_{1}$ is uniquely pinned down. Identification of all other $\aparam_{j}$'s immediately follow.

Step 7: Identification of $\alpha(\pparam)$ and $f(\pparam)$. The identification argument is iterative.  For each alternative $j$, define
\[
\Gamma^0_j \equiv \{  \pparam \in [0, \bar \pparam] : \exists \covx \in \mathcal{X} \text{ s.t. } \pparam = c_{1,j}(\covx)  \text{ and } c^*_j(\covx) \geq \bar \pparam\}.
\]
The set $\Gamma^0_j$ includes all preference parameters $\pparam$ covered by the cutoff $\cmap_{1,j}(\cdot)$  before any other relevant cutoffs for $d_j$ enter the support. Let $\Gamma^0 \equiv \bigcap_{j=1}^{D} \Gamma^0_j$.   By Step 4,  $\Gamma^0$ is a non-trivial interval and $\bar{\pparam} \in \Gamma^0$.  For each $\pparam \in \Gamma^0$ and each $d_j$, there is an $\covx^{j} \in \mathcal{X}$ such that $\cmap_{1,j}(\covx^{j})=\pparam$.  As a result, the following system of equations hold for each $\pparam \in \Gamma^0$:

\begin{align}{\label{right_corner}}
-\frac{d\Pr(d=d_2|x^2)}{dx}/\frac{d\cmap_{1,2}(x^2)}{dx}=& \aparam_{1}(\pparam)\aparam_{2}(\pparam)f(\pparam)  \nonumber \\ \nonumber
-\frac{d\Pr(d=d_3|x^3)}{dx}/\frac{d\cmap_{1,3}(x^3)}{dx}=& \aparam_{1}(\pparam)(1-\aparam_{2}(\pparam))\aparam_{3}(\pparam)f(\pparam)\\
-\frac{d\Pr(d=d_4|x^4)}{dx}/\frac{d\cmap_{1,4}(x^4)}{dx}=& \aparam_{1}(\pparam)(1-\aparam_{2}(\pparam))(1-\aparam_{3}(\pparam))\aparam_{4}(\pparam)f(\pparam)\\ \nonumber
\vdots \\ \nonumber
-\frac{d\Pr(d=d_{\dcon}|x^{d^*})}{dx}/\frac{d\cmap_{1,\dcon}(x^{d^*})}{dx}=& \aparam_{1}(\pparam)(1-\aparam_{2}(\pparam))(1-\aparam_{3}(\pparam))\dots (1-\aparam_{\dcon-1}(\pparam))f(\pparam) \nonumber,
\end{align}
The summation of these expressions recovers the quantity $\aparam_{1}(\pparam)f(\pparam)$.  Substitute this into the first equation to obtain $\aparam_2(\pparam)  = \aparam_2(1-\alpha(\pparam))$.  But $\aparam_2$ is already known, so $\alpha(\pparam)$ is identified on $\Gamma^0$.  Finally, since $\aparam_1(\pparam) = \aparam_1(1-\alpha(\pparam))$ is now identified, so is $f(\pparam)$ on $\Gamma^0$.

In the next step of the iteration, let $\bar \pparam^1 = \min_{\pparam \in \Gamma^0} \Gamma^0$ be the smallest value of $\pparam$ where $\alpha(\pparam)$ and $f(\pparam)$ are identified.  Define
\[
\Gamma^1_j \equiv \{  \pparam \in [0, \bar \pparam] : \exists \covx \in \mathcal{X} \text{ s.t. } \pparam = c_{1,j}(\covx)  \text{ and } c^*_j(\covx) \geq \bar \pparam^1\} \quad \text{and} \quad  \Gamma^1 \equiv \bigcap_{j=1}^{D} \Gamma^1_j.
\]
Then,  a similar system to \eqref{right_corner} holds $\forall \pparam \in \Gamma^1$, but may include additional terms. These terms are known, because they are functions of $f(\cdot)$ and $\alpha(\cdot)$ evaluated at $\pparam \in \Gamma^0$ (and also of $\{\aparam_j\}_{j=1}^D$).  We can therefore repeat the argument from the base case to establish that $\alpha(\pparam)$ and $f(\pparam)$ are identified on $\Gamma^1$.  We repeat this iterative procedure.  After a finite number of steps $T$, we obtain $\Gamma^T = [0,\bar \pparam]$ by Step 4; hence, $f(\cdot)$ and $\alpha(\cdot)$ are identified.

Edge Case I:
Suppose $d^* = d_1$ (the identity is known from Step 2). The following expressions hold for $\pparam$ close to $\bar{\pparam}$:
\begin{align}\label{eq:edge}
	-\frac{d\Pr(d=d_2|x^2)}{dx}/\frac{d\cmap_{1,2}(x^2)}{dx}=& \aparam_{2}(\pparam)f(\pparam)  \nonumber \\ \nonumber
	-\frac{d\Pr(d=d_3|x^3)}{dx}/\frac{d\cmap_{1,3}(x^3)}{dx}=& (1-\aparam_{2}(\pparam))\aparam_{3}(\pparam)f(\pparam)\\
	-\frac{d\Pr(d=d_4|x^4)}{dx}/\frac{d\cmap_{1,4}(x^4)}{dx}=& (1-\aparam_{2}(\pparam))(1-\aparam_{3}(\pparam))\aparam_{4}(\pparam)f(\pparam)\\ \nonumber
	\vdots \\ \nonumber
	-\frac{d\Pr(d=d_D|x^{D})}{dx}/\frac{d\cmap_{1,D}(x^{D})}{dx}=& (1-\aparam_{2}(\pparam))(1-\aparam_{3}(\pparam))\dots (1-\aparam_{D-1}(\pparam))\aparam_{D}(\pparam)f(\pparam),
\end{align}
The ratio of the first and second equations in System \eqref{eq:edge} yields:
\begin{align*}
	\frac{\aparam_{2}(\pparam)}{\aparam_{3}(\pparam)}&= (1-\aparam_{2}(\pparam))\frac{1}{A_3(\bs{\covx}_3(\pparam))},
\end{align*}
where $\bs{\covx}(\pparam) = (x^2(\pparam),x^3(\pparam),\dots,x^D(\pparam))$ is a known implicit function of $\pparam$ satisfying $\pparam = \cmap_{1,j}(x^j(\pparam))$ for $j=2,\dots,D$, and $A_3(\bs{\covx}(\pparam))$ is data:
\[
A_3(\bs{\covx}(\pparam)) = \frac{\frac{d\Pr(d=d_3|x^3(\pparam))}{dx}/\frac{d\cmap_{1,3}(x^3(\pparam))}{dx}}{\frac{d\Pr(d=d_2|x^2(\pparam))}{dx}/\frac{d\cmap_{1,2}(x^2(\pparam))}{dx}}.
\]
From this, we obtain
\begin{align*}
	\frac{\aparam_{2}}{\aparam_{3}}&=(1-\aparam_{2}(1-\alpha(\pparam))\frac{1}{A_3(\bs{\covx}(\pparam))}\\
	\frac{A_3(\bs{\covx}(\pparam))}{\aparam_{3}}&= \frac{1}{\aparam_{2}}-(1-\alpha(\pparam))\\
	\alpha(\pparam)&=1-\frac{1}{\aparam_{2}}+\frac{A_3(\bs{\covx}(\pparam))}{\aparam_{3}}\\
\alpha'(\pparam)&=\left[\frac{\partial A_3(\bs{\covx}(\pparam))}{\partial \covx^2}\frac{d\covx^2(\pparam)}{ d \pparam}+\frac{\partial A_3(\bs{\covx}(\pparam))}{\partial \covx^3}\frac{d\covx^3(\pparam)}{ d \pparam}\right]\frac{1}{\aparam_{3}}\equiv B_3(\bs{\covx}(\pparam))\frac{1}{\aparam_{3}},
\end{align*}
where $B_3(\bs{\covx}(\pparam))$ is a  known function of data.  A similar idea yields
\begin{align*}
	\alpha'(\pparam)&=B_4(\bs{\covx}(\pparam))\frac{1}{\aparam_{4}}\\
	\alpha'(\pparam)&=B_5(\bs{\covx}(\pparam))\frac{1}{\aparam_{5}}\\
	\vdots \\
	\alpha'(\pparam)&=B_D(\bs{\covx}(\pparam))\frac{1}{\aparam_{D}}.
\end{align*}
Hence, the ratios $\frac{\aparam_{3}}{\aparam_{4}}$, $\frac{\aparam_{4}}{\aparam_{5}}$, $\frac{\aparam_{5}}{\aparam_{6}}\dots$ are identified.  The ratio and the limit at the far end of the support of the third and fourth equations in System \eqref{eq:edge} yields
\begin{align*}
0 = 1-\frac{1}{\aparam_{3}} +  A_4(\bs{\covx}(\bar{\pparam}))\frac{1}{\aparam_{4}}\\
\aparam_3 = 1 - A_4(\bs{\covx}(\bar{\pparam}))\frac{\aparam_3}{\aparam_{4}},
\end{align*}
so $\aparam_{3}$ is identified and so are $\aparam_{4},\dots,\aparam_{D}$.  The ratio of the first and second equation in System \eqref{eq:edge}  identifies $\aparam_{2}$.  The proof continues on with Step 7. As such it does not require the second condition of the theorem.

Edge Case II: Suppose that $d^* = D$.  From Step 5 we obtain $\aparam_j$ for $j:1<j<D$.  Next, we show how to identify $\aparam_1$.  We can find $x^j$ such that $\pparam = c_{j,D}(\covx^j)$ is arbitrary close to zero satisfying $c_{j,k}(\covx^j)<0$ for all $k \neq j,D$, and so the following system of equations holds:
	\begin{align*}
\frac{d\Pr(d=d_{D-1}|\covx^{D-1})}{dx}/\frac{d\cmap_{D-1,D}(\covx^{D-1})}{dx}=& \aparam_{D-1}(\pparam)f(\pparam)  \nonumber \\ \nonumber
\frac{d\Pr(d=d_{D-2}|\covx^{D-2})}{dx}/\frac{d\cmap_{D-2,D}(\covx^{D-2})}{dx}=& (1-\aparam_{D-1}(\pparam))\aparam_{D-2}(\pparam)f(\pparam)\\
\frac{d\Pr(d=d_{D-3}|\covx^{D-3})}{dx}/\frac{d\cmap_{D-3,D}(\covx^{D-3})}{dx}=& (1-\aparam_{D-1}(\pparam))(1-\aparam_{D-2}(\pparam))\aparam_{D-3}(\pparam)f(\pparam)\\ \nonumber
\vdots \\ \nonumber
\frac{d\Pr(d=d_{1}|\covx^{1})}{dx}/\frac{d\cmap_{1,D}(\covx^{1})}{dx}=& (1-\aparam_{D-1}(\pparam))(1-\aparam_{D-2}(\pparam))\dots (1-\aparam_{2}(\pparam))\aparam_{1}(\pparam)f(\pparam). \nonumber
\end{align*}
The ratio of the first two equations yield
\[
A = \frac{ \aparam_{D-2}}{  \aparam_{D-1}} \cdot (1 - \aparam_{D-1}(1 + \alpha(\pparam)),
\]
where $A$, $\aparam_{D-2}$, and $\aparam_{D-1}$ are known terms; hence, $\alpha(0) = \lim_{\pparam \to 0} \alpha(\pparam)$ is identified.   Once $\alpha(0)$ is identified, the term $\aparam_1$ is identified from the ratio of the  first and last equations in the above system.  Finally, $f(\pparam)$ and $\alpha(\pparam)$  are identified by Step 7.

Edge Case III: Suppose $d^*=2$. By Steps 3, 5 and 6 all $\aparam_{j}$'s are identified. A modified version of Steps 4 and 7 can be applied.  We begin by starting at the lower end of the distribution with $d_D$ taking the role of $d_1$; $d_{D-1}$ taking the role of $d_2$, etc. Step 4 can be restated for
\[
\cmap_{j,D}(x)>c_j^{**}(\covx) \equiv \max_{k}\{\{\cmap_{k,j}(\covx)\}_{1\leq k<j}, \{\cmap_{j,k}(\covx)\}_{j<k< D}\}, ~\forall j .
\]
Finally Step 7 can be repeated starting at the lower end of the support and building on the following equations
	\begin{align*}
	\frac{d\Pr(d=d_{D-1}|\covx^{D-1})}{dx}/\frac{d\cmap_{D-1,D}(\covx^{D-1})}{dx}=& \aparam_{D}(\pparam) \aparam_{D-1}(\pparam)f(\pparam)  \nonumber \\ \nonumber
	\frac{d\Pr(d=d_{D-2}|\covx^{D-2})}{dx}/\frac{d\cmap_{D-2,D}(\covx^{D-2})}{dx}=& \aparam_{D}(\pparam)(1-\aparam_{D-1}(\pparam))\aparam_{D-2}(\pparam)f(\pparam)\\
	\frac{d\Pr(d=d_{D-3}|\covx^{D-3})}{dx}/\frac{d\cmap_{D-3,D}(\covx^{D-3})}{dx}=& \aparam_{D}(\pparam)(1-\aparam_{D-1}(\pparam))(1-\aparam_{D-2}(\pparam))\aparam_{D-3}(\pparam)f(\pparam)\\ \nonumber
	\vdots \\ \nonumber
	\frac{d\Pr(d=d_{3}|\covx^{3})}{dx}/\frac{d\cmap_{1,3}(\covx^{3})}{dx}=& \aparam_{D}(\pparam) (1-\aparam_{D-1}(\pparam))(1-\aparam_{D-2}(\pparam))\dots (1-\aparam_{3}(\pparam))f(\pparam). \nonumber
\end{align*}

\end{proof}

\begin{proof}[Proof of Theorem \ref{Broad-alternative specific_ARC}]
	
Let $\pparam$, $\covx$, $\mathcal{N}_{ \epsilon}(\covx) = \{ \covx' : \| \covx' -  \covx\| <  \epsilon\}$ satisfy the conditions in the theorem. Then $\pparam = c_{j,k}(\covx)$ for all $j,k$. For any pair of alternatives $(d_j,d_k)$ we can perturb $\covx_k$, $ k \notin \{j,d^*\}$, and $\covx_l$, $\forall l \notin \{j,d^*,k\}$, so that the resulting $x' \in \mathcal{N}_{ \epsilon}(\covx)$ is such that
\begin{align*}
& U_{ \pparam}(d_k, \covx'_k)> U_{ \pparam}(d_j,  \covx_j)  \\
&U_{ \pparam}(d_j, \covx_j) > U_{ \pparam}(d_l,  \covx'_l), \quad  \forall l \in \mathcal{D}\setminus\{j,k,d^*\}.
\end{align*}
And we can do another perturbation of $\covx_l$, $\forall l \notin \{j,d^*\}$, so that the resulting $x'' \in \mathcal{N}_{ \epsilon}(\tilde \covx)$ is
such that
\begin{align*}
&U_{ \pparam}(d_j, \covx_j) > U_{ \pparam}(d_l,\covx''_l), \quad  \forall l \in \mathcal{D}\setminus\{j,d^*\}.
\end{align*}
Then
\begin{align*}
 \frac{\partial \Pr(d=d_{j}| \covx')}{\partial \covx_{d^*}}
& = \aparam_j( \covx_j,\pparam) (1-\aparam_k( \covx'_k,\pparam)) f(\pparam)\frac{\partial c_{j,d^*}( \covx)}{\partial x_{d^*}}\\
\frac{\partial \Pr(d=d_{j}| \covx'')}{\partial \covx_{d^*}}
& = \aparam_j( \covx_j,\pparam) f(\pparam)\frac{\partial c_{j,d^*}( \covx)}{\partial x_{d^*}}.
\end{align*}
Taking the ratio of the expressions above identifies $(1-\aparam_k( \covx'_k,\pparam))$. By continuity we identify $\aparam_k( \covx_k,\pparam)$. Identical steps identify $\aparam_j( \covx_j,\pparam),~\forall j\neq d^*$, and hence $f(\pparam)$.

\end{proof}

\begin{proof}[Proof of Proposition \ref{prop:monotonicity}]

	Take any non-empty consideration set $\Ks$. For a given preference coefficient $\pparam$, let $j_{\Ks}(\covx,\pparam)$ denote the identity of the best alternative in this consideration set. By the natural ordering, $j_{\Ks}(\covx,\pparam)$ is an increasing step function in $\pparam$. Hence, $I(j_{\Ks}(\covx,\pparam)\leq J)$ is a decreasing step function.  The term $\Pr\left(\bigcup_{j=1}^J \dor_j \bigg|\covx,\pparam\right)$ is a non-negatively weighted sum of $I(j_{\Ks}(\covx,\pparam)\leq J)$. Hence it is decreasing in $\pparam$.
\end{proof}

\begin{proof}[Proof of Proposition \ref{fact:modGeneric}]
Consider a limited consideration model with preferences $U_{\pparam}(d_j,\covx)$ and consideration probability $ \mathcal{Q}_{\pparam}^{\covx}(\cal{K})$, $\cal{K} \subset \cal{D}$.  The optimal choice from $\Dc$ conditional on the DM facing the consideration set $\mathcal{K} \neq \emptyset$ is the alternative with the largest value of $U_{\pparam}(d_j,\covx)$ subject to $j \in \mathcal{K}$.   This is the same solution as the one that maximizes $V_{\pparam}(d_j,\covx,\epsilon_j)$ where $\epsilon_j = 0$ for all $j \in \mathcal{K}$ and $\epsilon_j = -\infty $ for all $j \in \Dc \setminus \mathcal{K}$.  Finally, since conditional on $\covx$ the consideration set $\Kc$ has the same distribution as the set of alternatives with $\epsilon_j=0$ (this is by construction), the  limited consideration model and this ORUM model yield the same model predictions, and hence they are equivalent.
\end{proof}

\section{Application: Verifying Cutoff Order}\label{app:AppAssumption}

We start by recalling that CARA and CRRA utility functions satisfy the following basic property \citep[see, e.g.,][]{Pratt1964,Barseghyan2018}.\footnote{%
	This property is equivalent to condition (e) in \citet[Theorem 1]{Pratt1964}. As shown there, it is equivalent to assuming that an increase in $\pparam $ corresponds to an increase in the coefficient of absolute risk aversion.}
\begin{lemmaB}
	For any $y_{0}>y_{1}>y_{2}>0$, the ratio $R(y_{0},y_{1},y_{2})\equiv \frac{u_{\pparam}(y_{1})-u_{\pparam}(y_{2})}{u_{\pparam}(y_{0})-u_{\pparam}(y_{1})} $ is strictly
	increasing in $\pparam $.
\end{lemmaB}
It follows that CARA and CRRA utility functions also satisfy a slightly extended version of the property above:
\begin{lemmaB}\label{propCARA}
	For any $y_{0}>y_{1}>y_{2}>y_{3}>0$, the ratio $M_{\pparam}(y_{0},y_{1},y_{2},y_{3})\equiv \frac{u_{\pparam}(y_{2})-u_{\pparam}(y_{3})}{u_{\pparam}(y_{0})-u_{\pparam}(y_{1})} $ is strictly
	increasing in $\pparam $.
\end{lemmaB}
\begin{proof}
	\begin{align*}
	M_{\pparam}(y_{0},y_{1},y_{2},y_{3})&=\frac{u_{\pparam}(y_{2})-u_{\pparam}(y_{3})}{u_{\pparam}(y_{0})-u_{\pparam}(y_{1})}
    =\frac{u_{\pparam}(y_{2})-u_{\pparam}(y_{3})}{u_{\pparam}(y_{1})-u_{\pparam}(y_{2})}\times
	 \frac{u_{\pparam}(y_{1})-u_{\pparam}(y_{2})}{u_{\pparam}(y_{0})-u_{\pparam}(y_{1})}\\
	 &=R_{\pparam}(y_{1},y_{2},y_{3})R_{\pparam}(y_{0},y_{1},y_{2})
	\end{align*}
\end{proof}
For our application, we show that $\cmap_{1,j}(\bar{p},\mu)<\cmap_{1,j+1}(\bar{p},\mu)$ for any $j \geq 2$ under both CARA and CRRA preferences.
\begin{propositionB}
Suppose deductibles and prices are such that
\[
\frac{p_{1}-p_{j}}{p_{1}-p_{j+1}}<\frac{d_{1}-d_{j}}{d_{1}-d_{j+1}}
\]
and $d_k + p_k < w$ for all $k$.
Under either CARA or CRRA expected utility preferences, the cutoff mapping is unique and satisfies $\cmap_{1,j}(\bar{p},\mu)<\cmap_{1,j+1}(\bar{p},\mu)$ for all $j>1$.
\end{propositionB}
Note that in a perfectly competitive markets where additional coverage is simply proportional to its price both ratios will be equal. In practice, however, one might expect that with some market power the prices increase faster than then coverage, and hence
\[
\frac{p_{1}-p_{j}}{p_{1}-p_{j+1}}<\frac{d_{1}-d_{j}}{d_{1}-d_{j+1}}
\]
This is exactly what we find in our data (as well as for a larger number of firms appearing in \cite{Barseghyan2011}).

\begin{proof}
We start with CARA preferences. The existence and the uniqueness of $c_{j,k}(\covx)$ for all $j<k$ follows directly from the Lemma \ref{propCARA}. Indeed note that $p_j<p_k<p_k+d_k<p_j+d_j$.$\footnote{If $p_k+d_k>p_j+d_j$, then alterantive $j$ first order stochastically dominates $k$ and hence the cuttoff is $+\infty$.}$ At the cutoff the DM is indifferent between lotteries $j$ and $k$. Equating two expected utilities and rearranging we have that
\begin{align}\label{Eq:prop5}
\frac{e^{-\pparam(w-p_k-d_k)}-e^{-\pparam(w-p_j-d_j)}}{e^{-\pparam (w-p_j)}-e^{-\pparam(w-p_k)}}&= \frac{1-\mu}{\mu},
\end{align}
where $w$ is the DM's initial wealth. By Lemma \ref{propCARA}, the L.H.S. of Equation \ref{Eq:prop5} is strictly monotone in $\pparam$, and it tends to $+\infty$ when $\pparam$ goes to $+\infty$ and to zero when $\pparam$ goes to $-\infty$. It follows that there exists a unique $\pparam$, i.e the cutoff $c_{j,k}(\covx)$, that solves the Equation \ref{Eq:prop5}. Moreover, since the L.H.S. is strictly monotone in $\pparam$ it follows from the Implicit Function Theorem that $c_{j,k}(\covx)$ is continuous in $\mu$ and $\bar{p}$.

The next step is to establish  $c_{1,j}(\bar{p},\mu)<c_{1,j+1}(\bar{p},\mu)$, $j>1$. For the purpose of obtaining a contradiction, suppose that there exists $(\bar{p},\mu)$ and an $j$ such that $c_{1,j}(\bar{p},\mu)\geq c_{1,j+1}(\bar{p},\mu)$. Since the expected utility of lottery $k$ is proportional to
\begin{align*}	
EU_{\pparam}(L_{k}) &\propto-e^{\pparam p_{k}}\left(1-\mu  +\mu e^{\pparam d_k}\right),
\end{align*}
there exists  $\nu = c_{1,j}(\bar{p},\mu) \geq c_{1,j+1}(\bar{p},\mu)$ such that
\begin{align*}	
\frac{1-\mu  +\mu e^{\pparam d_1}}{1-\mu  +\mu e^{\pparam d_{j}}} e^{\pparam(g_1-g_{j})\bar{p}} &=1\leq \frac{1-\mu  +\mu e^{\pparam d_1}}{1-\mu  +\mu e^{\pparam d_{j+1}}} e^{\pparam(g_1-g_{j+1})\bar{p}}
\end{align*}
Taking logs yields
\begin{align*}	
\log\left(\frac{1-\mu  +\mu e^{\pparam d_1}}{1-\mu  +\mu e^{\pparam d_{j}}}\right)&=- \pparam(g_1-g_{j})\bar{p} \\
\log\left(\frac{1-\mu  +\mu e^{\pparam d_1}}{1-\mu  +\mu e^{\pparam d_{j+1}}}\right)&\geq- \pparam(g_1-g_{j+1})\bar{p}.
\end{align*}
Dividing through and using the fact that $- \pparam(g_1-g_{j+1})\bar{p} \geq - \pparam(g_1-g_{j})\bar{p} \geq 0$ yields
\begin{align*}	
\frac{\log\left(\frac{1-\mu  +\mu e^{\pparam d_1}}{1-\mu  +\mu e^{\pparam d_{j}}}\right)}
{\log\left(\frac{1-\mu  +\mu e^{\pparam d_1}}{1-\mu  +\mu e^{\pparam d_{j+1}}}\right)} &\leq\frac{g_1-g_{j}}{g_1-g_{j+1}}.
\end{align*}
The R.H.S. is less than one. We claim that the L.H.S. is monotonically decreasing in $\mu<1$. To show this, denote $\hat{\mu}=\frac{1-\mu}{\mu}$, $\Delta_1=e^{\pparam d_1}$, $\Delta_j=e^{\pparam d_j}$, and $\Delta_{j+1}=e^{\pparam d_{j+1}}$ to rewrite the L.H.S. as follows
\begin{align*}	
\text{L.H.S} = f\left(\frac{1-\mu}{\mu}\right) = f(\hat \mu) \equiv \frac{\log(\Delta_1+\hat{\mu})-\log(\Delta_j+\hat{\mu})}{\log(\Delta_1+\hat{\mu})-\log(\Delta_{j+1}+\hat{\mu})}.
\end{align*}
First, we show that the above expression is monotonically increasing in $\hat{\mu}$. Observe that
\footnotesize
\begin{align*}	
\frac{f'(\hat \mu)}{f(\hat\mu)} &= \left(\frac{1}{\Delta_1+\hat{\mu}}-\frac{1}{\Delta_j+\hat{\mu}}\right)\frac{1}{\log(\Delta_1+\hat{\mu})-\log(\Delta_j+\hat{\mu})}-
\left(\frac{1}{\Delta_1+\hat{\mu}}-\frac{1}{\Delta_{j+1}+\hat{\mu}}\right)\frac{1}{\log(\Delta_1+\hat{\mu})-\log(\Delta_{j+1}+\hat{\mu})}
\end{align*}
\normalsize
After relabeling $\Lambda_1=-\log(\Delta_1+\hat{\mu})$, $\Lambda_j=-\log(\Delta_j+\hat{\mu})$ and $\Lambda_{j+1}=-\log(\Delta_{j+1}+\hat{\mu})$ we obtain
\begin{align*}	
\frac{f'(\hat \mu)}{f(\hat\mu)} = &\frac{e^{\Lambda_1}-e^{\Lambda_{j+1}}}{\Lambda_1-\Lambda_{j+1}}-\frac{e^{\Lambda_1}-e^{\Lambda_j}}{\Lambda_1-\Lambda_j}.
\end{align*}
Since $\Lambda_1<\Lambda_j<\Lambda_{j+1}$ and exponential function is convex,  the expression above is positive. Thus, the derivative of $f\left( \frac{1-\mu}{\mu} \right)$ W.R.T. $\mu$ is negative as claimed.  That is, $f\left( \frac{1-\mu}{\mu} \right)$ achieves its lowest value at $\mu=1$ and is equal to $\frac{d_1-d_j}{d_1-d_{j+1}}$. Finally, a contradiction is obtained since
\begin{align*}	
\text{L.H.S}  \geq \min_{\mu}f\left( \frac{1-\mu}{\mu} \right) =  \frac{d_1-d_j}{d_1-d_{j+1}} > \frac{g_1-g_{j}}{g_1-g_{j+1}} = \text{R.H.S},
\end{align*}
where the strict inequality is by assumption.  Therefore, $c_{1,j}(\bar{p},\mu)< c_{1,j+1}(\bar{p},\mu)$ under CARA as claimed.

Under CRRA, $c_{j,k}(\bar{p},\mu)$ exist and are continuous exactly for the same reasons as under CARA. It remains to establish that $c_{1,j}(\bar{p},\mu)< c_{1,j+1}(\bar{p},\mu)$.  For the purpose of obtaining a contradiction, suppose $c_{1,j}(\bar p,\mu) \geq c_{1,j+1}(\bar p,\mu)$ for some $(\bar p,\mu)$. Consider the following Taylor expansion for the CRRA Bernoulli utility function $u_{\pparam}(w) \equiv  \frac{ w^{1-\pparam}}{1-\pparam} $:
\begin{align*}
\frac{ (w-p_k)^{1-\pparam}}{1-\pparam}
& =  \frac{ w^{1-\pparam}}{1-\pparam}
      +  \frac{w^{-\pparam}}{1!}(-p_k)
        -\pparam \frac{w^{-\pparam-1}}{2!} (-p_k)^2
      +  \pparam(\pparam +1) \frac{w^{-\pparam-2}}{3!} (-p_k)^3
      + \dots
\end{align*}
Or, equivalently,
\begin{align*}
(1-\pparam)w^{\pparam-1} [u_{\pparam}(w-p_k) - u_{\pparam}(w) ]
& =   (\pparam-1)\frac{1}{1!}w^{-1}p_k
        +(\pparam-1)\pparam \frac{w^{-2}}{2!}p_k^2
      +  (\pparam-1)\pparam(\pparam +1) \frac{w^{-3}}{3!} p_k^3
      + \dots
\end{align*}
Hence,
\begin{align*}	
EU_{\pparam}(L_{k}) \propto \left( 1-\mu \right)\sum_{t=1}^{\infty}\omega_t(\pparam) p_{k}^t +\mu \sum_{t=1}^{\infty}\omega_t(\pparam) \left(p_{k}+d_{k}\right)^t.
\end{align*}
where $\omega_t(\pparam) \equiv (t!w^{t})^{-1}\prod_{t'=0}^{t-1}(\pparam-1+t')<0$ when $\pparam \in (0,1)$.  When $\pparam>1$, $\omega_t(\pparam)>0$ but the factor premultiplying  $u_{\pparam}(w-p_k)$ above is negative, so we would still come to the same conclusion that $EU_{\pparam}(L_{k})$ is proportional to a power series with coefficients $\tau_t(\pparam) = - \omega_t(\pparam) <0$. The power series are absolutely convergence provided that $p_k+d_k<w$, so the difference in the power series for $EU_{\pparam}(L_{j})$ and $EU_{\pparam}(L_{k})$ is equal to the sum of the difference:
\begin{align*}	
EU_{\pparam}(L_{j})-EU_{\pparam}(L_{k}) & \propto \left( 1-\mu \right)\sum_{t=1}^{\infty}\omega_t(\pparam) \left(p_{j}^t-p_{k}^t\right) +\mu \sum_{t=1}^{\infty}\omega_t(\pparam) \left(\left(p_{j}+d_{j}\right)^t-\left(p_{k}+d_{k}\right)^t \right)\\
&= \left(p_{j}-p_{k}\right) \left( 1-\mu \right)\sum_{t=1}^{\infty}\omega_t(\pparam) \sum_{h=0}^{t-1}p_{j}^h p_{k}^{t-h}+\\ &+\left(\left(p_{j}-p_{k}\right)+\left(d_{j}-d_{k}\right) \right) \mu \sum_{t=1}^{\infty}\omega_t(\pparam) \sum_{h=0}^{t-1} \left(p_{j}+d_{j}\right)^{h}\left(p_{k}+d_{k}\right)^{t-h}.
\end{align*}
The condition $\pparam=c_{1,j}(\bar{p},\mu)\geq c_{1,j+1}(\bar{p},\mu)$ implies
\begin{align}\label{eq:crraineq}
\frac{p_{1}-p_{j}}{p_{1}-p_{j+1}}\geq &\frac{p_{1}-p_{j}+d_{1}-d_{j}}{p_{1}-p_{j+1}+d_{1}-d_{j+1}} \delta(\pparam),
\end{align}
where
\begin{align*}
\delta(\pparam) \equiv & \frac{\sum_{t=1}^{\infty}\omega_t(\pparam) \sum_{h=0}^{t-1} \left(p_{1}+d_{1}\right)^{h}\left(p_{j}+d_{j}\right)^{t-h}}{\sum_{t=1}^{\infty}\omega_t(\pparam) \sum_{h=0}^{t-1} \left(p_{1}+d_{1}\right)^{h}\left(p_{j+1}+d_{j+1}\right)^{t-h}} \frac{\sum_{t=1}^{\infty}\omega_t(\pparam) \sum_{h=0}^{t-1}p_{1}^h p_{j+1}^{t-h}}{\sum_{t=1}^{\infty}\omega_t(\pparam) \sum_{h=0}^{t-1}p_{1}^h p_{j}^{t-h}}.
\end{align*}
Under the assumption $\pparam=c_{1,j}(\bar{p},\mu)\geq c_{1,j+1}(\bar{p},\mu)$ it is also the case that $\pparam=c_{1,j}(\bar{p},\mu)\geq c_{1,j+1}(\bar{p},\mu)\geq c_{j,j+1}(\bar{p},\mu)$ by Fact \ref{A:CO}.  Hence, $p_{j}+d_{j}>p_{j+1}+d_{j+1}$. Indeed otherwise $p_{j+1}-p_{j}>d_{j}-d_{j+1}$ is a violation of the first order stochastic dominance.  Taken with  $p_{j+1}>p_{j}$, it follows that $\delta(\pparam)>1$.  Finally, a contradiction will be obtained if
\begin{align*}	
\frac{p_{j}-p_{1}}{p_{j+1}-p_{1}}&\leq \frac{p_{j+1}-p_{1}+d_{1}-d_{j}}{p_{1}-p_{j+1}+d_{1}-d_{j+1}},
\end{align*}
since then Equation \eqref{eq:crraineq} will not hold.  Re-arranging this expression we obtain:
\begin{align*}	
\frac{p_{1}-p_{j+1}+d_{1}-d_{j+1}}{p_{j+1}-p_{1}}&\leq\frac{p_{1}-p_{j}+d_{1}-d_{j}}{p_{j}-p_{1}}\\
\frac{d_{1}-d_{j+1}}{p_{j+1}-p_{1}}&\leq\frac{d_{1}-d_{j}}{p_{j}-p_{1}}\\
\frac{p_{1}-p_{j}}{p_{1}-p_{j+1}}&\leq\frac{d_{1}-d_{j}}{d_{1}-d_{j+1}}.
\end{align*}
The latter inequality holds by assumption. It follows that  $c_{1,j}(\bar{p},\mu)<c_{1,j+1}(\bar{p},\mu)$, $j>1$.

\end{proof}
\section{Monetary Cost of Limited Consideration}\label{sec:monetarycost}
We view limited consideration as a mechanism that constrains households from achieving their first-best alternative either because the market setting forces some alternatives to become more salient than others (e.g. agent effects) or because of time or psychological costs that prevent the household from evaluating all alternatives in the choice set. Regardless of the underlying mechanism(s) of limited consideration, we can quantify its \emph{monetary} cost within our framework. We ask, \emph{ceteris paribus}, how much money the households ``leave on the table" when choosing deductibles in property insurance under limited consideration rather than under full consideration. This is likely to be a lower bound on actual monetary losses arising from limited consideration, because insurance companies might be exploiting sub-optimality of households choices when setting prices or choosing menus.

We measure the monetary costs of limited consideration as follows. For each household we compute (the expected value of) the certainty equivalent of the lottery associated with the households' optimal choice, as well as of the one associated with their choice under limited consideration.\footnote{Certainty equivalent of the lottery is defined as the minimum amount they are willing to accept in lieu of the lottery. In our case, for alternative $j$, it is simply $ce_j \equiv \frac{1}{\pparam} \ln [ (1-\mu) \exp({\pparam} p_j ) + \mu \exp({\pparam}(p_j + d_j))]$.} We then take the difference between these certainty equivalent values and average them across all households in the sample. On average, we find that households lose $\$50$ dollars across the three deductibles because of limited consideration. See Table \ref{F:WelfareWide} for variation conditional on demographic characteristics and insurance score.
We also find wide dispersion in loss across households (see Figure \ref{F:WelfareWide}).  In particular, the $10^{th}$ percentile of losses is $\$31$ and the $90^{th}$ is $\$73$.
\clearpage
\section{Data}
\renewcommand\thetable{\thesection.\arabic{table}}
\setcounter{table}{0}
\begin{normalsize}
\begin{table}[!ht]\centering
\begin{threeparttable}
\caption{Descriptive Statistics}
\label{T:descripive}
	\begin{tabular}{l c c c c}
		\hline \hline\\[-1.5ex]
		Variable       & Mean & Std. Dev.  & 1st $\%$  & 99th $\%$ \\\\[-1.5ex]
		\hline \\[-1.5ex]
		Age            & 53.3 & 15.7 & 25.4 & 84.3  \\\\[-1.5ex]
		Female         &0.40  &      &      &       \\\\[-1.5ex]
		Single		   &0.22  &      &      &       \\\\[-1.5ex]		
		Married 	   &0.55  &      &      &       \\\\[-1.5ex] 		
		Second Driver  &0.43  &      &      &       \\\\[-1.5ex] 		
		Insurance Score& 767  &112   & 532  & 985  \\\\[-1.5ex] 		
		\hline
	\end{tabular}
\end{threeparttable}
\end{table}
\end{normalsize}

\begin{normalsize}
\begin{table}[!ht]\centering
\begin{threeparttable}
\caption{Frequency of Deductible Choices Across Contexts}
\label{T:deductible}
	\begin{tabular}{l c c c c c c}
		\hline \hline \\[-1.5ex]
		Deductible     &1000  & 500  & 250  & 200  &100  &  50  \\\\[-1.5ex]
		\hline \\[-1.5ex]
		Collision      &0.064  &0.676  & 0.122 & 0.129 &0.009 &      \\\\[-1.5ex]
		Comprehensive  &0.037  &0.430  & 0.121 & 0.329 &0.039 & 0.044 \\\\[-1.5ex]
		Home 		   &0.176  &0.559  & 0.262 &      &0.002 &      \\\\[-1.5ex] 		
		\hline
	\end{tabular}
\end{threeparttable}
\end{table}
\end{normalsize}

\begin{normalsize}
\begin{table}[!ht]\centering
\begin{threeparttable}
\caption{Deductible Rank Correlations Across Contexts}
\label{T:correlation}
	\begin{tabular}{l c c c }
		\hline \hline\\[-1.5ex]
		     &Collision  & Comprehensive   & Home  \\
		\cmidrule(lr){2-4}
         \\[-1.5ex]
		Collision      &1     &      &         \\\\[-1.5ex]
		Comprehensive  &0.61  &1     &         \\\\[-1.5ex]
		Home 		   &0.37  &0.35  & 1       \\\\[-1.5ex] 		
		\hline
	\end{tabular}
\end{threeparttable}
\end{table}
\end{normalsize}

\begin{normalsize}
\begin{table}[!ht]\centering
\begin{threeparttable}
\caption{Joint Distribution of Auto Deductibles}
\label{T:autodist}
\begin{tabular}{l c c c c c c}
\hline \hline  \\[-1.5ex]
    & \multicolumn{6}{c}{\underline{Comprehensive}}  \\\\[-1.5ex]
Collision	&	1000	&	500	&	250	&	200	&	100	&	50	 \\
\cmidrule(lr){1-1}
\cmidrule(lr){2-7}
\\[-1.5ex]
1000&	3.71	&	1.93	&	0.18	&	0.44	&	0.05	&	0.04	\\\\[-1.5ex]
500	&	0	&	40.99	&	6.46	&	17.84	&	1.27	&	1.00	\\\\[-1.5ex]
250	&	0	&	0.04	&	5.42	&	4.55	&	1.28	&	0.94	\\\\[-1.5ex]
200	&	0.01	&	0.05	&	0.03	&	9.99	&	1.07	&	1.78	\\\\[-1.5ex]
100	&	0	&	0	&	0	&	0.04	&	0.23	&	0.66	\\\\[-1.5ex]
		\hline		
	\end{tabular}
\caption*{The distribution is reported in percent.}
\end{threeparttable}
\end{table}
\end{normalsize}

\begin{normalsize}
\begin{table}\centering
\begin{threeparttable}
\caption{Average Premiums Across Coverages}
\label{T:premium}
	\begin{tabular}{l c c c c c c}
		\hline \hline \\[-1.5ex]
		 Deductible    &1,000  & 500 & 250 & 200 &100 &  50\\\\[-1.5ex]
		\hline  \\[-1.5ex]
		Collision      & 145  & 187 & 243 & 285 &321 &    \\\\[-1.5ex]
		Comprehensive  & 94   & 117 & 147 & 155 &178 & 224\\\\[-1.5ex]
		Home 		   & 594  & 666 & 720 &     &885 &    \\\\[-1.5ex] 		
		\hline
	\end{tabular}
\end{threeparttable}
\end{table}
\end{normalsize} 
\newpage
\clearpage
\section{Empirical Results: Figures and Tables}\label{app:tables}
\subsection{The $\modA$ Model with Observable Demographics}\label{app:observables}
While it is ideal to control for households' observable characteristics non-parametrically, it is data demanding. In practice, it is commonly assumed that household characteristics shift the expected value of the preference-coefficient distribution.$\footnote{For exmaple, \cite{Cohen2007} assume that $\log \pparam_i = \mathbf{Z}_i \gamma +\varepsilon_i$, where $\mathbf{Z}_i$ are the observables for household $i$ and $\varepsilon_i$ is i.i.d. N(0, $\sigma^2$). Hence, $E(\pparam_i)=e^{\mathbf{Z}_i \gamma +\sigma^2/2}.$}$ We adopt the same strategy here by assuming that for each household $i$, $\log\frac{\beta_{1,i}}{\beta_{2}}= \mathbf{Z}_i\gamma$, where $\gamma$ is an unknown vector to be estimated. The terms $\beta_{1,i}$ and $\beta_{2}$ denote the parameters of the Beta distribution, where $\beta_{1,i}$ is household specific and $\beta_2$ is common across households.  The preference coefficients are random draws from a distribution with an expected value that is a function of the observable characteristics given by
$ E(\pparam_i)=\frac{\beta_{1,i}}{\beta_{1,i}+\beta_{2}}\bar{\pparam}=\frac{e^{\mathbf{Z}_i\gamma}}{1+e^{\mathbf{Z}_i \gamma}}\bar{\pparam}
$.\footnote{If, instead, we assume $\log\frac{\beta_{2,i}}{\beta_{1}}= \mathbf{Z}_i\tilde \gamma$, then we arrive to the same expression for the expected value with the exception that $\tilde \gamma = -\gamma$.} The results of this estimation are in line with our first estimation. (See Column 2 in Table \ref{T:mleARCcoll}, as well as Figures \ref{F:ARCaggX} and \ref{F:ARCbygroupX}.) The new observation here is that the model closely matches the distribution of choices across various sub-populations in the sample including gender, age, credit worthiness, and contracts with multiple drivers.  The model's ability to match these conditional distributions can be attributed, in part, to the dependence of risk preferences on household characteristics. The model is, however, fairly parsimonious as the consideration parameters are restricted to be the same across all households. Finally, estimated consideration probabilities are close in magnitude to those estimated above. In particular, the highest deductibles ($\$1,000$ and $\$500$) are most likely to be considered, with respective frequencies of $0.94$ and $0.92$.  The remaining alternatives are considered at much lower frequencies.
\clearpage
\subsection{Figures}
\renewcommand\thefigure{\thesection.\arabic{figure}}
\setcounter{figure}{0}
\begin{figure}[!htbp]\centering
    \caption{The $\modA$ Model with Observable Demographics}
    \label{F:ARCaggX}
    \includegraphics[trim={5cm 0 0 0.5cm},scale=0.4]{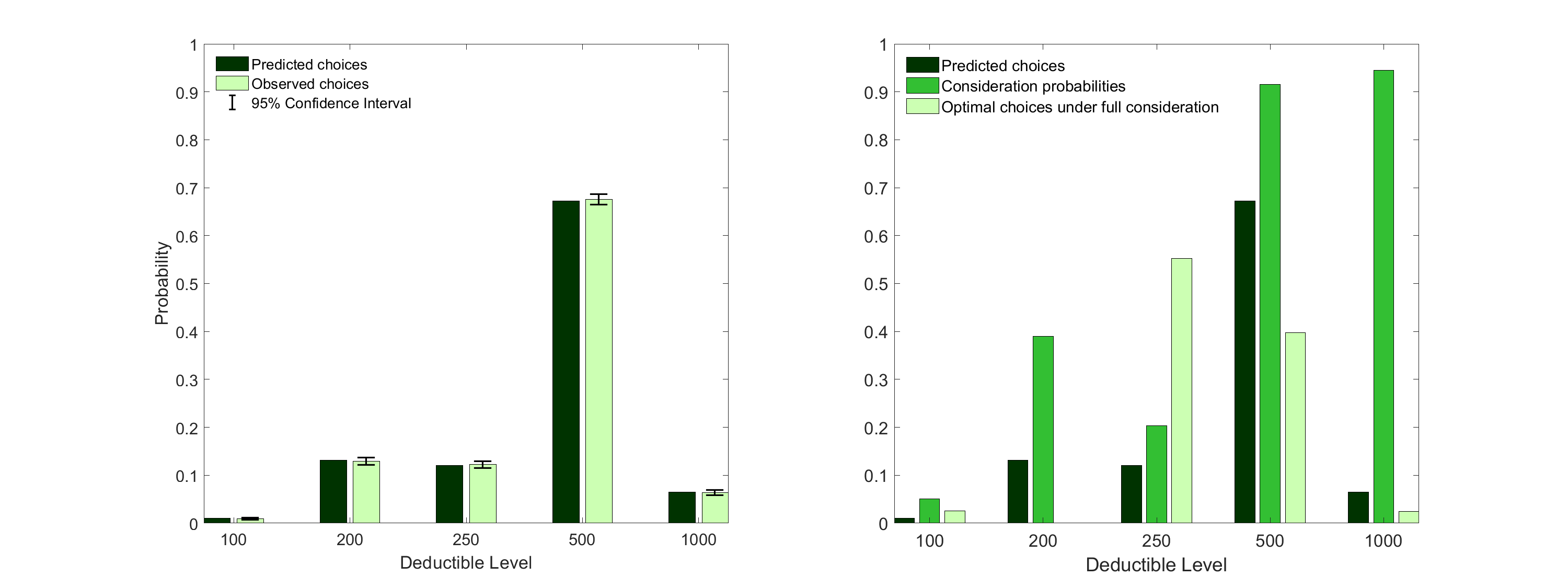}
    \captionsetup{font={footnotesize,sf},justification=raggedright}
    \caption*{The first panel reports the distribution of predicted and observed choices.  The second panel displays consideration probabilities and the distribution of optimal choices under full consideration. }
\end{figure}

\begin{figure}[!htbp]\centering
    \caption{The $\modA$ Model with Observable Demographics:  Conditional Distributions}
    \label{F:ARCbygroupX}
    \includegraphics[trim={5cm 0 0 0cm},scale=0.4]{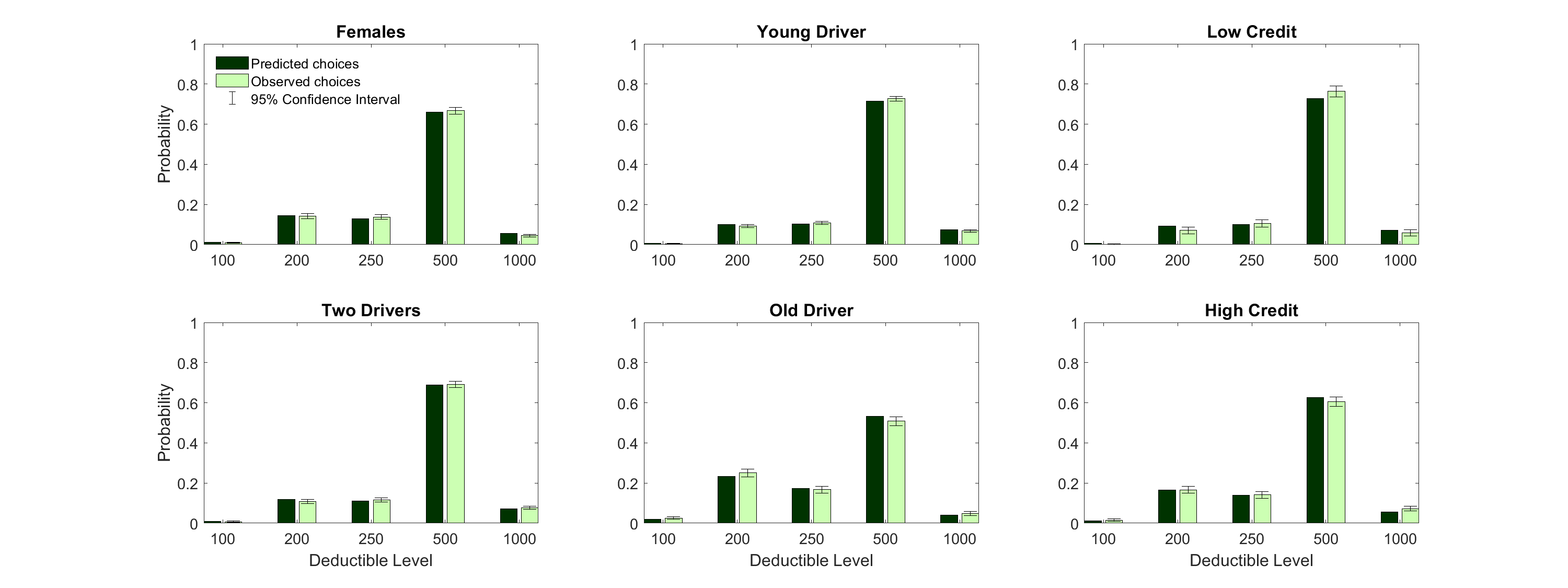}
\end{figure}

\begin{figure}\centering
	\caption{The Mixed Logit}
	\label{F:RCLagg}
	\includegraphics[trim={2.5cm 0 0 0.5cm},scale=0.225]{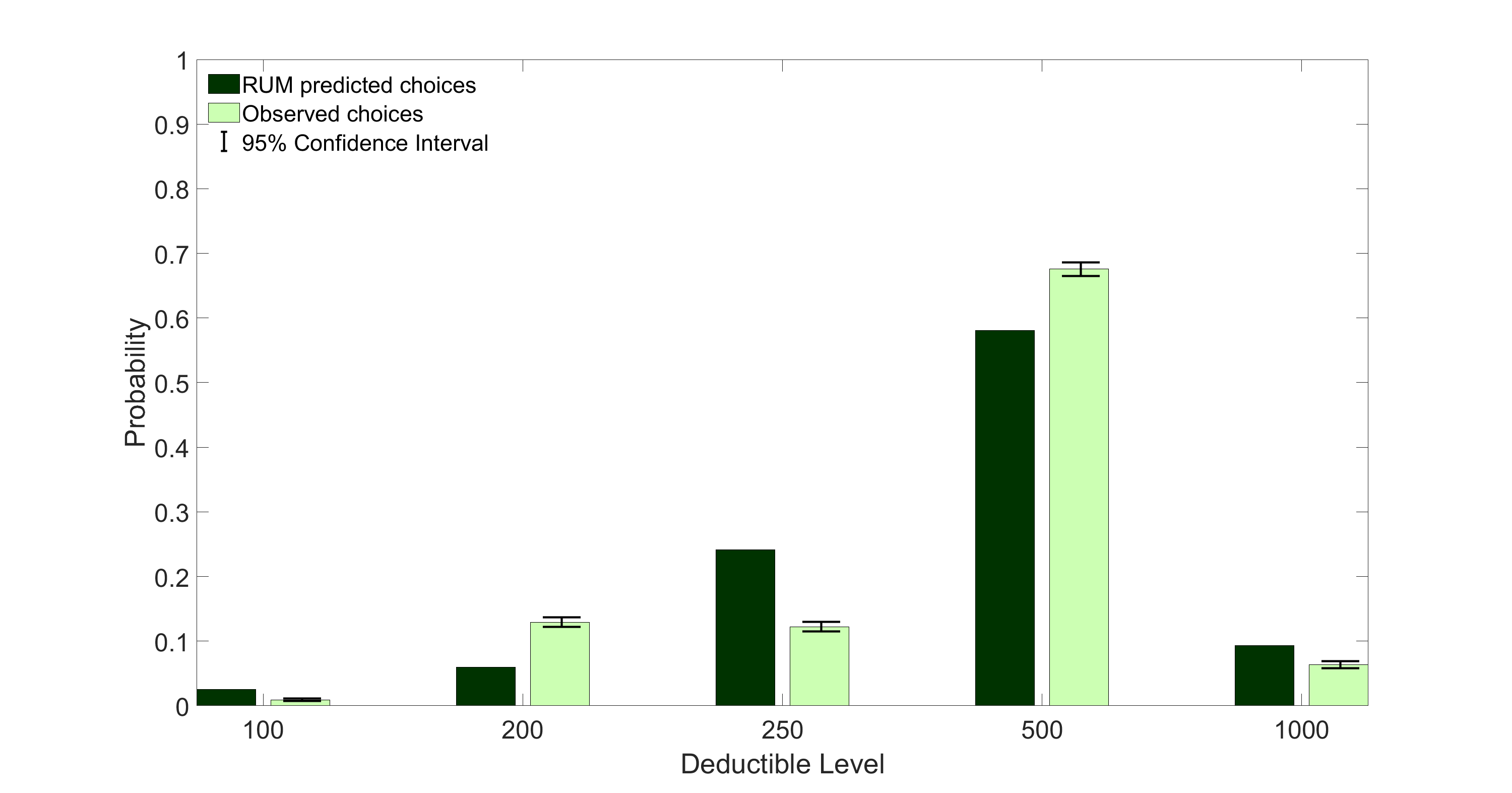}
\end{figure}

\begin{figure}[!htbp]\centering
	\caption{The Mixed Logit, Three Coverages}
	\label{F:RUMwide}
	\includegraphics[trim={5cm 0 0 0cm},scale=0.4]{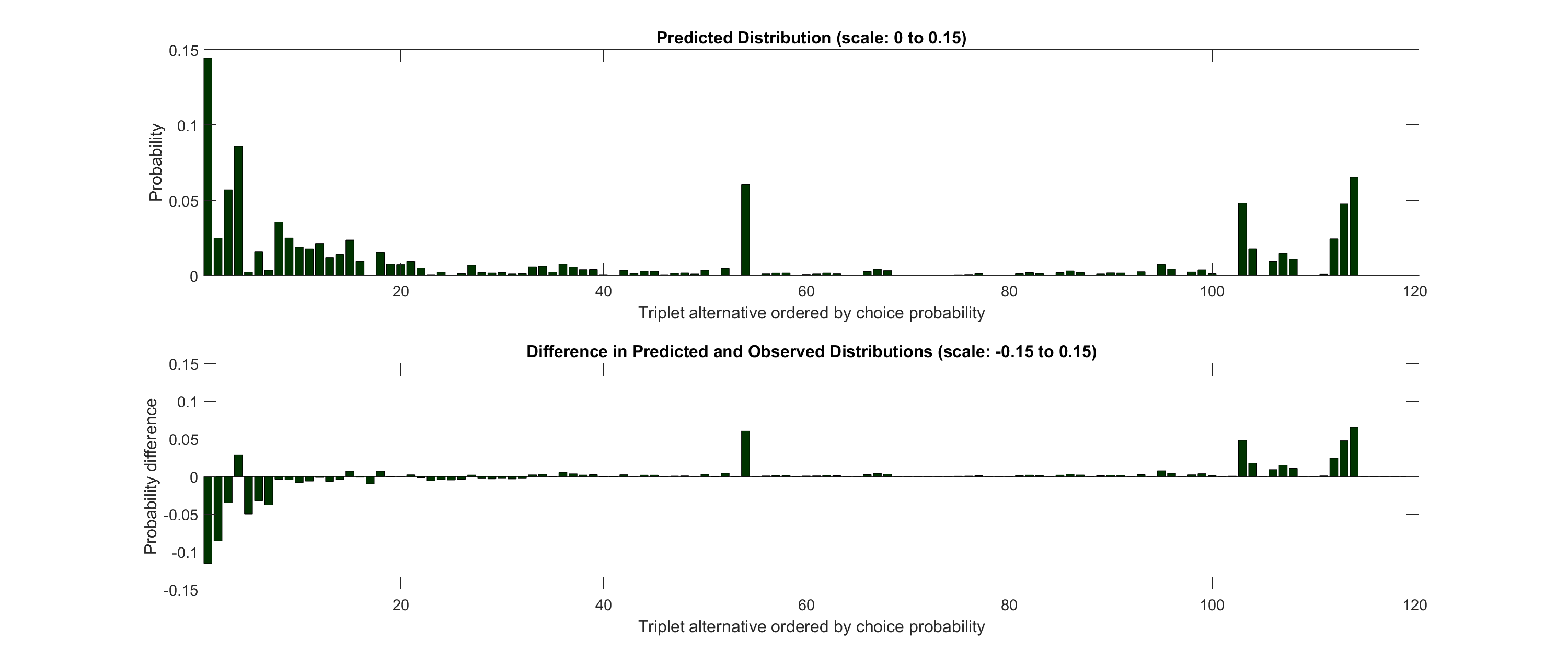}
    \captionsetup{font={footnotesize,sf},justification=raggedright}
    \caption*{Triplets are sorted by observed frequency at which they are chosen.  The first panel reports the predicted choice frequency and the second panel reports the difference in predicted and observed choice frequencies.}
\end{figure}

\begin{figure}[!htbp]\centering
	\caption{The $\modA$ Model, Three Coverages: \\Consideration and Optimal Choice Distribution}
	\label{F:ARCwideconsid}
	\includegraphics[trim={5cm 0 0 0cm},scale=0.4]{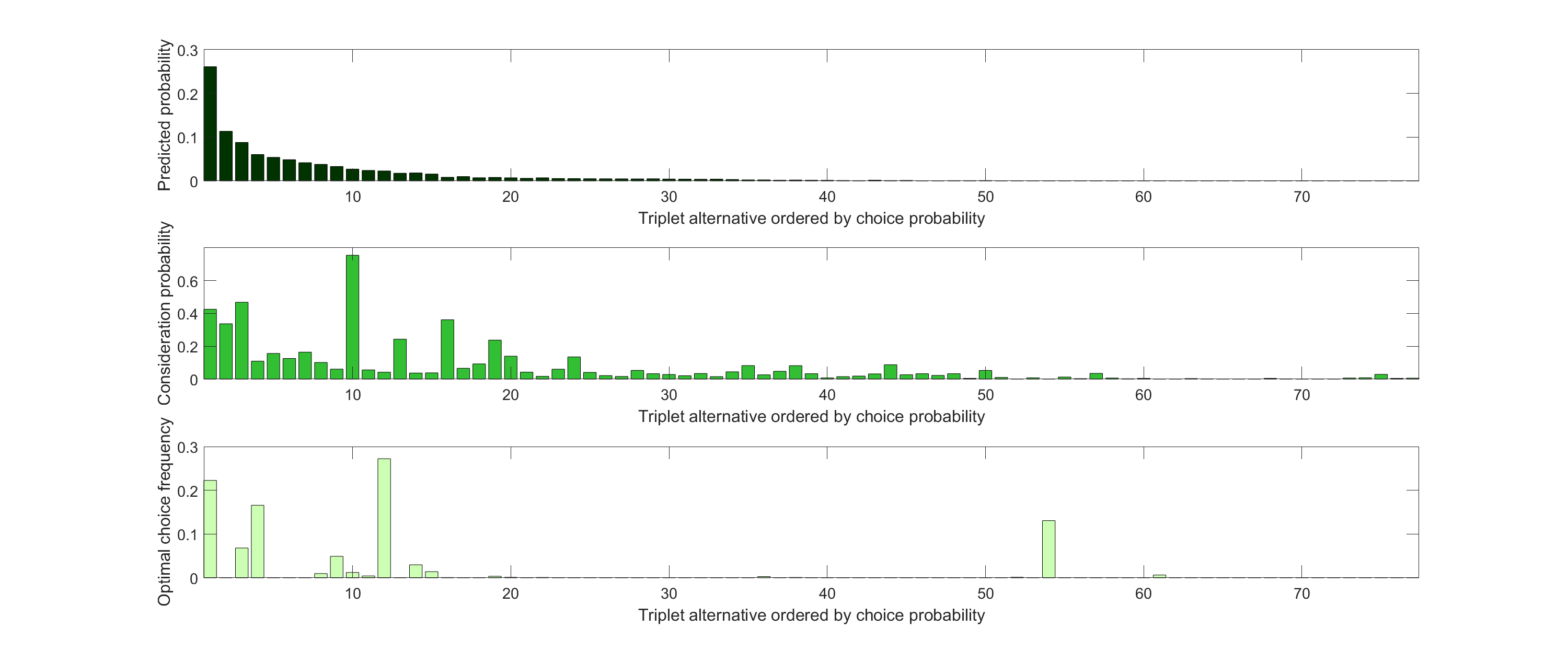}
    \captionsetup{font={footnotesize,sf},justification=raggedright}
    \caption*{Triplets are sorted by observed frequency at which they are chosen.}
\end{figure}

\begin{figure}[!htbp]\centering
	\caption{The $\modA$ Model with Three Coverages: \\ Monetary Loss From Limited Consideration}
	\label{F:WelfareWide}
	\includegraphics[trim={0cm 0 0 0cm},scale=0.3]{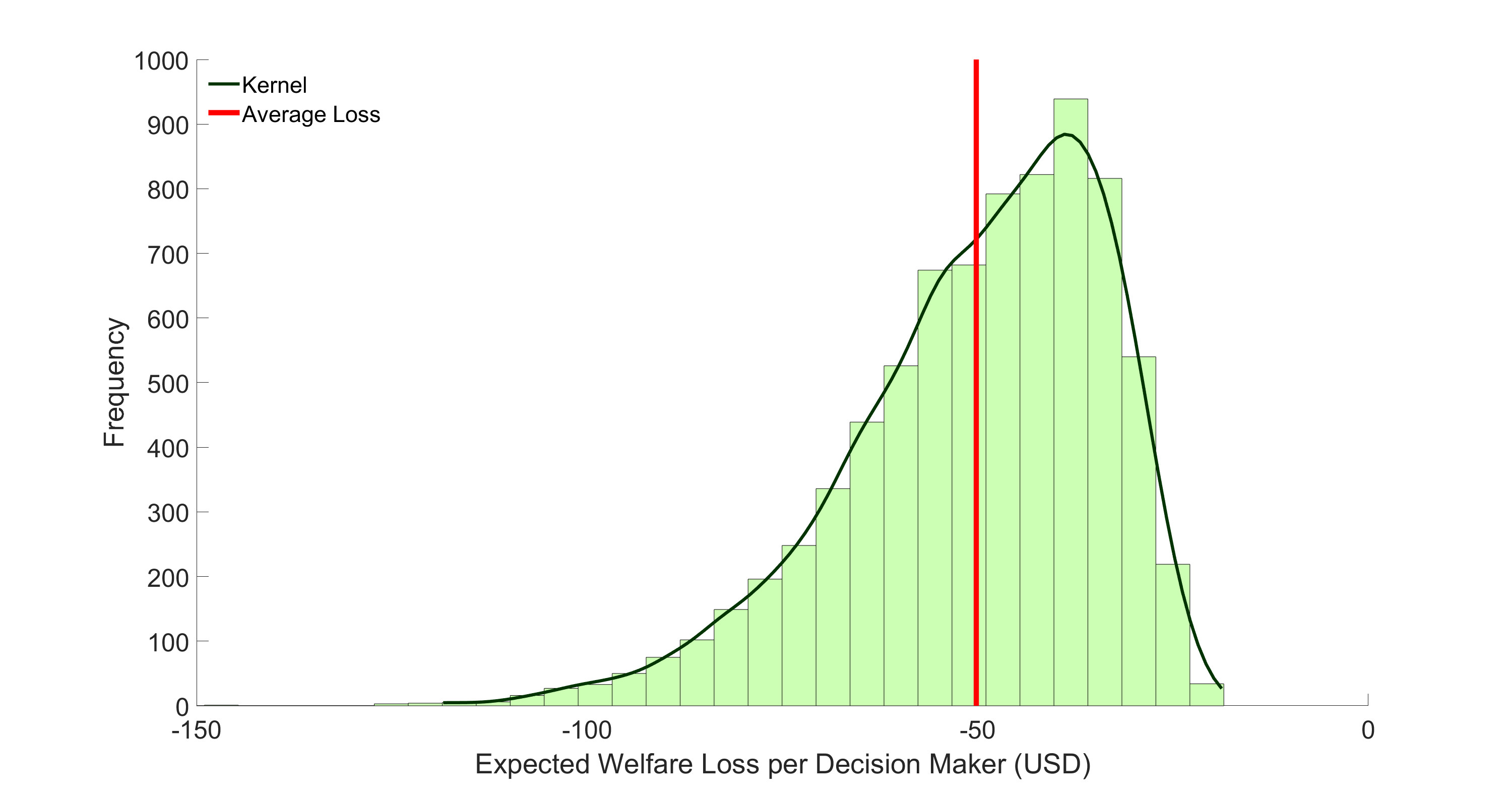}
\end{figure}

\newpage
\clearpage

\subsection{Tables}
\renewcommand\thetable{\thesection.\arabic{table}}
\setcounter{table}{0}
\begin{table}[!htbp]\centering
\normalsize
\begin{threeparttable}
\caption{MLE Estimation Results for the  $\modA$ Model: Collision Only}
\label{T:mleARCcoll}
\begin{tabular}{lllll}
\hline\hline
&\multicolumn{2}{c}{\textbf{ARC Model}}&\multicolumn{2}{c}{\textbf{ARC Model with Observables}}\\
\hline\\[-1.5ex]
\textbf{Average $\beta_{1i}$}&$1.70$&[1.56, 1.82]&$2.23$&[1.93, 2.50]\\\\[-1.5ex]
\textbf{$\beta_2$}&$7.45$&[6.68, 8.08]&$9.20$&[8.09, 10.1]\\\\[-1.5ex]
\textbf{Mean of $\nu$}&$0.0037$&[0.0036, 0.0038]&$0.0038$&[0.0036, 0.0040]\\\\[-1.5ex]
\textbf{SD of $\nu$}&$0.0024$&[0.0023, 0.0026]&$0.0022$&[0.0021, 0.0023]\\\\[-1.5ex]
\textbf{Intercept}&-&-&$-1.41$&[-1.47, -1.33]\\\\[-1.5ex]
\textbf{Age}&-&-&$0.207$&[0.173, 0.237]\\\\[-1.5ex]
\textbf{Age$^2$}&-&-&$0.048$&[0.022, 0.073]\\\\[-1.5ex]
\textbf{Female Driver}&-&-&$0.077$&[0.051, 0.104]\\\\[-1.5ex]
\textbf{Single Driver}&-&-&$0.050$&[0.022, 0.079]\\\\[-1.5ex]
\textbf{Married Driver}&-&-&$0.103$&[0.062, 0.144]\\\\[-1.5ex]
\textbf{Credit Score}&-&-&$0.134$&[0.107, 0.160]\\\\[-1.5ex]
\textbf{2+ Drivers}&-&-&$-0.302$&[-0.370, -0.224]\\\\[-1.5ex]
\textbf{Collision \$100}&$0.059$&[0.050, 0.068]&$0.050$&[0.042, 0.058]\\\\[-1.5ex]
\textbf{Collision \$200}&$0.412$&[0.391, 0.433]&$0.390$&[0.364, 0.413]\\\\[-1.5ex]
\textbf{Collision \$250}&$0.206$&[0.198, 0.214]&$0.204$&[0.193, 0.212]\\\\[-1.5ex]
\textbf{Collision \$500}&$0.920$&[0.913, 0.926]&$0.915$&[0.909, 0.924]\\\\[-1.5ex]
\textbf{Collision \$1000}&$1.000$&[1.000, 1.000]&$0.944$&[0.899, 1.000]\\\\[-1.5ex]
\hline\end{tabular}
\end{threeparttable}
\end{table}

\begin{table}[!htbp]\centering
\normalsize
\begin{threeparttable}
\caption{MLE Estimation Results for the  Proportionally Shifting Consideration Model}
\label{T:mleARCcollprop}
\begin{tabular}{lllll}
\hline\hline
&\multicolumn{2}{c}{\textbf{ARC Model}}&\multicolumn{2}{c}{\textbf{ARC Model with Observables}}\\
\hline\\[-1.5ex]
\textbf{Average $\beta_{1i}$}&$1.44$&[1.31, 1.55]&$2.11$&[1.86, 2.28]\\\\[-1.5ex]
\textbf{$\beta_2$}&$6.07$&[5.35, 6.67]&$8.74$&[7.73, 9.58]\\\\[-1.5ex]
\textbf{Mean of $\nu$}&$0.0038$&[0.0037, 0.0040]&$0.0038$&[0.0036, 0.0040]\\\\[-1.5ex]
\textbf{SD of $\nu$}&$0.0027$&[0.0026, 0.0028]&$0.0023$&[0.0021, 0.0024]\\\\[-1.5ex]
\textbf{Intercept}&-&-&$-1.40$&[-1.47, -1.35]\\\\[-1.5ex]
\textbf{Age}&-&-&$0.194$&[0.160, 0.222]\\\\[-1.5ex]
\textbf{Age$^2$}&-&-&$0.036$&[0.010, 0.059]\\\\[-1.5ex]
\textbf{Female Driver}&-&-&$0.070$&[0.046, 0.096]\\\\[-1.5ex]
\textbf{Single Driver}&-&-&$0.049$&[0.021, 0.076]\\\\[-1.5ex]
\textbf{Married Driver}&-&-&$0.091$&[0.047, 0.130]\\\\[-1.5ex]
\textbf{Credit Score}&-&-&$0.135$&[0.110, 0.160]\\\\[-1.5ex]
\textbf{2+ Drivers}&-&-&$-0.283$&[-0.348, -0.200]\\\\[-1.5ex]
\textbf{Collision \$100}&$0.061$&[0.051, 0.070]&$0.055$&[0.046, 0.063]\\\\[-1.5ex]
\textbf{Collision \$200}&$0.424$&[0.401, 0.446]&$0.408$&[0.382, 0.433]\\\\[-1.5ex]
\textbf{Collision \$250}&$0.211$&[0.202, 0.220]&$0.212$&[0.201, 0.222]\\\\[-1.5ex]
\textbf{Collision \$500}&$0.985$&[0.974, 0.998]&$0.961$&[0.929, 0.977]\\\\[-1.5ex]
\textbf{Collision \$1000}&$1.000$&-&$1.000$&-\\\\[-1.5ex]
\textbf{Average $\xi_{1i}$}&$0.478$&[0.277, 0.652]&$0.148$&[0.021, 0.212]\\\\[-1.5ex]
\textbf{$\xi_2$}&$26.7$&[16.1, 37.9]&$7.14$&[0.939, 10.3]\\\\[-1.5ex]
\textbf{$\xi_1$: Intercept}&-&-&$-2.24$&[-3.59, 0.0011]\\\\[-1.5ex]
\textbf{$\xi_1$: Age}&-&-&$1.24$&[-0.736, 1.75]\\\\[-1.5ex]
\textbf{$\xi_1$: Age$^2$}&-&-&$-0.382$&[-0.701, 0.584]\\\\[-1.5ex]
\textbf{$\xi_1$: Female Driver}&-&-&$-0.323$&[-1.16, 0.910]\\\\[-1.5ex]
\textbf{$\xi_1$: Single Driver}&-&-&$0.382$&[-1.51, 0.650]\\\\[-1.5ex]
\textbf{$\xi_1$: Married Driver}&-&-&$0.0017$&[-2.35, 1.45]\\\\[-1.5ex]
\textbf{$\xi_1$: Credit Score}&-&-&$0.405$&[-0.688, 0.642]\\\\[-1.5ex]
\textbf{$\xi_1$: 2+ Drivers}&-&-&$0.485$&[-2.22, 1.98]\\\\[-1.5ex]
\hline\end{tabular}
\end{threeparttable}
\end{table}

\begin{table}[!htbp]\centering
\normalsize
\begin{threeparttable}
\caption{MLE Estimation Results for the Mixed Logit: \\ Collision Only}
\label{T:mleRCLcoll}
\begin{tabular}{lll}
\hline\hline
&\multicolumn{2}{c}{\textbf{Mixed Logit}}\\
\hline\\[-1.5ex]
\textbf{Average $\beta_{1i}$}&$9.07$&[7.54, 10.2]\\\\[-1.5ex]
\textbf{$\beta_2$}&$124.4$&[106.0, 137.5]\\\\[-1.5ex]
\textbf{Mean of $\nu$}&$0.0014$&[0.0013, 0.0014]\\\\[-1.5ex]
\textbf{SD of $\nu$}&$0.0004$&[0.0004, 0.0005]\\\\[-1.5ex]
\textbf{Intercept}&$-2.59$&[-2.63, -2.55]\\\\[-1.5ex]
\textbf{Age}&$-0.139$&[-0.156, -0.122]\\\\[-1.5ex]
\textbf{Age$^2$}&$-0.024$&[-0.037, -0.010]\\\\[-1.5ex]
\textbf{Female Driver}&$-0.0035$&[-0.019, 0.012]\\\\[-1.5ex]
\textbf{Single Driver}&$-0.0098$&[-0.026, 0.0060]\\\\[-1.5ex]
\textbf{Married Driver}&$-0.030$&[-0.054, -0.0078]\\\\[-1.5ex]
\textbf{Credit Score}&$0.091$&[0.076, 0.105]\\\\[-1.5ex]
\textbf{2+ Drivers}&$-0.016$&[-0.061, 0.029]\\\\[-1.5ex]
\textbf{Sigma}&$0.039$&[0.037, 0.041]\\\\[-1.5ex]
\hline\end{tabular}
\end{threeparttable}
\end{table}

\begin{table}[!htbp]\centering
 \vspace*{-2em}
\scriptsize
\caption{MLE Estimation Results for the $\modA$ Model, Three Coverages}
\label{T:mleARCwide}
\begin{threeparttable}
\begin{minipage}{.5\linewidth}
    \begin{tabular}{lll}
\hline\hline
&\multicolumn{2}{c}{\textbf{ARC Model}}\\
\hline\\[-1.5ex]
\textbf{Average $\beta_{1i}$}&$4.70$&[3.89, 5.30]\\\\[-1.5ex]
\textbf{$\beta_2$}&$24.0$&[19.7, 27.2]\\\\[-1.5ex]
\textbf{Mean of $\nu$}&$0.0032$&[0.0032, 0.0033]\\\\[-1.5ex]
\textbf{SD of $\nu$}&$0.0013$&[0.0012, 0.0014]\\\\[-1.5ex]
\textbf{Intercept}&$-1.68$&[-1.72, -1.64]\\\\[-1.5ex]
\textbf{Age}&$0.162$&[0.146, 0.180]\\\\[-1.5ex]
\textbf{Age$^2$}&$0.041$&[0.026, 0.054]\\\\[-1.5ex]
\textbf{Female Driver}&$0.043$&[0.027, 0.061]\\\\[-1.5ex]
\textbf{Single Driver}&$0.010$&[-0.0075, 0.028]\\\\[-1.5ex]
\textbf{Married Driver}&$0.030$&[0.0054, 0.054]\\\\[-1.5ex]
\textbf{Credit Score}&$0.137$&[0.121, 0.153]\\\\[-1.5ex]
\textbf{2+ Drivers}&$-0.097$&[-0.141, -0.052]\\\\[-1.5ex]
\textbf{(100,50,250)}&$0.041$&[0.033, 0.049]\\\\[-1.5ex]
\textbf{(100,50,500)}&$0.015$&[0.0099, 0.021]\\\\[-1.5ex]
\textbf{(100,50,1000)}&$0.013$&[0.0048, 0.023]\\\\[-1.5ex]
\textbf{(100,100,100)}&$0.0023$&[0.0009, 0.0042]\\\\[-1.5ex]
\textbf{(100,100,250)}&$0.0077$&[0.0049, 0.010]\\\\[-1.5ex]
\textbf{(100,100,500)}&$0.0050$&[0.0027, 0.0078]\\\\[-1.5ex]
\textbf{(100,100,1000)}&$0.0047$&[0.0019, 0.0086]\\\\[-1.5ex]
\textbf{(100,200,250)}&$0.0005$&[0.0002, 0.0010]\\\\[-1.5ex]
\textbf{(100,200,500)}&$0.0008$&[0.0004, 0.0015]\\\\[-1.5ex]
\textbf{(100,200,1000)}&$0.0036$&[0.0013, 0.0065]\\\\[-1.5ex]
\textbf{(200,50,100)}&$0.011$&[0.0052, 0.016]\\\\[-1.5ex]
\textbf{(200,50,250)}&$0.066$&[0.057, 0.075]\\\\[-1.5ex]
\textbf{(200,50,500)}&$0.061$&[0.051, 0.071]\\\\[-1.5ex]
\textbf{(200,50,1000)}&$0.033$&[0.018, 0.048]\\\\[-1.5ex]
\textbf{(200,100,100)}&$0.0020$&[0.0008, 0.0037]\\\\[-1.5ex]
\textbf{(200,100,250)}&$0.021$&[0.017, 0.026]\\\\[-1.5ex]
\textbf{(200,100,500)}&$0.028$&[0.023, 0.034]\\\\[-1.5ex]
\textbf{(200,100,1000)}&$0.023$&[0.012, 0.033]\\\\[-1.5ex]
\textbf{(200,200,100)}&$0.0017$&[0.0007, 0.0032]\\\\[-1.5ex]
\textbf{(200,200,250)}&$0.157$&[0.147, 0.167]\\\\[-1.5ex]
\textbf{(200,200,500)}&$0.165$&[0.154, 0.176]\\\\[-1.5ex]
\textbf{(200,200,1000)}&$0.136$&[0.113, 0.158]\\\\[-1.5ex]
\textbf{(200,250,250)}&$0.0004$&[0.0002, 0.0006]\\\\[-1.5ex]
\textbf{(200,250,500)}&$0.0005$&[0.0002, 0.0009]\\\\[-1.5ex]
\textbf{(200,500,250)}&$0.0015$&[0.0008, 0.0023]\\\\[-1.5ex]
\textbf{(200,1000,1000)}&$0.0047$&[0.0017, 0.0085]\\\\[-1.5ex]
\textbf{(250,50,100)}&$0.0020$&[0.0008, 0.0037]\\\\[-1.5ex]
\textbf{(250,50,250)}&$0.021$&[0.017, 0.025]\\\\[-1.5ex]
\textbf{(250,50,500)}&$0.033$&[0.027, 0.039]\\\\[-1.5ex]
\textbf{(250,100,250)}&$0.017$&[0.015, 0.020]\\\\[-1.5ex]
\textbf{(250,100,500)}&$0.016$&[0.014, 0.020]\\\\[-1.5ex]
\textbf{(250,100,1000)}&$0.019$&[0.011, 0.026]\\\\[-1.5ex]
\textbf{(250,200,100)}&$0.0010$&[0.0003, 0.0019]\\\\[-1.5ex]
\hline\end{tabular}
\end{minipage}
\begin{minipage}{.5\linewidth}
 \begin{tabular}{lll}
\hline\hline
&\multicolumn{2}{c}{\textbf{ARC Model}}\\
\hline\\[-1.5ex]
\textbf{(250,200,250)}&$0.037$&[0.033, 0.041]\\\\[-1.5ex]
\textbf{(250,200,500)}&$0.056$&[0.051, 0.061]\\\\[-1.5ex]
\textbf{(250,200,1000)}&$0.045$&[0.035, 0.055]\\\\[-1.5ex]
\textbf{(250,250,100)}&$0.0011$&[0.0004, 0.0020]\\\\[-1.5ex]
\textbf{(250,250,250)}&$0.042$&[0.038, 0.046]\\\\[-1.5ex]
\textbf{(250,250,500)}&$0.061$&[0.056, 0.066]\\\\[-1.5ex]
\textbf{(250,250,1000)}&$0.026$&[0.019, 0.034]\\\\[-1.5ex]
\textbf{(250,500,500)}&$0.0007$&[0.0004, 0.0013]\\\\[-1.5ex]
\textbf{(500,50,250)}&$0.034$&[0.028, 0.042]\\\\[-1.5ex]
\textbf{(500,50,500)}&$0.053$&[0.044, 0.063]\\\\[-1.5ex]
\textbf{(500,50,1000)}&$0.033$&[0.018, 0.047]\\\\[-1.5ex]
\textbf{(500,100,250)}&$0.015$&[0.012, 0.019]\\\\[-1.5ex]
\textbf{(500,100,500)}&$0.043$&[0.036, 0.049]\\\\[-1.5ex]
\textbf{(500,100,1000)}&$0.048$&[0.034, 0.062]\\\\[-1.5ex]
\textbf{(500,200,100)}&$0.0079$&[0.0040, 0.012]\\\\[-1.5ex]
\textbf{(500,200,250)}&$0.126$&[0.119, 0.133]\\\\[-1.5ex]
\textbf{(500,200,500)}&$0.337$&[0.322, 0.352]\\\\[-1.5ex]
\textbf{(500,200,1000)}&$0.243$&[0.221, 0.264]\\\\[-1.5ex]
\textbf{(500,250,100)}&$0.0018$&[0.0008, 0.0032]\\\\[-1.5ex]
\textbf{(500,250,250)}&$0.038$&[0.034, 0.042]\\\\[-1.5ex]
\textbf{(500,250,500)}&$0.102$&[0.095, 0.109]\\\\[-1.5ex]
\textbf{(500,250,1000)}&$0.093$&[0.080, 0.106]\\\\[-1.5ex]
\textbf{(500,500,100)}&$0.0032$&[0.0015, 0.0059]\\\\[-1.5ex]
\textbf{(500,500,250)}&$0.110$&[0.103, 0.116]\\\\[-1.5ex]
\textbf{(500,500,500)}&$0.426$&[0.413, 0.439]\\\\[-1.5ex]
\textbf{(500,500,1000)}&$0.469$&[0.450, 0.486]\\\\[-1.5ex]
\textbf{(1000,50,250)}&$0.0069$&[0.0020, 0.013]\\\\[-1.5ex]
\textbf{(1000,50,500)}&$0.0085$&[0.0022, 0.015]\\\\[-1.5ex]
\textbf{(1000,50,1000)}&$0.029$&[0.0031, 0.053]\\\\[-1.5ex]
\textbf{(1000,100,250)}&$0.0049$&[0.0019, 0.0090]\\\\[-1.5ex]
\textbf{(1000,100,500)}&$0.0060$&[0.0021, 0.011]\\\\[-1.5ex]
\textbf{(1000,100,1000)}&$0.035$&[0.0080, 0.064]\\\\[-1.5ex]
\textbf{(1000,200,250)}&$0.032$&[0.019, 0.045]\\\\[-1.5ex]
\textbf{(1000,200,500)}&$0.082$&[0.061, 0.105]\\\\[-1.5ex]
\textbf{(1000,200,1000)}&$0.088$&[0.046, 0.128]\\\\[-1.5ex]
\textbf{(1000,250,250)}&$0.0067$&[0.0027, 0.012]\\\\[-1.5ex]
\textbf{(1000,250,500)}&$0.027$&[0.015, 0.039]\\\\[-1.5ex]
\textbf{(1000,250,1000)}&$0.053$&[0.025, 0.083]\\\\[-1.5ex]
\textbf{(1000,500,250)}&$0.033$&[0.022, 0.044]\\\\[-1.5ex]
\textbf{(1000,500,500)}&$0.140$&[0.119, 0.161]\\\\[-1.5ex]
\textbf{(1000,500,1000)}&$0.362$&[0.309, 0.405]\\\\[-1.5ex]
\textbf{(1000,1000,250)}&$0.082$&[0.058, 0.107]\\\\[-1.5ex]
\textbf{(1000,1000,500)}&$0.238$&[0.199, 0.267]\\\\[-1.5ex]
\textbf{(1000,1000,1000)}&$0.755$&[0.652, 0.829]\\\\[-1.5ex]
\\\\[-1.5ex]
\hline\end{tabular}
\end{minipage}
\end{threeparttable}
\end{table}

\begin{table}[!htbp]\centering
\normalsize
\begin{threeparttable}
\caption{MLE Estimation Results for RUM, Three Coverages}
\label{T:mleRUMbroad}
\begin{tabular}{lll}
\hline\hline
&\multicolumn{2}{c}{\textbf{Mixed Logit}}\\
\hline\\[-1.5ex]
\textbf{Average $\beta_{1i}$}&$4.89$&[4.60, 5.16]\\\\[-1.5ex]
\textbf{$\beta_2$}&$54.2$&[51.6, 56.6]\\\\[-1.5ex]
\textbf{Mean of $\nu$}&$0.0017$&[0.0016, 0.0017]\\\\[-1.5ex]
\textbf{SD of $\nu$}&$0.0007$&[0.0007, 0.0007]\\\\[-1.5ex]
\textbf{Intercept}&$-2.37$&[-2.39, -2.34]\\\\[-1.5ex]
\textbf{Age}&$-0.077$&[-0.088, -0.066]\\\\[-1.5ex]
\textbf{Age$^2$}&$-0.015$&[-0.024, -0.0059]\\\\[-1.5ex]
\textbf{Female Driver}&$0.0008$&[-0.0098, 0.012]\\\\[-1.5ex]
\textbf{Single Driver}&$-0.014$&[-0.025, -0.0030]\\\\[-1.5ex]
\textbf{Married Driver}&$-0.018$&[-0.033, -0.0029]\\\\[-1.5ex]
\textbf{Credit Score}&$0.034$&[0.023, 0.045]\\\\[-1.5ex]
\textbf{2+ Drivers}&$-0.048$&[-0.075, -0.020]\\\\[-1.5ex]
\textbf{Sigma}&$0.224$&[0.209, 0.238]\\\\[-1.5ex]
\hline\end{tabular}
\end{threeparttable}
\end{table}

\begin{table}[!htbp]\centering
\normalsize
\begin{threeparttable}
\caption{Average Monetary Loss by Group}
\label{T:ARCwelfare}
\begin{tabular}{lll}
\hline\hline
&\multicolumn{2}{c}{\textbf{Average Monetary Loss}}\\
\hline\\[-1.5ex]
\textbf{All}&$-50.2$&[-52.4, -47.3]\\\\[-1.5ex]
\textbf{Female Driver}&$-54.3$&[-56.8, -51.0]\\\\[-1.5ex]
\textbf{Single Driver}&$-45.1$&[-47.1, -42.3]\\\\[-1.5ex]
\textbf{Young}&$-45.5$&[-47.2, -43.2]\\\\[-1.5ex]
\textbf{Old}&$-65.4$&[-69.3, -59.4]\\\\[-1.5ex]
\textbf{Low Credit Driver}&$-47.6$&[-49.3, -44.9]\\\\[-1.5ex]
\textbf{High Credit Driver}&$-54.3$&[-57.3, -50.3]\\\\[-1.5ex]
\hline\end{tabular}
\end{threeparttable}
\end{table}

\end{appendices}

\end{document}